\documentclass[conference]{IEEEtran}
\IEEEoverridecommandlockouts
\usepackage{cite}
\usepackage{amsmath,amssymb,amsfonts}
\usepackage[ruled,vlined]{algorithm2e}
\usepackage[noend]{algpseudocode}
\usepackage{graphicx}
\usepackage{textcomp}
\usepackage{xcolor}
\usepackage[section]{placeins}
\usepackage{float}
\usepackage{epstopdf}
\usepackage{subfigure}
\usepackage{mathrsfs}
\usepackage{enumitem}
\usepackage{soul}
\usepackage{times}
\usepackage{booktabs} 
\usepackage{bbm}
\usepackage{graphicx}
\usepackage{epstopdf}
\usepackage{booktabs}
\usepackage{url}
\usepackage{multirow}
\usepackage{balance}

\newtheorem{defn}{Definition}
\newenvironment{proof}{\quad{\it Proof:}}{\hfill $\square$\par}

\newtheorem{Lemma}{Lemma}
\newtheorem{example}{Example}
\newtheorem{theo}{Theorem}

\newtheorem{obser}{Observation}

\def\BibTeX{{\rm B\kern-.05em{\sc i\kern-.025em b}\kern-.08em
    T\kern-.1667em\lower.7ex\hbox{E}\kern-.125emX}}

\newcommand{\kw}[1]{{\ensuremath {\mathsf{#1}}}\xspace}

\newcommand{\stitle}[1]{\vspace{0.5ex} \noindent{\bf #1}}
\long\def\comment#1{}

\newcommand{\etal}{\emph{et~al.}\xspace}

\newcommand\figref[1]{Fig.~\ref{#1}}
\newcommand\lamref[1]{Lemma~\ref{#1}}

\newcommand\tabref[1]{Table~\ref{#1}}

\newcommand\secref[1]{Section~\ref{#1}}

\newcommand{\baseboundalg}{\kw{BaseBSearch}}
\newcommand{\baseupdateprc}{\kw{UptSMap}}
\newcommand{\bound}{\kw{\overline {ub}}}
\newcommand{\optbound}{\kw{\widetilde {ub}}}
\newcommand{\topbound}{\kw{\widehat {tb}}}

\newcommand{\optboundalg}{\kw{OptBSearch}}
\newcommand{\egobwcalalg}{\kw{EgoBWCal}}

\newcommand{\vertexparaalg}{\kw{VertexPEBW}}
\newcommand{\vparainneralg}{\kw{PEgoBWCal}}
\newcommand{\edgeparaalg}{\kw{EdgePEBW}}

\newcommand{\insertedge}{\kw{LocalInsert}}
\newcommand{\deleteedge}{\kw{LocalDelete}}
\newcommand{\uptlocalsmap}{\kw{LocalUptSMap}}

\newcommand{\insertedgetop}{\kw{LazyInsert}}
\newcommand{\deleteedgetop}{\kw{LazyDelete}}

\newcommand{\pokec}{\kw{Pokec}}
\newcommand{\dblp}{\kw{DBLP}}
\newcommand{\youtube}{\kw{Youtube}}
\newcommand{\livejournal}{\kw{LiveJournal}}
\newcommand{\wikitalk}{\kw{WikiTalk}}
\newcommand{\db}{\kw{DB}}
\newcommand{\ir}{\kw{IR}}

\newcommand{\astroph}{\kw{Astroph}}
\newcommand{\enron}{\kw{Enron}}
\newcommand{\epinions}{\kw{Epinions}}
\newcommand{\euall}{\kw{Euall}}
\newcommand{\flickr}{\kw{Flickr}}
\newcommand{\gowalla}{\kw{Gowalla}}
\newcommand{\hepph}{\kw{Hepph}}
\newcommand{\notredame}{\kw{Notredame}}
\newcommand{\slashdot}{\kw{Slashdot}}
\newcommand{\wikivote}{\kw{Wikivote}}

\newcommand{\bbsalg}{\kw{BaseBS}}
\newcommand{\obsalg}{\kw{OptBS}}

\newcommand{\btalg}{\kw{TopBW}}
\newcommand{\ebtalg}{\kw{TopEBW}}

\begin{document}

\title{Efficient Top-$k$ Ego-Betweenness Search\\
}
\author{{Qi Zhang$\scriptsize^{\dag}$, Rong-Hua Li$\scriptsize^{\dag}$, Minjia Pan$\scriptsize^{\dag}$, Yongheng Dai$\scriptsize^{\ddag}$, Guoren Wang$\scriptsize^{\dag}$, Ye Yuan$\scriptsize^{\dag}$}
	\vspace{1.6mm}\\
	\fontsize{9}{9}\selectfont\itshape
	$\scriptsize^{\dag}$Beijing Institute of Technology, Beijing, China; $\scriptsize^{\ddag}$Diankeyun Technologies Co. , Ltd.; \\
	\fontsize{8}{8}\selectfont\ttfamily\upshape
	qizhangcs@bit.edu.cn; lironghuabit@126.com; panminjia\_cs@163.com; \\toyhdai@163.com; wanggrbit@126.com; yuan-ye@bit.edu.cn
}

\maketitle

\begin{abstract}
Betweenness centrality, measured by the number of times a vertex occurs on all shortest paths of a graph, has been recognized as a key indicator for the importance of a vertex in the network. However, the betweenness of a vertex is often very hard to compute because it needs to explore all the shortest paths between the other vertices. Recently, a relaxed concept called ego-betweenness was introduced which focuses on computing the betweenness of a vertex in its ego network. In this work, we study a problem of finding the top-$k$ vertices with the highest ego-betweennesses. We first develop two novel search algorithms equipped with a basic upper bound and a dynamic upper bound to efficiently solve this problem. Then, we propose local-update and lazy-update solutions to maintain the ego-betweennesses for all vertices and the top-$k$ results when the graph is updated, respectively. In addition, we also present two efficient parallel algorithms to further improve the efficiency. The results of extensive experiments on five large real-life datasets demonstrate the efficiency, scalability, and effectiveness of our algorithms.
\end{abstract}

\section{Introduction} \label{sec:introduction}

Betweenness centrality is a fundamental metric in network analysis~\cite{betweencent1977freeman,eigencent2008newman}. The betweenness centrality of a vertex $v$ is the sum of the ratio of the shortest paths that pass through $v$ between other vertices in a graph. Such a centrality metric has been successfully used in a variety of network analysis applications, such as social network analysis~\cite{DBLP:conf/semco/Ostrowski15}, biological network analysis~\cite{jeong2001lethalitybio}, communication network analysis~\cite{baldesi2017usetrans} and so on. More specifically, in social networks, a vertex with a high betweenness centrality is plausibly an influential user who can decide whether to share information or not~\cite{DBLP:conf/semco/Ostrowski15}. In protein interaction networks, the high-betweenness proteins represent important connectors that link some modular organizations~\cite{jeong2001lethalitybio}. In communication networks, the nodes with higher betweennesses might have more control over the network, thus attacking these nodes may cause severe damage to the network~\cite{baldesi2017usetrans}.

Although betweenness centrality plays a critical role in network analysis, computing betweenness scores for all vertices is notoriously expensive because it requires exploring the shortest paths between all vertices in a graph. The state-of-the-art algorithm for betweenness computation is the Brandes' algorithm \cite{brandes2001faster} which takes ${\mathcal O}(nm)$ time. Such a time complexity is acceptable only in small graphs with a few tens of thousands of vertices and edges, but it is prohibitively expensive on modern networks with millions of vertices and tens of millions of edges.



To avoid the high computational cost problem, Everett \etal \cite{everett2005ego} introduced a relaxed concept called ego-betweenness centrality which focuses on computing a vertex's betweenness in its ego network, where the ego network of a vertex $u$ is the subgraph induced by $u$ and $u$'s neighbors. More specifically, the ego-betweenness of a vertex $u$ is measured by the sum of the ratio of the shortest paths that pass through $u$ between $u$'s neighbors in the ego network. Everett \etal showed that the ego-betweenness centrality is highly correlated with the traditional betweenness centrality in networks, thus it can be considered as a good approximation of the traditional betweenness. Hence, like betweenness centrality, ego-betweenness can measure the importance of a node as a ``link'' between different parts of the graph. For vertices $v$ and $w$, if $u$ is the only vertex that connects $v$ and $w$ in $u$'s ego network, then $u$ is important to control the information flow between $v$ and $w$. On the other hand, there are alternative vertices to connect the two vertices and $u$ can be easily bypassed. A vertex with a high ego-betweenness indicates that it has higher control over its ego network and is not easily replaced by other vertices, thus it plays an important role in the graph. Moreover, real-life applications often require retrieving the top-$k$ vertices with the highest ego-betweenness scores, rather than the exact ego-betweenness scores for all vertices. Motivated by this, we in this paper study the problem of identifying the top-$k$ vertices in a graph with the highest ego-betweennesses.


To solve the top-$k$ ego-betweenness search problem, a straightforward algorithm is to calculate the ego-betweennesses for all vertices and then select the top-$k$ results. However, such a straightforward algorithm is very costly for large graphs, because the total cost for constructing the ego network for each vertex is very expensive in large graphs. To efficiently compute the top-$k$ vertices, the general idea of top-$k$ search frameworks \cite{15vldbjstrucdiv, chang2017scalable, zhang2020efficient} can be used, which explores the vertices based on a predefined ordering and then applies some upper-bounding rules to prune the unpromising vertices. Inspired by these algorithms, we first derive a basic upper bound and a dynamic upper bound of ego-betweenness. Then, we develop two top-$k$ search algorithms with those bounds to efficiently solve the top-$k$ ego-betweenness search problem. To handle dynamic graphs, we present local-update solutions to maintain ego-betweennesses for all vertices, and also develop lazy-update techniques to maintain the top-$k$ results. Additionally, we propose two efficient parallel algorithms to improve the efficiency of ego-betweenness computation. In summary, we make the following contributions.

{\kw{Top}-$k$~\kw{search~algorithms}.} We develop a basic algorithm with a static upper bound and an improved algorithm with a tighter and dynamically-updating upper bound to find the top-$k$ vertices with the highest ego-betweennesses. Both the algorithms consume $O(\alpha m d_{\max})$ time using $O( md_{\max})$ space in the worst case. Here $\alpha$ is the arboricity of the graph \cite{chiba1985arboricity} which is typically very small in real-life graphs \cite{12tcsarboricity}. We show that both algorithms can significantly prune the vertices that are definitely not contained in the top-$k$ results. Moreover, the improved algorithm can achieve more effective pruning performance due to the tighter and dynamically-updating upper bound.

{\kw{Ego}-\kw{betweenness~maintenance~and~parallel~algorithms}.} We develop local-update algorithms to maintain the ego-betweennesses for all vertices when the graph is updated. We also propose lazy-update techniques to maintain the top-$k$ results for dynamic graphs. To further improve the efficiency, we present two efficient parallel algorithms to compute all vertices' ego-betweennesses. Compared with the sequential algorithms, our parallel solutions can achieve a high degree of parallelism, thus improving the efficiency of ego-betweenness computation significantly.


{\kw{Extensive~experiments}.} We conduct comprehensive experimental studies to evaluate the proposed algorithms using five large real-world datasets. The results show that 1) our improved algorithm with a dynamic upper bound is roughly 3-23 times faster than the basic algorithm; 2) our maintenance algorithms can maintain the top-$k$ results in less than 0.3 seconds in a large graph with 3,997,962 vertices and 34,681,189 edges; 3) our best parallel algorithm can achieve near 16 speedup ratio when using 16 threads; 4) the top-$k$ results of ego-betweenness are highly similar to the top-$k$ results of traditional betweenness. Thus, our results indicate that the ego-betweenness metric can be seen as a very good approximation of the traditional betweenness metric, but it is much cheaper to compute by utilizing the proposed algorithms.

{\kw{Reproducibility}.} For reproducibility, the source code of this paper is released at github: \url{https://github.com/QiZhang1996/egobetweenness}.


\stitle{Organization.} We introduce some important notations and formulate our problem in \secref{sec:preliminaries}. \secref{sec:onlinealg} presents the top-$k$ search algorithms. The ego-betweenness maintenance algorithms are developed in \secref{sec:onlineupdate}. We propose two parallel algorithms to speed up the ego-betweenness computation in \secref{sec:parallelalg}. \secref{sec:experiments} reports the experimental results. We survey related studies in \secref{sec:relatedwork} and conclude this work in \secref{sec:conclusion}.

\section{Preliminaries} \label{sec:preliminaries}

Let $G = (V, E)$ be an undirected and unweighted graph with $n = | V |$ vertices and $m = | E |$ edges. We denote the set of neighbors of a vertex $u$ by $N(u)$, i.e., $N(u) = \{ v \in V | (u, v) \in E\}$, and the degree of $u$ by $d(u) = |{N(u)}|$. Similarly, the neighbors of an edge $(u, v)$, denoted by $N(u, v)$, are the vertices that are adjacent to both $u$ and $v$, i.e., $N(u, v) = \{ w \in V|(u, w) \in E, (v, w) \in E\}$. For a subset $S \subseteq V$, the subgraph of $G$ induced by $S$ is defined as ${G_S} = ({V_S}, {E_S})$ where ${V_S} = S$ and ${E_S} = \{(u, v) | u, v \in S, (u, v) \in E\}$.

We define a total order $\prec$ on $V$ as follows. For vertices $u$ and $v$ in $V$, we say $u \prec v$, if and only if 1) $d(u) > d(v)$ or 2) $d(u) = d(v)$ and $u$ has a larger ID than $v$. Based on such a degree ordering $\prec$, we can construct a directed graph $G^+$ from $G$ by orientating each undirected edge $(u, v) \in G$ to respect the total order $u \prec v$. We denote the out-neighborhood of $u$ in $G^+$ as $N^+(u) = \{v \in V | (u, v) \in E^+\}$.

We give an essential concept, called \emph{ego network}, as follows.

\begin{defn} \label{def:ego-net}
\kw{(Ego}~\kw{network)} For vertex $p$ in $G = (V, E)$, the \emph{ego network} of $p$, denoted by $G_{E(p)}$, is a subgraph of $G$ induced by the vertex set $N(p) \cup \{ p\} $.
\end{defn}

Given a graph $G = (V, E)$ and a vertex $p \in V$. We use $\bar{S}_{E(p)}$ to denote the edges between the neighbors of $p$, i.e., $\bar{S}_{E(p)} = \left\{(u, v) | u, v \in N(p), (u, v) \in E \right\}$. For vertices $u, v \in N(p)$ and $(u, v) \notin E$, we suppose that $u \prec v$. Let $\hat{S}_p(u, v)$, which does not include $p$, be the set of vertices that connect $u$ and $v$ in $G_{E(p)}$, i.e., $\hat{S}_p(u, v) = \left\{w | u, v, w \in N(p), (u, v) \notin E, (u, w) \in E, (v, w) \in E \right\}$. If there is only one vertex $p$ that links $u$ and $v$ in $G_{E(p)}$, we add the pair $(u, v)$ into the set $\ddot{S}_{E(p)}$, i.e., $\ddot{S}_{E(p)} = \{(u, v) | u, v, w \in N(p), (u, v) \notin E, (u, w) \notin E || (v, w) \notin E\}$. Denote by $\hat{S}_{E(p)}$ the set of all $\hat{S}_p(u, v)$s. We use $\bar{C}_p$ to represent the size of $\bar{S}_{E(p)}$, i.e., $\bar{C}_p = |\bar{S}_{E(p)}|$. Similarly, we denote $\hat{C}_p = |\hat{S}_{E(p)}|$ and $\ddot{C}_p = |\ddot{S}_{E(p)}|$.

\begin{figure}[t!]\vspace*{-0.2cm}
\centering
  \subfigure[$G$]{
  \label{fig:expgraph}
  \begin{minipage}{4.5cm}
  \centering
  \includegraphics[width=\textwidth]{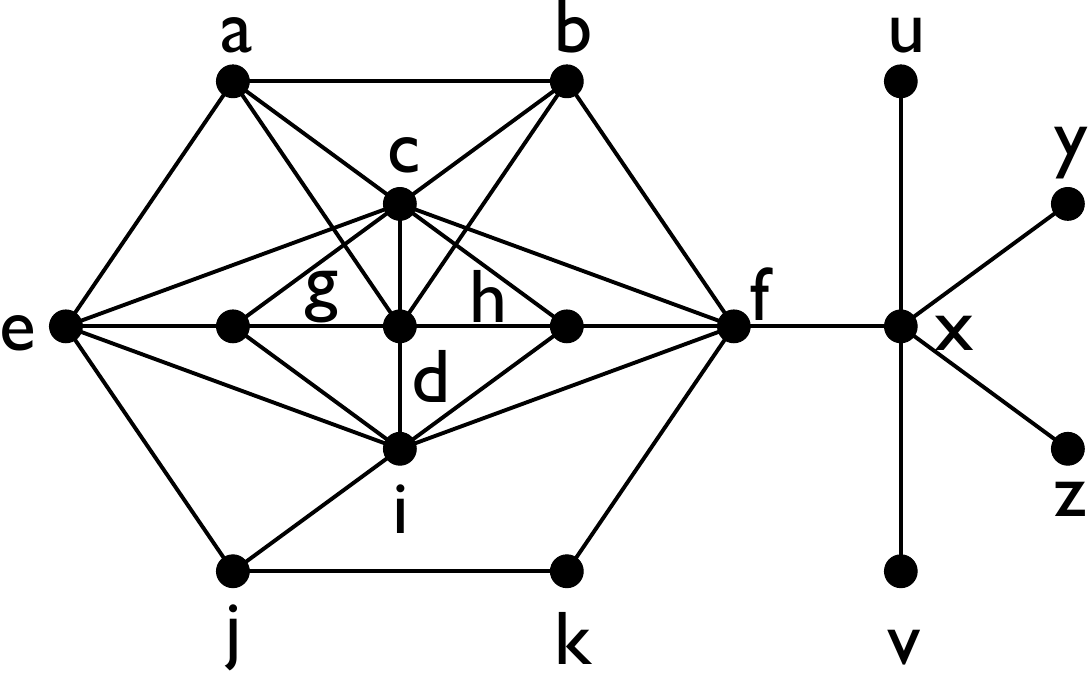}
  \end{minipage}
  }
  \subfigure[$G_{E(d)}$]{
  \label{fig:expego}
  \begin{minipage}{1.55cm}
  \centering
  \includegraphics[width=\textwidth]{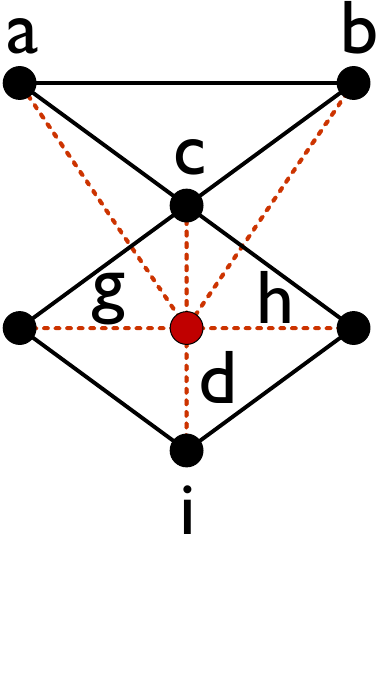}
  \end{minipage}
  }
\vspace*{-0.2cm}
\caption{Running example}
\label{fig:runexpgraph}
\vspace*{-0.2cm}
\end{figure}

\comment{
\begin{example}
Consider a graph $G$ in \figref{fig:expgraph}. Take the vertex $d$ as an example, the set of its neighbors $N(d)$ is $ \left\{ {a, b, c, g, h, i} \right\}$. \figref{fig:expego} shows the ego network $G_{E(d)}$ of $d$. Clearly, the edges colored black in \figref{fig:expego} belong to $\bar{S}_{E(d)}$, i.e., $\bar{S}_{E(d)} = \{(a, b), (a, c), (b, c), (c, g), (c, h),(g, i), (h, i)\}$. For vertex $c$ and vertex $i$ ($c \prec i$), there is no edge between $c$ and $i$, but $c$ can reach $i$ through vertices $d, g$ and $h$, thus $\hat{S}_d(c, i) = \{g, h\}$. We can easily check that $\hat{S}_{E(d)} = \{\hat{S}_d(a, g), \hat{S}_d(a, h), \hat{S}_d(b, g), \hat{S}_d(b, h), \hat{S}_d(c, i), \hat{S}_d(g, h)\}$ with $\hat{S}_d(a, g) = \hat{S}_d(a, h) = \hat{S}_d(b, g) = \hat{S}_d(b, h) = \{c\}$, $\hat{S}_d(c, i) = \{g, h\}$ and $\hat{S}_d(g, h) = \{c, i\}$. For vertices $a$ and $i$ ($i \prec a$), $d$ is the only one vertex that makes $a$ and $i$ reachable, thus the pair $(i, a)$ is added into $\ddot{S}_{E(d)}$. Analogously, we have $\ddot{S}_{E(p)} = \{(i, a), (i, b)\}$.
\end{example}
}

Given a vertex $p$ in $G$ and its ego network $G_{E(p)}$, for $u, v \in N(p)$, let $g_{uv}$ be the number of the shortest paths connecting $u$ and $v$ in $G_{E(p)}$ and $g_{uv}(p)$ be the number of those shortest paths that contain vertex $p$. Note that in $G_{E(p)}$, $g_{uv}(p)$ is either $0$ or $1$. Denote by $b_{uv}(p) = g_{uv}(p)/g_{uv}$ the probability that a randomly selected shortest path connecting $u$ with $v$ contains $p$ in $G_{E(p)}$. Based on the above notions, the definition of \emph{ego-betweenness} is following.


\begin{defn} \label{def:ego-betweenness}
\kw{(Ego}-\kw{betweenness)} For a vertex $p$ in $G$ , the ego-betweenness of $p$, denoted by ${C_B(p)}$, is defined as $C_B(p) = \sum_{u \prec v} b_{uv}(p), u, v \in {N(p)}$.
\end{defn}

\begin{example}
Consider a graph $G$ in \figref{fig:expgraph} and a vertex $d \in G$ with the ego network $G_{E(d)}$ illustrated in \figref{fig:expego}. For vertices $c$ and $i$, there are three shortest paths connecting $c$ and $i$ in $G_{E(d)}$, namely, $c \rightarrow g \rightarrow i$, $c \rightarrow h \rightarrow i$, and $c \rightarrow d \rightarrow i$, thus $g_{ci} = 3$ and $g_{ci}(d) = 1$, further $b_{ci}(d) = g_{ci}(d)/g_{ci} = 1/3$ holds. Analogously, we have $b_{hg}(d) = 1/3$, $b_{ga}(d) = b_{gb}(d) = b_{ha}(d) = b_{hb}(d) = 1/2$, $b_{ia}(d) = b_{ib}(d) = 1$ and the probabilities for other vertex pairs in $G_{E(d)}$ are equal to 0, thus ${C_B(d)} = 14/3$.
\end{example}

\stitle{Problem definition.} Given a graph $G$ and an integer $k$, the top-$k$ ego-betweenness search problem is to identify the $k$ vertices in $G$ with the highest ego-betweenness scores.

The following example illustrates the definition of our problem.

\begin{example}
Reconsider the graph $G$ shown in \figref{fig:expgraph}. Based on the definition of ego-betweenness, we can easily derive the ego-betweennesses of all vertices. For instance, we have $C_B(f) = 11$, $C_B(x) = 10$, and $C_B(i) = 8$. Suppose that $k = 1$, $f$ is the answer because it has the highest ego-betweenness score among all vertices. When $k = 3$, the answers are $f$, $x$ and $i$. This is because there is no other vertex with the ego-betweenness greater than $8$ in $G$.
\end{example}

In addition, real-world networks undergo dynamically updates. To this end, we also investigate the problem of top-$k$ ego-betweenness maintenance when the graph is updated.

\stitle{Challenges.} To solve the top-$k$ ego-betweenness search problem, a straightforward algorithm is to compute the ego-betweenness for each vertex, and then pick the top-$k$ vertices as the answers. Such an approach, however, is costly for large graphs. This is because the algorithm needs to explore the ego network $G_{E(p)}$ to compute the ego-betweenness for each vertex $p$. The total size of all ego networks could be very large, thus the straightforward algorithm might be very expensive for large graphs. Since we are only interested in the top-$k$ results, we do not need to compute all vertices’ ego-betweenness scores exactly. The challenges of the problem are: 1) how to efficiently prune the vertices that are definitely not contained in the top-$k$ results; 2) how to efficiently compute the ego-betweenness for each vertex; 3) how to maintain the top-$k$ vertices with the highest ego-betweennesses in dynamic networks. To tackle these challenges, we will develop two new online search algorithms with two non-trivial punning techniques to efficiently search the top-$k$ ego-betweenness vertices. Then, We also design local update techniques and lazy update techniques to handle frequent updates and maintain the top-$k$ results.

\section{Top-$k$ ego-betweenness search} \label{sec:onlinealg}
In this section, we first present a top-$k$ ego-betweenness search algorithm, called \baseboundalg, which is equipped with an upper-bounding strategy to prune the search space. Then, to further improve the efficiency, we propose the \optboundalg algorithm with a dynamic upper bound which is tighter than that of \baseboundalg. 

\subsection{The \baseboundalg algorithm} \label{subsec:onlineICDE}
Before introducing the \baseboundalg algorithm, we first give some useful lemmas which lead to an upper bound of ego-betweenness for pruning search space in \baseboundalg.

\begin{Lemma}
\label{lem:baselemma1}
For any vertex $p$ in $G$, we have $\bar{C}_p + \hat{C}_p + \ddot{C}_p = \frac{d(p)*(d(p)-1)}{2}$.
\end{Lemma}

\begin{proof}
Clearly, the vertex pairs between vertex $p$'s neighbors are divided into three categories, namely, $\bar{S}_{E(p)}$, $\hat{S}_{E(p)}$ and $\ddot{S}_{E(p)}$. Therefore, the sum of $\bar{C}_p$, $\hat{C}_p$ and $\ddot{C}_p$ is the number of all vertex pairs between $N(p)$, i.e., $\bar{C}_p + \hat{C}_p + \ddot{C}_p = \frac{d(p)*(d(p)-1)}{2}$.
\end{proof}


\begin{Lemma}
\label{lem:baselemma2}
For any vertex $p$ in $G$, $C_B(p)=\frac{d(p)*(d(p)-1)}{2} - \bar{C}_p - \hat{C}_p + \sum_{(u, v)} \frac{1}{|\hat{S}_p(u, v)|+1} \le \bound(p) = \frac{d(p)*(d(p)-1)}{2}$ holds.
\end{Lemma}

\begin{proof}
Based on Definition \ref{def:ego-betweenness}, $C_B(p)$ is closely related to the number of shortest paths between $u$ and $v$ in $G_{E(p)}$. First, for each $(u, v) \in \ddot{S}_{E(p)}$, there is only one vertex $p$ that can link $u$ and $v$, so $b_{uv}(p)$ is equal to $1$. Thus, $\ddot{C}_p$ is a part of $C_B(p)$ which equals $\frac{d(p)*(d(p)-1)}{2} - \bar{C}_p - \hat{C}_p$ according to \lamref{lem:baselemma1}. Second, for every vertex pair $(u, v) \in \hat{S}_{E(p)}$, $\hat{S}_p(u, v)$ is the set of vertices connecting $u$ with $v$ in $G_{E(p)}$ but does not include $p$, thus the probability $b_{uv}(p)$ is equal to $\frac{1}{|\hat{S}_p(u, v)|+1}$. To sum up, $C_B(p) = \frac{d(p)*(d(p)-1)}{2} - \bar{C}_p - \hat{C}_p + \sum_{(u, v)} \frac{1}{|\hat{S}_p(u, v)|+1}$. As $\hat{C}_p = |\hat{S}_{E(p)}| = \sum_{(u, v) \in \hat{S}_p(u, v)} 1 \ge \sum_{(u, v)} \frac{1}{|\hat{S}_p(u, v)|+1}$, thus $C_B(p) \le \bound(p)$ holds.
\end{proof}

\comment{
\begin{Lemma}
\label{lem:baselemma3}
For vertex $p$ in $G$, $C_B(p) \le \lfloor \frac{d(p)*(d(p)-1)}{2} \rfloor$ holds.
\end{Lemma}

\begin{proof}
The proof is intuitive according to \lamref{lem:baselemma2}, thus is omitted.
\end{proof}
}

\begin{algorithm}[t]
  \scriptsize
  \caption{\baseboundalg$(G, k)$}
  \label{alg:bboundalg}
  \KwIn{$G = (V, E)$, an integer $k \ge 1$.}
  \KwOut{The top-$k$ vertex set $R$.}
  \For{$u \in V$}
  {
    $\bound(u) \leftarrow \frac{d(u)*(d(u)-1)}{2}$; $C_{B}(u) \leftarrow \bound(u)$\;
  }
  $R \leftarrow \emptyset$\;
  Construct the oriented graph $G^+ = (V, E^+)$ of $G$\;
  Initialize an array $B$ with $B(i) = false, 0 \le i < n$\;
  \For{$u \in V$ according to the total order}
  {
    {\bf {if}} $|R| = k$ {\it{and}} $\min_{v \in R}C_{B}(v) \ge \bound(u)$ {\bf {then}} {\bf break}\;
    {\bf {for}} $v \in N^+(u)$ {\bf {do}} $B(v) \leftarrow true$\;
    \For{$v \in N^+(u)$}{
        \For{$w\in N^+(v)$}
    {
        \If{$B(w) = true$}
		{
            $\baseupdateprc(S_{u}, v, w)$; $\baseupdateprc(S_{v}, u, w)$\;
            {\bf {if}} $\nexists S_w(u, v)$ {\bf {then}} $S_{w}.{\kw{insert}}((u, v),0)$\;
	    }
    }

    }
        {\bf {for}} $v \in N^+(u)$ {\bf {do}} $B(v) \leftarrow false$\;
	\For{$((i, j), val) \in S_u$}
    {
        $C_{B}(u) \leftarrow C_{B}(u) - 1$\;
        {\bf {if}} $val \ne 0$ {\bf {then}} $C_{B}(u) \leftarrow C_{B}(u) + \frac{1}{val+1}$\;
    }
	Update $R$ based on $u$ and $C_B(u)$\;
  }
  {\bf return} $R$\;

  \vspace*{0.1cm}
  {\bf Procedure} $\baseupdateprc(S_u, v, w)$\\
  \For{$x \in N(u)$}
  {
    {\bf {if}} $(x, v) \in E$ {\it{and}} $\nexists S_u(x, v)$ {\bf {then}} $S_{u}.{\kw{insert}}((x, v),0)$\;
    {\bf {if}} $(x, w) \in E$ {\it{and}} $\nexists S_u(x, w)$ {\bf {then}} $S_{u}.{\kw{insert}}((x, w),0)$\;
    \If{$(x, v) \in E$ and $(x, w) \notin E$}
    {
        {\bf {if}} $\nexists S_{u}(x, w)$ {\bf {then}} $S_{u}.{\kw{insert}}((x, w), 1)$\;
        {\bf {else if}} $S_{u}(x, w).val \ne 0$ {\bf {then}} $S_{u}(x, w).val$++\;
    }
    \If{$(x, v) \notin E$ and $(x, w) \in E$}
    {
        Update $S_u$ and $S_w$ as lines 25-26\;
    }
  }
\end{algorithm}

Equipped with \lamref{lem:baselemma2}, we present a basic search approach, called \baseboundalg, which computes the vertices' ego-betweennesses in non-increasing order of their upper bounds. The main idea of \baseboundalg is that a vertex with a large upper bound may have a high chance contained in the top-$k$ results. Based on this idea, the exact computations for the vertices with small upper bounds will be postponed or even avoided, thus \baseboundalg can significantly improve the efficiency compared with the algorithm calculating all ego-betweennesses.

The pseudo-code of \baseboundalg is outlined in Algorithm \ref{alg:bboundalg}. For each vertex $u$, $S_u$ is a map to maintain the number of the shortest paths that do not go through $u$ for all neighbor pairs. Algorithm \ref{alg:bboundalg} works as follows. It first calculates the upper bound $\bound(u)$ for each vertex $u$ based on \lamref{lem:baselemma2} and initializes $C_B(u)$ as $\bound(u)$ (lines 1-2). Then, it sorts the vertices in non-increasing order with respect to their upper bounds, and picks an unexplored vertex $u$ with the maximum $\bound(u)$ to calculate $C_B(u)$ until the top-$k$ vertices are found (lines 6-19). During the processing of vertex $u$, if the result set $R$ has $k$ vertices and the $\min_{v \in R}C_{B}(v) \ge \bound(u)$ holds, the algorithm terminates (line 7). Otherwise, \baseboundalg computes $C_B(u)$ and identifies whether $u$ should be added into the answer set $R$ (lines 8-18). For vertex $u$, we explore the number of shortest paths between $u$'s neighbors by enumerating the triangles including $u$ and maintain them in the hash map $S_u$. In $S_u$, we always keep a vertex pair $(i, j)$ with $val = 0$ if $i$ and $j$ are connected in $G_{E(u)}$; on the other hand, $val$ records the number of vertices that link $i$ and $j$ but not contain $u$. When a $\triangle_{(u, v, w)}$ is found, we update the hash maps for $u$, $v$ and $w$ (lines 12-13). Note that \baseboundalg processes vertices in the order of the upper bounds (i.e., the total order), all triangles containing $u$ can be touched without omission after handling $u$ and $S_u$ maintains the number of the shortest paths correctly. Further, the algorithm calculates $C_B(u)$ according to \lamref{lem:baselemma2} and updates $R$ (lines 15-18). Finally, \baseboundalg outputs the answer set $R$.



\begin{example}
Consider a graph $G$ as shown in \figref{fig:expgraph} and an integer $k = 5$. The running process of Algorithm \ref{alg:bboundalg} on this graph is illustrated in \figref{fig:baseexp}. The algorithm computes the ego-betweennesses of $c, i, f, d, x, e, h, g, b, a$ in turn based on their upper bounds (i.e., the total order). After computing $C_B(a)$, the largest upper bound among the remaining vertices: $j, k, u, v, x, y, z$ is $\bound(j) = 3 < C_B(d) = 14/3$ ($d$ is the $5$-th element in $R$), thus Algorithm \ref{alg:bboundalg} terminates. Compared with calculating the ego-betweennesses of all vertices, \baseboundalg can save 6 ego-betweenness computations by utilizing the upper bound \bound.
\end{example}

\begin{figure}[t]\vspace*{-0.1cm}
	\centering
	\includegraphics[width=0.38\textwidth]{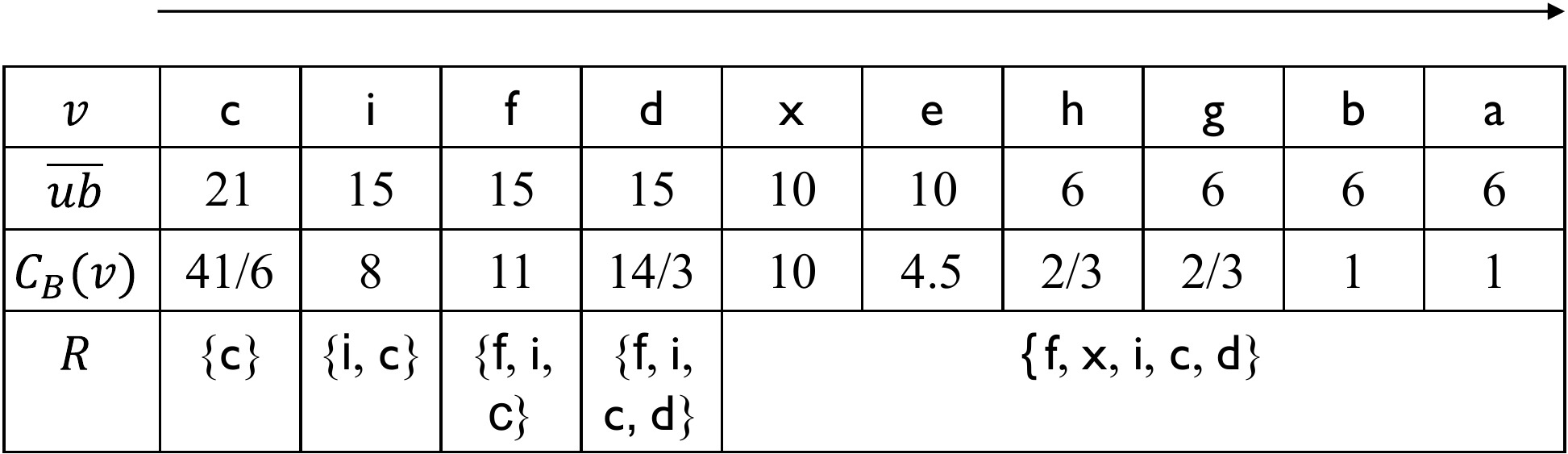}
    \vspace*{-0.2cm}
	\caption{The running process of \baseboundalg on $G$}
    \label{fig:baseexp}\vspace*{-0.2cm}
\end{figure}

\subsection{The \optboundalg algorithm} \label{subsec:onlineVLDB}
\baseboundalg may not be very efficient for top-$k$ search because the upper bound \bound is not very tight. To further improve the efficiency, we propose the \optboundalg algorithm with a dynamic upper bound \optbound which is tighter than \bound.

Recall that we calculate $C_B(u)$ with the information of the shortest paths which is derived by touching the triangles including vertex $u$. In this processing, some useful information about the number of shortest paths for $u$'s neighbors can also be obtained. We refer to those information as identified information which include some vertex pairs and edges. Below, we will use these identified information to derive a tighter and dynamically-updated upper bound of ego-betweenness.

Given a vertex $p$, let $*\bar{S}_{E(p)}$ be the collection of identified edges in $G_{E(p)}$ and $*\hat{S}_{E(p)}$ be the set of the currently identified vertex pairs whose property is the same as the pairs in $\hat{S}_{E(p)}$. For a vertex pair $(u, v)$ in $*\hat{S}_{E(p)}$, denote by $*\hat{S}_{p(u, v)}$ the set of identified vertices that link $u$ and $v$ but does not contain $p$. Let $*\bar{C}_p$ and $*\hat{C}_p$ be the size of $*\bar{S}_{E(p)}$ and $*\hat{S}_{E(p)}$, respectively. We develop a tighter upper bound of ego-betweenness \optbound in \lamref{lem:optlemma}.


\begin{Lemma}
\label{lem:optlemma}
For a vertex $p$ in $G$, $C_B(p) \le \optbound=\frac{d(p)*(d(p)-1)}{2} - *\bar{C}_p - *\hat{C}_p + \sum_{(u, v)} \frac{1}{|*\hat{S}_{p(u, v)}|+1}$ holds.
\end{Lemma}

\begin{proof}
By definition, we have $*\bar{C}_p \le \bar{C}_p$, $*\hat{C}_p \le \hat{C}_p$ and $|*\hat{S}_{p(u, v)}| \le |\hat{S}_{p(u, v)}|$. Further, $\sum_{(u, v)} \frac{1}{|*\hat{S}_{p(u, v)}|+1} \ge \sum_{(u, v)} \frac{1}{|\hat{S}_{p(u, v)}|+1}$ holds. According to \lamref{lem:baselemma2}, we can obtain $C_B(p) \le \optbound = \frac{d(p)*(d(p)-1)}{2} - *\bar{C}_p - *\hat{C}_p + \sum_{(u, v)} \frac{1}{|*\hat{S}_{p(u, v)}|+1}$.
\end{proof}


\begin{algorithm}[t]
  \scriptsize
  \LinesNumbered
  \caption{\optboundalg$(G, k, \theta)$}
  \label{alg:optboundalg}
  \KwIn{$G = (V, E)$, an integer $k \ge 1$, a gradient ratio $\theta \ge 1$.}
  \KwOut{The top-$k$ vertex set $R$.}
  $H \leftarrow \emptyset$; $R \leftarrow \emptyset$\;
  Initialize an array $B$ with $B(i) = false, 0 \le i < n$\;
  \For{$v \in V$}
  {
      $\optbound(v) \leftarrow \frac{d(v)*(d(v)-1)}{2}$; $C_{B}(v) \leftarrow \optbound(v)$; $H.push(v, \optbound(v))$\;
  }
  \While{$H \ne \emptyset$}
  {
      $(v^*, \topbound) \leftarrow H.pop()$\;
	  Compute $\optbound(v^*)$ according to \lamref{lem:optlemma}\;
	  \If{$\theta \cdot \optbound(v^*) < \topbound$}
	  {
        \If{$|R| < k$ {\it{or}} $\optbound(v^*) > \min_{v \in R}C_{B}(v)$}{$H.push(v^*, \optbound(v^*))$\;}
        {\bf continue}\;
		
	  }
      {\bf {if}} $|R| = k$ {\it{and}} $\topbound \le \min_{v \in R}C_{B}(v)$ {\bf {then}} {\bf break}\;
	  \egobwcalalg($G, v^*, B$)\;
      {\bf {if}} $|R| < k$ {\bf {then}} $R \leftarrow R \cup \{ v^* \}$\;
	  \ElseIf{$C_{B}(v^*) > \min_{v \in R}C_{B}(v)$}
	  {
        $u \leftarrow \arg\min_{v \in R}C_{B}(v)$; $R \leftarrow (R - \{ u \}) \cup \{ v^* \}$\;
	  }
      $B(v^*) \leftarrow true$\;
  }
  {\bf return} $R$\;
\end{algorithm}

\begin{figure*}[t!]\vspace*{-0.2cm}
\centering
  \subfigure[\scriptsize{Pop out $f$}]{
  \label{fig:optexp1}
  \begin{minipage}{5cm}
  \centering
  \includegraphics[width=\textwidth]{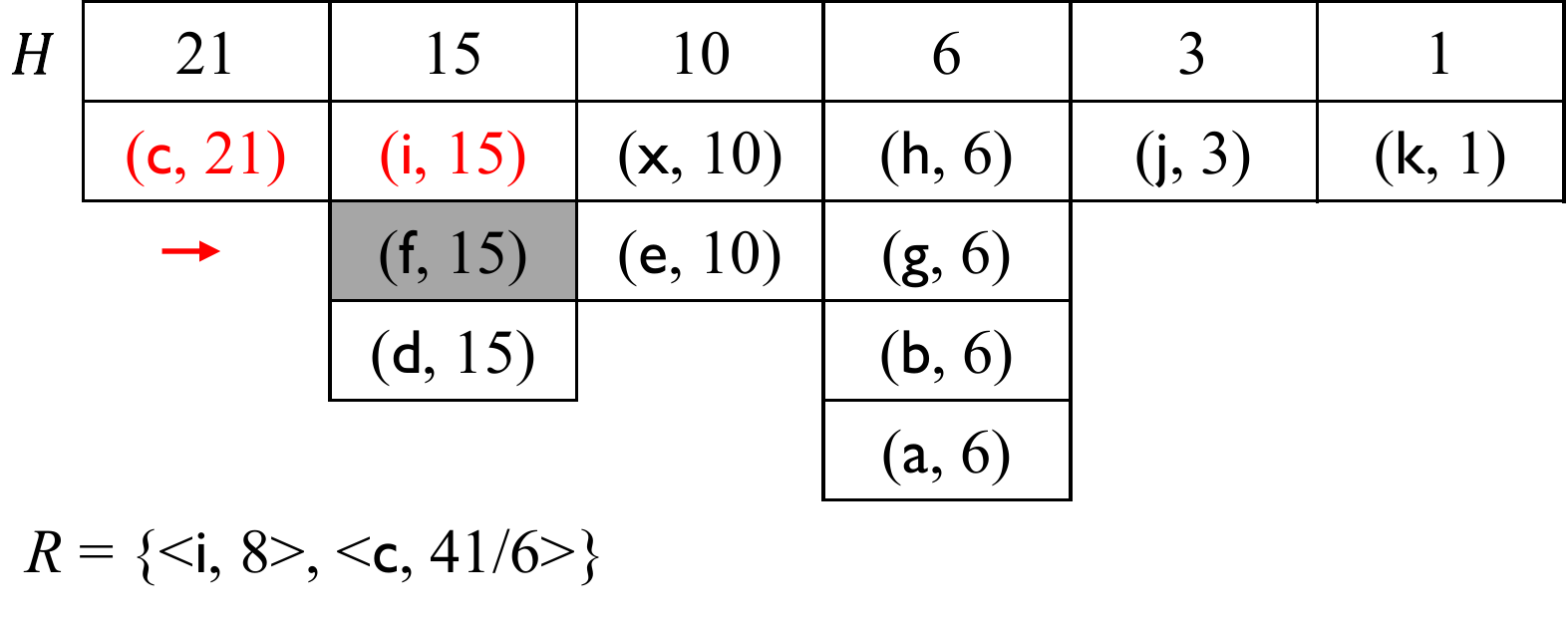}
  \end{minipage}
  }
  \subfigure[\scriptsize{Pop out $d$}]{
  \label{fig:optexp2}
  \begin{minipage}{5cm}
  \centering
  \includegraphics[width=\textwidth]{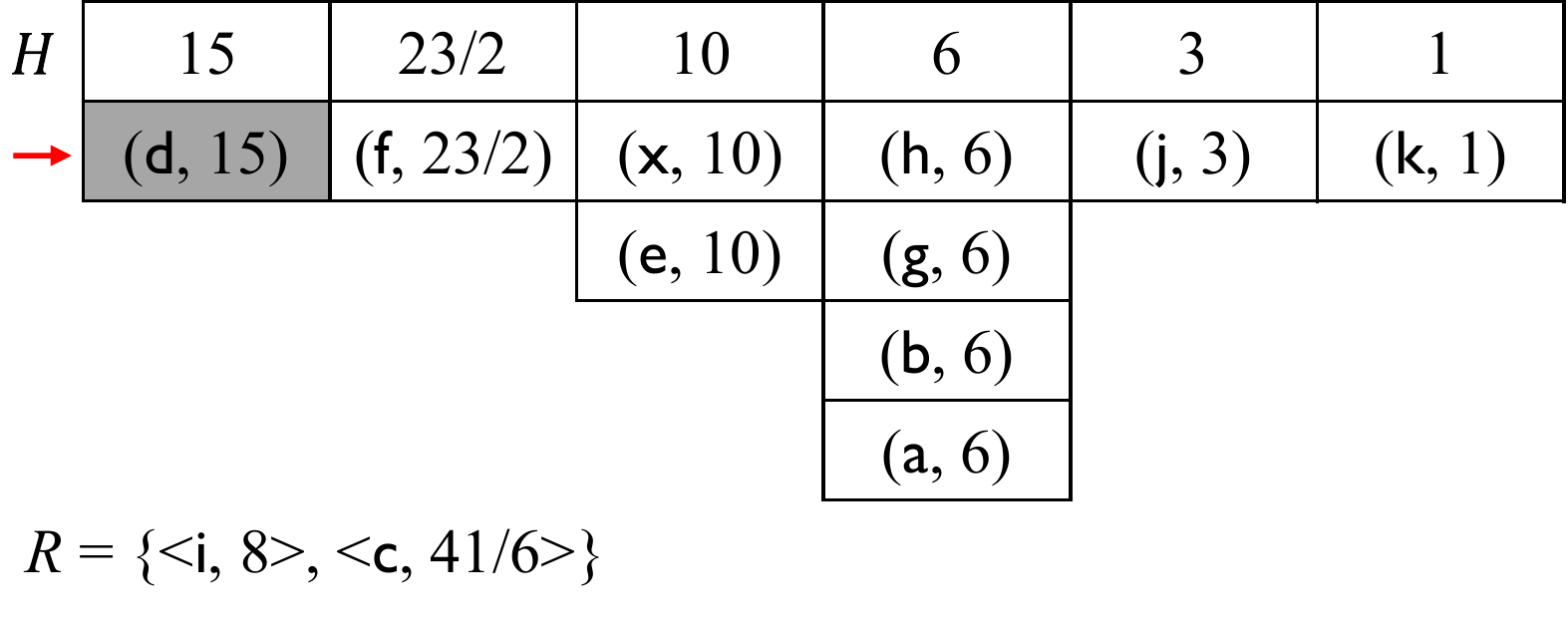}
  \end{minipage}
  }
  \subfigure[\scriptsize{Pop out $e$}]{
  \label{fig:optexp3}
  \begin{minipage}{5cm}
  \centering
  \includegraphics[width=\textwidth]{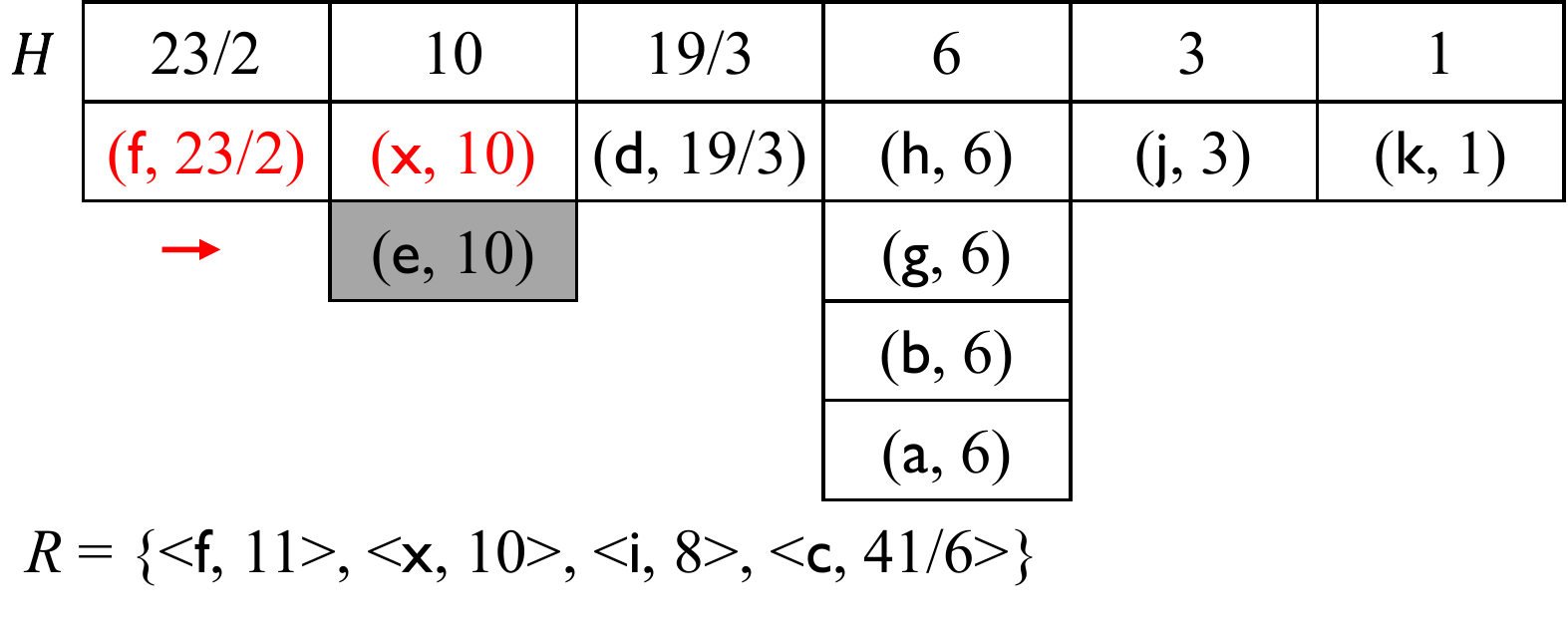}
  \end{minipage}
  }
  \vspace*{-0.2cm}

  \subfigure[\scriptsize{Pop out $h, g, b, a$}]{
  \label{fig:optexp4}
  \begin{minipage}{5cm}
  \centering
  \includegraphics[width=\textwidth]{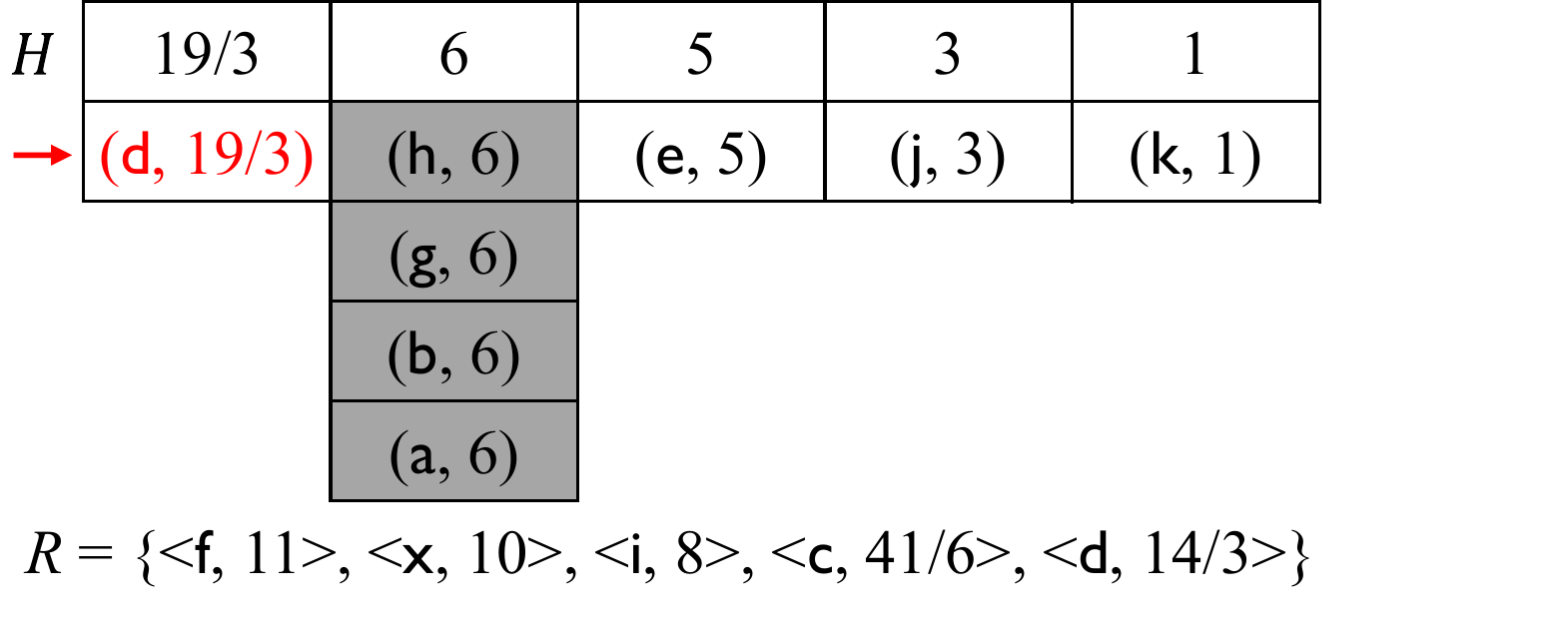}
  \end{minipage}
  }
  \subfigure[\scriptsize{Pop out $e$}]{
  \label{fig:optexp5}
  \begin{minipage}{5cm}
  \centering
  \includegraphics[width=\textwidth]{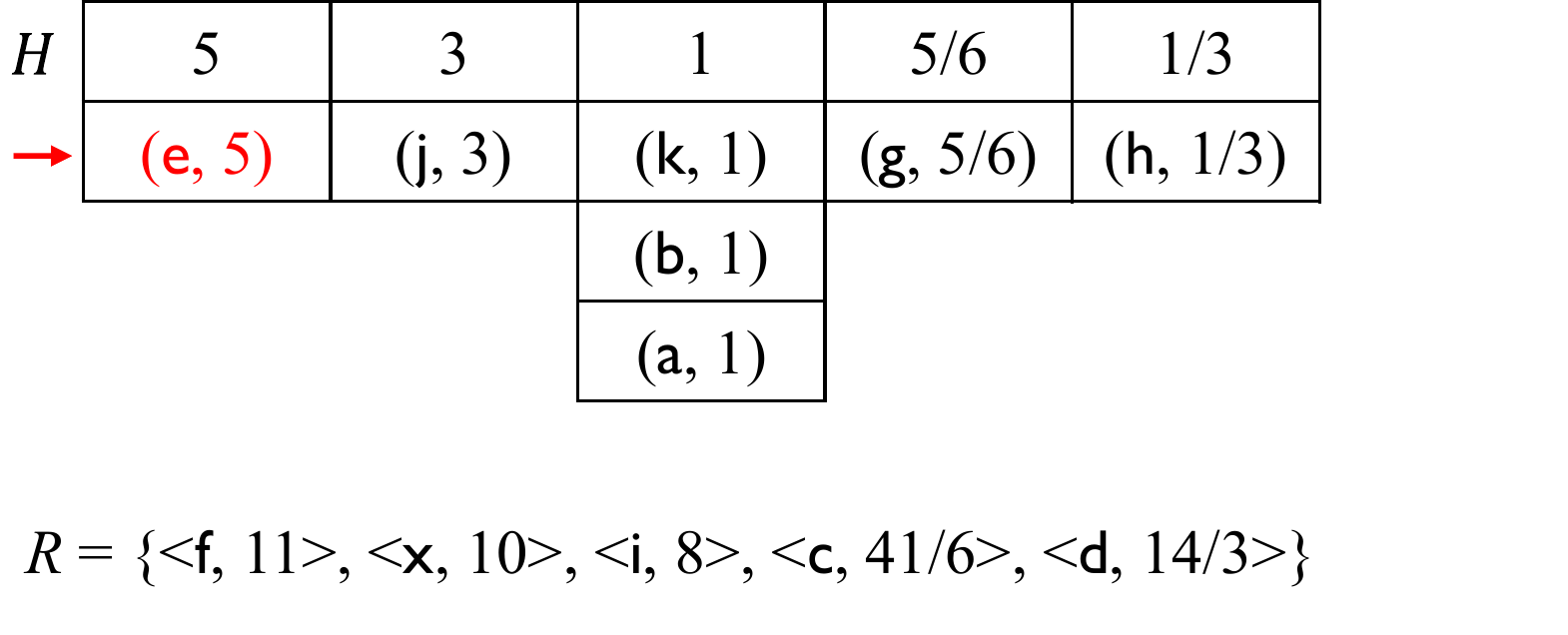}
  \end{minipage}
  }
  \subfigure[\scriptsize{Pop out $j$}]{
  \label{fig:optexp6}
  \begin{minipage}{5cm}
  \centering
  \includegraphics[width=\textwidth]{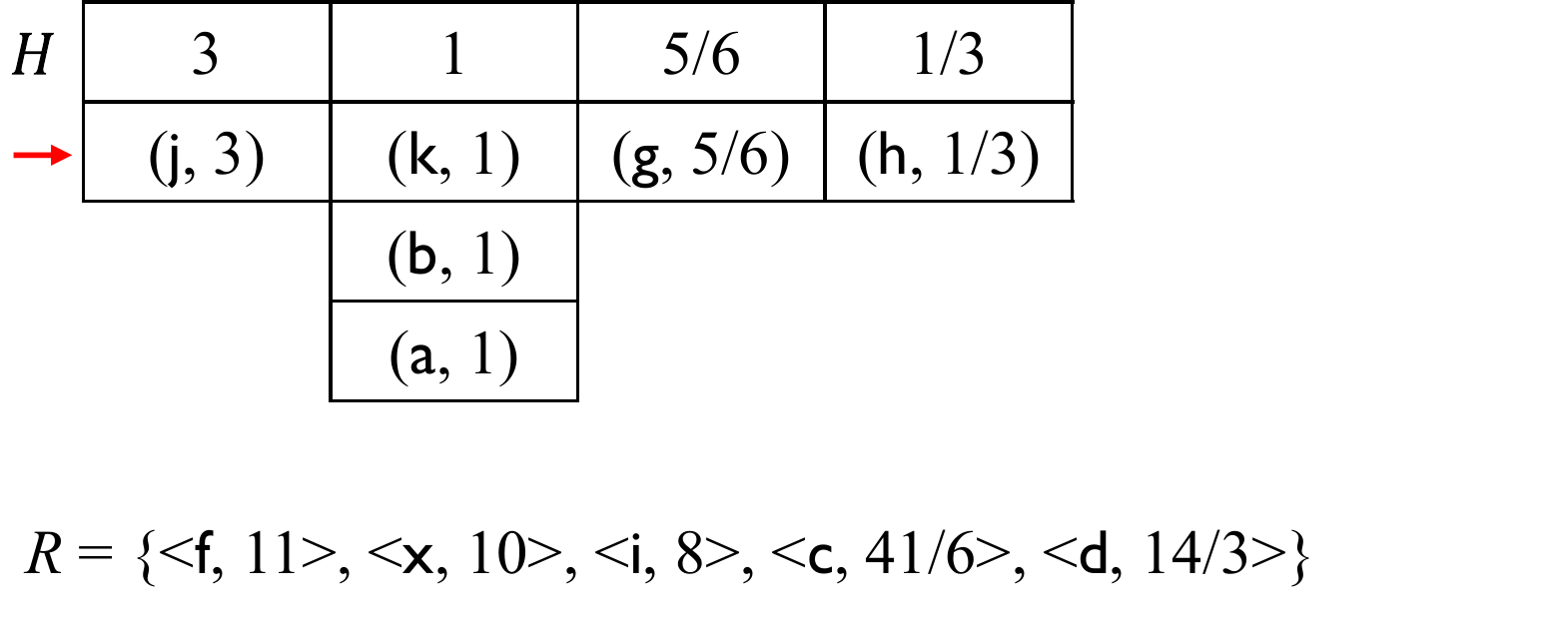}
  \end{minipage}
  }
\vspace*{-0.2cm}
\caption{The running process of \optboundalg on $G$}
\label{fig:optexp}\vspace*{-0.2cm}
\end{figure*}

Note that the upper bound \optbound in \lamref{lem:optlemma} will be dynamically updated during the execution of the top-$k$ search algorithm, because $*\bar{C}_p$, $*\hat{C}_p$ and $|*\hat{S}_{p(u, v)}|$ will be updated when calculating vertices' ego-betweennesses exactly. The \optboundalg framework with such a dynamic upper bound \optbound is depicted in Algorithm \ref{alg:optboundalg}. It first calculates $\optbound(v)$ and $C_B(v)$ for each vertex $v$, and pushes $v$ with the initial bound $\optbound(v)$ into a sorted list $H$ (lines 3-4). Then, the \optboundalg iteratively finds the top-$k$ results (lines 5-17). It pops the vertex $v^*$ with the largest upper bound value \topbound from $H$. As the number of shortest paths between $v^*$'s neighbors may be updated, the algorithm calculates $\optbound(v^*)$ based on \lamref{lem:optlemma}. \optboundalg then compares $\optbound(v^*)$ with the old bound $\topbound$ by employing a parameter $\theta \ge 1$ to avoid frequently calculating the upper bounds and updating $H$. When $\theta \cdot \optbound(v^*) < \topbound$, that means $\optbound(v^*)$ is substantially smaller than $\topbound$. If $|R| < k$ or $\optbound(v^*) > \min_{v \in R}C_{B}(v)$, we push $v^*$ to $H$ again with the tighter bound $\optbound(v^*)$ (line 10). Otherwise, $v^*$ does not belong to the top-$k$ answers and thus can be pruned. In both cases, the algorithm needs to pop the next vertex from $H$. If the early termination condition (line 12) is not satisfied, the algorithm performs \egobwcalalg to compute $C_B(v^*)$ exactly and updates $R$ based on $C_B(v^*)$ (lines 13-17). Note that we use an array $B$ to record the vertices whose ego-betweennesses have been calculated, which can reduce redundant computations in the \egobwcalalg procedure.

\begin{algorithm}[t]
  \scriptsize
  \LinesNumbered
  \caption{\egobwcalalg$(G, u, B)$}
  \label{alg:egobwdalg}
  \KwIn{$G = (V, E)$, vertex $u$, an array $B$.}
  \KwOut{$C_B(u)$.}
  Initialize $DN$ and $EN$ according to $B$\;
  Initialize an array $V_{is}$ with $V_{is}(i) = false, 0 \le i < n$\;
  {\bf {for}} $i \in N(u)$ {\bf {do}} $rd(i) \leftarrow \emptyset$\;
  \For{$((i, j), val) \in S_u$}
  {
    \If{$val = 0$}
	{	
        $rd(i) \leftarrow rd(i) \cup  \{ j \}$; $rd(j) \leftarrow rd(j) \cup  \{ i \}$\;
    }
  }

  \For{$i \in DN$}
  {
    {\bf {for}} $p \in rd(i)$ {\bf {do}} $V_{is}(p) \leftarrow true$\;
    \For{$j \in DN - \{ i \}$}
    {
        \If{$V_{is}(j) = false$}
		{
            \For{$p \in rd(j)$}
            {
                \If{$V_{is}(p) = true$ and $B(p) = false$}
        		{

                    $S_{u}(i, j).val$++; $S_{p}(i, j).val$++\;
                }

            }
        }
    }

  }
  \For{$i \in EN$}
  {
     \For{$j \in EN - \{ i \}$}
     {
        \If{$(i, j)\in E$}
		{
            $S_{u}.{\kw{insert}}((i, j), 0)$; $S_{i}.{\kw{insert}}((u, j), 0)$\; $S_{j}.{\kw{insert}}((u, i), 0)$\;
            \For{$k \in rd(j)$}
            {
                \If{$\nexists S_{u}(i, k)$}
                {
                    $S_{u}.{\kw{insert}}((i, k), 1); S_{j}.{\kw{insert}}((i, k), 1)$;
                }
                \ElseIf{$S_{u}(i, k).val \ne 0$}
                {
                    $S_{u}(i, k).val$++; $S_{j}(i, k).val$++\;
                }
            }
            Update $S_u, S_i$ by $rd(i)$ as lines 19-23\;
            $rd(i) \leftarrow rd(i) \cup \{j\}$; $rd(j) \leftarrow rd(j) \cup \{i\}$\;
	    }
    }
  }
  Calculate $C_{B}(u)$ as lines 15-17 of Algorithm \ref{alg:bboundalg}\;
  {\bf return} $C_{B}(u)$\;
\end{algorithm}

Algorithm \ref{alg:egobwdalg} outlines the \egobwcalalg procedure. Like \baseboundalg, a key issue is maintaining the number of the shortest paths in $S_u$ correctly by finding the triangles containing $u$. To avoid reduction, a simple but efficient approach is to record those enumerated triangles and update $S_u$ by deriving the shortest paths from these triangles. To this end, for each neighbor $i$ of $u$, Algorithm \ref{alg:egobwdalg} uses $rd(i)$ to store such vertices that are contained in the touched triangles $\triangle_{(i, *, u)}$. It first initializes $rd(i)$ for every $i \in N(u)$ with the current $S_u$ as $S_u(i, j).val$ equals $0$ indicates a visited triangle $\triangle_{(i, j, u)}$ (lines 3-6). Then, the procedure handles $u$'s neighbors to maintain $S_u$ according to whether they have been processed (lines 7-25). Specifically, if $B(i)=true$, we put $i$ into the set $DN$ and call it a processed vertex; otherwise, $i$ is added into the set $EN$ where stores the vertices to be processed. For the vertices $i, j \in DN$, \egobwcalalg finds their common neighbors (denote by $p$) based on $rd(i)$ and $rd(j)$ and updates the number of the shortest paths between $i$ and $j$ for $S_u$ and $S_p$ (lines 7-13). On the other hand, given $i, j \in EN$, the procedure enumerates new triangles and maintains related hash maps with $rd(i)$ and $rd(j)$ (lines 14-25). Note that with the discovery of new triangles, \egobwcalalg also updates the related $rd(i)$s to avoid reduction (line 25). Finally, \egobwcalalg calculates $C_B(u)$ with the same method as used in \baseboundalg.

\begin{example}
Reconsider the graph $G$ in \figref{fig:expgraph}. Suppose that $k = 5$ and $\theta = 1$. The running process of Algorithm \ref{alg:optboundalg} is illustrated in \figref{fig:optexp}. The vertices colored red are computed their ego-betweennesses exactly and the vertices in gray grids need to update their upper bounds and push back into $H$ again. The algorithm pushes all vertices with the initial upper bounds into $H$ and then processes them based on $H$. First, it pops $c$ with the largest upper bound $\topbound = 21$ and calculates $\optbound(c)$ and $C_B(c)$. Due to $R = \emptyset$, $c$ is added into $R$ and \optboundalg does the same operation for $i$. Then, $f$ is popped with $\topbound = 15$ and Algorithm \ref{alg:optboundalg} calculates $\optbound(f)$ as shown in \figref{fig:optexp1}. Since $\optbound(f)=23/2$ is substantially smaller than $\topbound$ based on $\theta = 1$, we push $(f, 23/2)$ into $H$ again and pops $d$ as the next processing vertex in \figref{fig:optexp2}. The tighter bound $\optbound(d)=19/3$ is less than $15$, thus \optboundalg pushes $d$ into $H$ again with $\optbound(d)$. In the following three iterations, \optboundalg computes $C_B(f)$ and $C_B(x)$ and adds them into $R$, and then processes $e$ as shown in \figref{fig:optexp3}. $e$ is pushed into $H$ with $\optbound(e)=4$ and the algorithm pops $d$ to calculate $C_B(d)$ and adds $d$ into $R$ in \figref{fig:optexp4}. Due to $\topbound = 6 > C_B(d)$, $h$ is popped and we calculate $\optbound(h)$ to update $H$. Similarly, we push $g, b, a$ into $H$ again with $\optbound(g), \optbound(b), \optbound(a)$ as shown in \figref{fig:optexp5}. When $e$ is processed, $C_B(e)=9/2 < C_B(d)$ and $|R|=k=5$ hold, thus $e$ is not an answer of top-$k$ results. When pops $j$ in \figref{fig:optexp6}, the algorithm safely prunes $j$ since $\optbound(j) < C_B(d)$. Obviously, the remaining vertices can also be pruned. In \optboundalg, we invoke \egobwcalalg six times to calculate the ego-betweennesses, while \baseboundalg performs ten ego-betweenness computations.
\end{example}

\comment{
\begin{algorithm}[t]
  \caption{\egobwcalalg$(G, v)$}
  \label{alg:egobwdalgooo}
  \KwIn{$G = (V, E)$, vertex $v$.}
  \KwOut{$C_B(v)$.}

  $S \leftarrow \emptyset$\;
  \For{$i \in N(v)$}
  {
    {\bf {if}} $B(i) = true$ {\bf {then}} $DN \leftarrow DN \cup \{ i \}$\;
    {\bf {else}} $EN \leftarrow EN \cup \{ i \}$\;
    $S \leftarrow S \cup \{ i \}$\;
  }
  \For{each $i \in S$}
  {
     \For{each $j \in (S - \{ i \})$}
     {
	    \For{each $k \in (S - \{ i, j\})$}
        {
            \If{$(i, j)\in E$}
			{
		        $S_{v}.insert((i, j), 0)$\;
		    }
			\If{$(j, k)\in E$}
			{
		        $S_{v}.insert((j, k), 0)$\;
		    }
			\If{$(i, k)\in E$}
			{
		        $S_{v}.insert((i, k), 0)$\;
		    }
	        \If{$(i, j) \in E$ and $(i, k) \in E$ and $(j, k)\notin E$}
			{
		        \If{$S_{i}.find((j, k)) = empty$}
				{
			       $S_{i}.insert((j, k), 0)$\;
			    }
  				\If{$S_{v}.find((j, k)) = empty$}
				{
			       $S_{v}.insert((j, k), 0)$\;
			    }
				$S_{i}.get((j,k))++;S_{v}.get((j,k))++$\;
		    }
			\If{$(i, j) \in E$ and $(j, k) \in E$ and $(i, k) \notin E$}
			{
		        the same as line12-16\;
		    }
            \If{$(i,k)\in E$ and $(j, k) \in E$ and $(i, j) \notin E$}
			{
		        the same as line12-16\;
		    }
        }
	 }
  }
  \For{each (key,value) $\in S_{v}$}
  {
    $C_{B}(v)--$\;
    \If{$value \ne 0$}
    {
      $C_{B}(v) \leftarrow C_{B}(v) + \frac{1}{value+1}$\;
    }
  }
  {\bf return} $C_{B}(v)$\;
\end{algorithm}
}

\subsection{ Analysis of the proposed algorithms} \label{subsec:onlinecplex}

Below, we mainly analyze the correctness of Algorithm \ref{alg:bboundalg}. The correctness analysis of Algorithm \ref{alg:optboundalg} is similar to that of Algorithm \ref{alg:bboundalg}, thus we omit it for brevity.

\begin{theo}
Given a graph $G = (V, E)$ and an integer $k$, Algorithm \ref{alg:bboundalg} correctly computes the top-$k$ vertices with the highest ego-betweennesses.
\end{theo}

\begin{proof}
Recall that Algorithm \ref{alg:bboundalg} iteratively processes the vertices based on their upper bounds (\lamref{lem:baselemma2}). When a vertex $u$ is handled, if the answer set $R$ has $k$ vertices and $\min_{v \in R}C_{B}(v) \ge \bound(u)$, then $C_{B}(u) \le \bound(u) \le \min_{v \in R}C_{B}(v)$ holds. For any vertex $w \in V$ with a smaller degree, we have $C_{B}(w) \le \bound(w) \le \bound(u) \le \min_{v \in R}C_{B}(v)$. Therefore, the algorithm can safely prune the remaining vertices and terminate, thereby the set $R$ exactly contains the top-$k$ answers.
\end{proof}

Below, we analyze the time and space complexity of Algorithm~\ref{alg:bboundalg} and Algorithm~\ref{alg:optboundalg}. Let $d_{\max}$ be the maximum degree of the vertices in $G$, and $\alpha$ be the arboricity of $G$ \cite{64arboricity, 12tcsarboricity}.

\begin{theo}
 In the worst case, both Algorithm~\ref{alg:bboundalg} and Algorithm~\ref{alg:optboundalg} take $O(\alpha m d_{\max})$ time using $O(d_{\max} m)$ space.
\end{theo}

\begin{proof}
We mainly show the complexity of Algorithm~\ref{alg:bboundalg}, and the complexity analysis of Algorithm~\ref{alg:optboundalg} is similar. First, in lines 6-18 of Algorithm~\ref{alg:bboundalg}, the algorithm needs to enumerate each triangle once which takes $O(\alpha m)$ time. Note that when a triangle $\triangle_{(u, v, w)}$ is enumerated, the algorithm requires to maintain $S_{u}, S_{v}, S_{w}$. The time overhead of the update operator can be bounded by $O(d(u))\le O(d_{\max})$. Hence, the time complexity of Algorithm \ref{alg:bboundalg} is $O(\alpha md_{\max})$. Second, we analyze the space complexity of Algorithm~\ref{alg:bboundalg}. Clearly, the space overhead is dominated by the size of the map structure $S_{u}$. For $u \in V$, the map structure $S_{u}$ contains $O(d(u)^2)$ vertex pairs, thus the space complexity is $O(\sum_{u \in V}d(u)^2) \le O(d_{\max} m)$.
\end{proof}

Although Algorithm~\ref{alg:bboundalg} and Algorithm~\ref{alg:optboundalg} have the same worst-case time complexity, the practical performance of Algorithm~\ref{alg:optboundalg} is much faster than that of Algorithm~\ref{alg:bboundalg} due to the dynamic and tight upper bound, which is also confirmed in our experiments.

\section{The update algorithms} \label{sec:onlineupdate}

Real-world networks are often frequently updated. In this section, we develop local update algorithms to maintain the ego-betweennesses for all vertices when the graph is updated. We also propose lazy update techniques to efficiently maintain the top-$k$ results. We mainly focus on the cases of edge insertion and deletion, as vertex insertion and deletion can be seen as a series of edge insertions and deletions.

Our update algorithms are based on the following key observation.


\begin{obser}\label{obs:insertVertex}
After inserting/deleting an edge $(u, v)$ into/from $G$, the ego-betweennesses of the vertices in $N(u, v) \cup \{u, v\}$ need to be updated, and the ego-betweennesses of the vertices that are not in $N(u, v) \cup \{u, v\}$ remain unchanged.
\end{obser}

\begin{proof}
Here, we prove the edge insertion case and the proof for edge deletion is similar. The insertion of $(u, v)$ causes the insertions of vertex $v/u$ and a series of edges $\{(v, w) | w \in N(u, v) \}/\{(u, w) | w \in N(u, v) \}$ into $u/v$'s ego network $G_{E(u)}/G_{E(v)}$, thus the ego-betweennesses of $u$ and $v$ need to be updated. In addition, for a common neighbor $w \in N(u, v)$, there is a new edge $(u, v)$ in $G_{E(w)}$, thus the ego-betweenness of $w$ should be re-computed.
\end{proof}

\subsection{Local-update for edge insertion} \label{sec:updateinsert}

We present the update rules for the vertices $u$, $v$ and $w \in N(u, v)$ when inserting an edge $(u, v)$. For brevity, let $L = N(u, v)$ denote the common neighbors of $u$ and $v$ and $S_u(x, y)$ be the number of vertices that link $x$ and $y$ but does not include $u$. Unless otherwise specified, $S_u(x, y)$ represents the value after inserting the edge.

\begin{algorithm}[t]
  \scriptsize
  \LinesNumbered
  \caption{\insertedge}
  \label{alg:insertuptalg}
  \KwIn{$G = (V, E)$, ego-betweenness array $ C_B$, an inserted edge $(u, v)$.}
  \KwOut{the updated $\cal C_B$.}
  Insert $(u, v)$ into $G$\;
  ${\mathcal S} \leftarrow \uptlocalsmap(G, (u, v))$\;
  $L \leftarrow N(u) \cap N(v)$\;
  \For{$(x, y) \in S_{u}$}
  {
    \If{$x = v$ or $y = v$}
    {
        $ {C_B}(u) \leftarrow { C_B}(u) + 1/(S_{u}(x, y)+1)$\;
    }
    \Else
    {
        ${ C_B}(u) \leftarrow { C_B}(u) + 1/(S_{u}(x, y)+1) - 1/S_{u}(x, y) $\;
    }
  }
  Update ${C_B}(v)$ as lines 4-8\;
  \For{$x \in L$}
  {
    \For{$(y, z) \in S_{x}$}
    {
        \If{($y = u$ and $z = v$) or ($y = v$ and $z = u$)}
        {
            ${C_B}(x) \leftarrow { C_B}(x) - 1/(S_{x}(y, z)+1)$\;
        }
        \Else
        {
            ${C_B}(x) \leftarrow { C_B}(x)+ 1/(S_{x}(y, z)+1) - 1/S_{x}(y, z) $\;
        }
    }
  }
  {\bf return} ${C_B}$\;
\end{algorithm}

\begin{Lemma}
\label{obs:insertuv}
Consider an inserted edge $(u, v)$, the updated ego-betweenness of $u$ is: $C_B(u) = C_B(u)+\sum_{x, y \in L, (x, y) \notin E}(1/(S_u(x, y)+1)-1/S_u(x, y))+ \sum_{x \in N(u), x \notin L} 1/(S_u(v, x)+1)$. The case of updating $C_B(v)$ is similar.
\end{Lemma}

\begin{proof}
For vertex $u$, after inserting an edge $(u, v)$ into $G$, $v$ is a new neighbor and is added into $G_{E}(u)$. For $x, y \in L$ and $(x, y) \notin E$, $C_B(u)$ has included the contribution of vertex pair $(x, y)$, thus we should update this part. $v$ is a new vertex that connects $x$ and $y$, and the number of the shortest paths between $x$ and $y$ only adds 1, thus we can calculate $S_u(x, y)$ and reveal the previous contribution to update $C_B(u)$, i.e., $C_B(u) = C_B(u) + 1/(S_u(x, y)+1)-1/S_u(x, y)$. In addition, for $x \in L$, $x$ and $v$ are connected, thus it does not contribute to $C_B(u)$. For $x \notin L$, $(v, x)$ is a new vertex pair which makes $C_B(u)$ increase, thus we need to compute $S_u(v, x)$ and update $C_B(u)$ by adding $1/(S_u(v, x)+1)$.
\end{proof}

\begin{Lemma}
\label{obs:insertcnbr}
Consider an inserted edge $(u, v)$, the updated ego-betweenness of $w \in L$ is: $C_B(w) = C_B(w) - 1/(S_w(u, v)+1) + \sum_{x \in N(w) \cap N(u)-\{v\}, (x, v) \notin E} (1/(S_w(x, v)+1)-1/S_w(x, v)) + \sum_{x \in N(w) \cap N(v)-\{u\}, (x, u) \notin E} (1/(S_w(x, u)+1)-1/S_w(x, u))$.
\end{Lemma}

\begin{proof}
For vertex $w \in L$, the insertion of $(u, v)$ causes the direct connection between $u$ and $v$ in $G_{E}(w)$ which makes $C_B(w)$ decrease. We need to compute $S_w(u, v)$ before the insert operation and update $C_B(w)$ as: $C_B(w) = C_B(w) - 1/(S_w(u, v)+1)$. In addition, for $x \in N(w) \cap N(v)$ and $(u, x) \notin E$, $v$ now is a new vertex that links $u$ and $x$, thus we calculate $S_w(u, x)$ and update $C_B(w)$ as: $C_B(w) = C_B(w) + 1/(S_w(u, x)+1)-1/S_w(u, x)$. Analogously, for $x \in N(w) \cap N(u)$ and $(v, x) \notin E$, $u$ is a new vertex that connects $v$ and $x$, we update $C_B(w)$ as the above operation.
\end{proof}

Equipped with the above lemmas, we propose a local update algorithm, called \insertedge, to maintain the ego-betweennesses for handling edge insertion. The pseudo-code of \insertedge is illustrated in Algorithm~\ref{alg:insertuptalg}. \insertedge first inserts the edge $(u, v)$ into $G$ (line 1). Then, it invokes the \uptlocalsmap (Algorithm~\ref{alg:uptlocalsmapalg}) to recompute the number of shortest paths of the affected vertex pairs in the ego networks of $u, v$ and their common neighbors (Observation~\ref{obs:insertVertex}). Finally, \insertedge updates the ego-betweennesses for affected vertices. For the endpoints $u, v$ of the inserted edge, we calculate $C_B(u)$ and $C_B(v)$ based on \lamref{obs:insertuv} (lines 4-9); On the other hand, for the common neighbor $w$, \insertedge updates $C_B(w)$ according to \lamref{obs:insertcnbr} (lines 10-15).

\begin{algorithm}[t]
  \scriptsize
  \LinesNumbered
  \caption{\uptlocalsmap}
  \label{alg:uptlocalsmapalg}
  \KwIn{$G = (V, E)$, an edge $(u, v)$.}
  \KwOut{The map set $\mathcal S$.}
  $L \leftarrow N(u) \cap N(v)$\;
  {\bf {for}} $x \in N(u) \backslash L$ {\bf {do}} $S_{u}.{\kw{insert}}((x, v), 0)$\;
  {\bf {for}} $x \in N(v) \backslash L$ {\bf {do}} $S_{v}.{\kw{insert}}((x, u), 0)$\;
  {\bf {for}} $x \in L$ {\bf {do}} $S_{x}.{\kw{insert}}((u, v), 0)$\;
  \For{$p \in L$}
  {
    \For{$x \in N(u) \cap N(p)$}
    {
        \If{$(x, v) \notin E$ and $x \neq v$}
        {
            $S_{u}(x, v)$++\;
            \For{$y \in N(x) \cap N(v) \cap N(p)$}
            {
                {\bf {if}} $\nexists S_{p}(x, v)$ {\bf {then}} $S_{p}.{\kw{insert}}((x, v), 0)$\;
                $S_{p}(x, v)$++\;
            }
        }

    }
    Update $S_v$, $S_p$ as lines 6-11\;
    \For{$q \in L$}
    {
        \If{$(p, q) \notin E$ and $q \prec p$ }
        {
            \For{$y \in N(u) \cap N(p) \cap N(q)$}
            {
                {\bf {if}} $\nexists S_{u}(p, q)$ {\bf {then}} $S_{u}.insert((p, q), 0)$\;
                $S_{u}(p, q)$++\;
            }
            Update $S_v$ as lines 15-17\;
            \comment
            {
            \For{$N(v) \cap N(p) \cap N(q)$}
            {
                {\bf {if}} $\nexists S_{v}(p, q)$ {\bf {then}} $S_{v}.insert((p, q), 0)$\;
                $S_{v}(p, q)$++\;
            }
            }
        }
        \If{$(p, q) \in E$ and $p \prec q$}
        {
            $S_{p}(u, v)$++; $S_{q}(u, v)$++\;
            \comment
            {
            {\bf {if}} $\nexists S_{q}(u, v)$ {\bf {then}} $S_{q}.insert((u, v), 0)$\;
            $S_{q}(u, v)$++\;
            }
        }
    }
  }
  ${\mathcal S} \leftarrow \{ S_x | x \in L \cup \{u, v\}\}$\;
  {\bf return} ${\mathcal S}$\;
\end{algorithm}

\begin{figure}[t]\vspace*{-0.2cm}
\centering
  \subfigure[$G_{E(k)}$]{
  \label{fig:insertegok}
  \begin{minipage}{2.5cm}
  \centering
  \includegraphics[width=\textwidth]{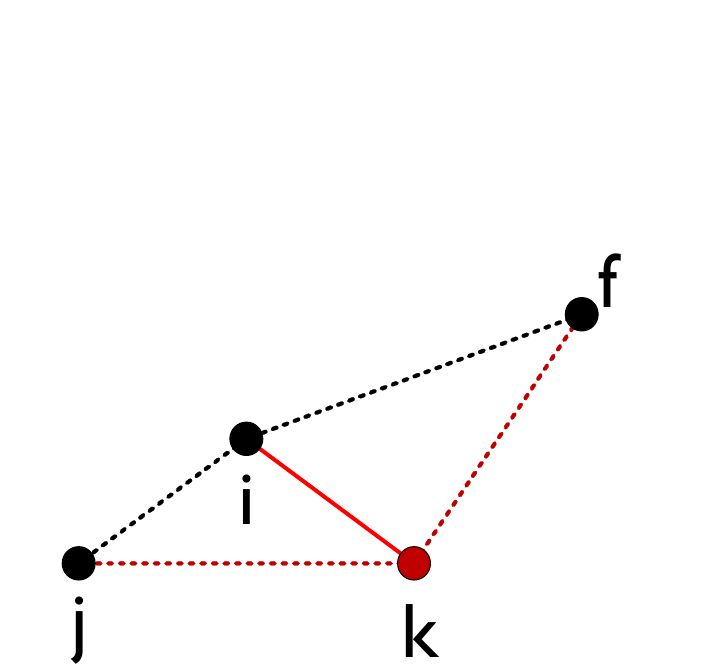}
  \end{minipage}
  }
  \subfigure[$G_{E(i)}$]{
  \label{fig:insertegoi}
  \begin{minipage}{2.5cm}
  \centering
  \includegraphics[width=\textwidth]{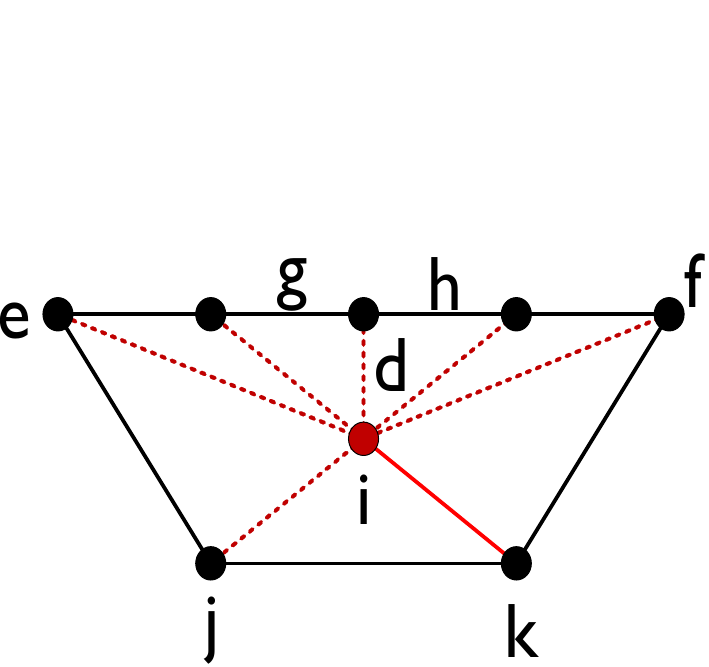}
  \end{minipage}
  }
  \subfigure[$G_{E(f)}$]{
  \label{fig:insertegof}
  \begin{minipage}{2.5cm}
  \centering
  \includegraphics[width=\textwidth]{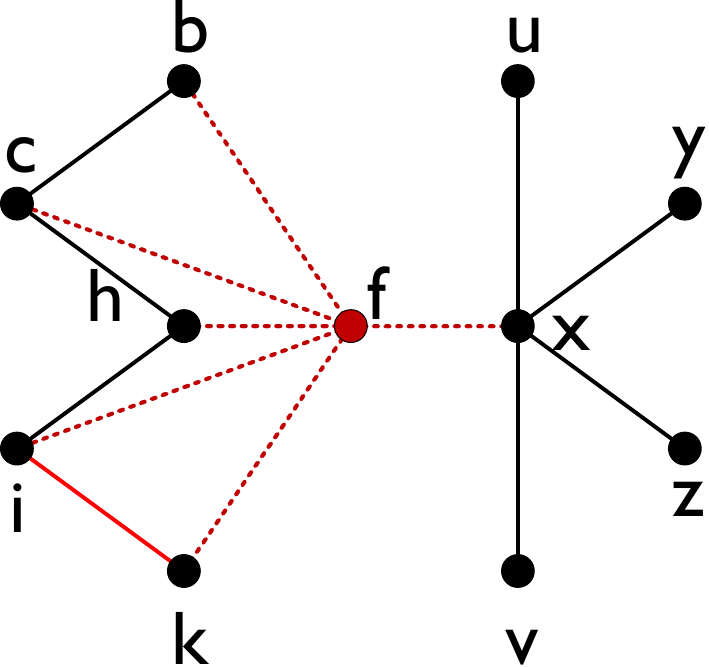}
  \end{minipage}
  }
  \vspace*{-0.2cm}
\caption{Running example}
\label{fig:insertedge}
\vspace*{-0.2cm}
\end{figure}

\begin{example}
Reconsider the graph $G$ in \figref{fig:expgraph}. Suppose that we insert an edge $(i, k)$ into $G$. Clearly, the ego-betweennesses of $i, k$ and their common neighbor change based on Observation \ref{obs:insertVertex}. \figref{fig:insertegok} and \figref{fig:insertegoi} depict the ego networks of $k$ and $i$, respectively. In \figref{fig:insertegok}, the new pairs, i.e., $(f, i)$ and $(j, i)$, are generated due to the connection of $i$ and $k$, thus $C_B(k)$ changes. According to \lamref{obs:insertuv}, the new $C_B(k)$ is $C_B(k) = C_B(k) + 1/(S_k(f, j)+1) - 1/S_k(f, j) = 1 + 1/(1+1) - 1/1= 1/2$. Similarly, we can easily check that the updated ego-betweenness of $i$ is $C_B(i) = 10.5$ from $G_{E(i)}$. For the common neighbor $f$, its ego network $G_{E(f)}$ is shown in \figref{fig:insertegof}. After the insertion of $(i, k)$, $C_B(f)$ decreases from 11 to 9.5. This is because $i$ is a neighbor of $k$ and they no longer need intermediate vertices to reach each other. In addition, the shortest paths for some vertex pairs may pass through $i$ or $j$ and the number of shortest paths of these pairs increases, thus makes $C_B(f)$ decrease.
\end{example}

\subsection{Local update for edge deletion} \label{sec:updatedeletion}

\comment{
\begin{algorithm}[t]
  \caption{\deleteedge}
  \label{alg:deleteuptalg}
  \KwIn{$G = (V, E)$, ego-betweenness array $C_B$, a deleted edge $(u, v)$.}
  \KwOut{the updated $C_B$.}
  ${\cal S} \leftarrow \uptlocalsmap(G, (u, v))$\;
  $L \leftarrow N(u) \cap N(v)$\;
  \For{$(x, y) \in S_{u}$}
  {
    \If{$x = v$ or $y = v$}
    {
        $C_{B}(u) \leftarrow C_{B}(u) - 1/(S_{u}(x, y)+1)$\;
    }
    \Else
    {
        $C_{B}(u) \leftarrow C_{B}(u) - 1/(S_{u}(x, y)+1) + 1/S_{u}(x, y)$\;
    }
  }
  Update $C_{B}(v)$ as lines 3-7\;
  \For{$x \in L$}
  {
    \For{$(y, z) \in S_{x}$}
    {
        \If{($y = u$ and $z = v$) or ($y = v$ and $z = u$)}
        {
            $C_{B}(x) \leftarrow C_{B}(x) + 1/(S_{x}(y, z)+1)$\;
        }
        \Else
        {
            $C_{B}(x) \leftarrow C_{B}(x) - 1/(S_{x}(y, z)+1) + 1/S_{x}(y, z)$\;
        }
    }
  }
  Delete $(u, v)$ from $G$\;
  {\bf return} $C_{B}$\;
\end{algorithm}
}

Here we consider the case of deleting an edge $(u, v)$ from $G$. When $(u, v)$ is deleted, only the vertices in $L \cup \{u, v\}$ need to update their ego-betweennesses according to Observation~\ref{obs:insertVertex}. Below we introduce the update rules for $u$, $v$ and $w \in L$. Since the proofs of the following lammas are similar to that of \lamref{obs:insertuv} and \lamref{obs:insertcnbr}, we omit them due to the space limitation.


\begin{Lemma}
\label{obs:deleteuv}
Consider a deleted edge $(u, v)$, the updated ego-betweenness of $u$ is: $C_B(u) = C_B(u) + \sum_{x, y \in L, (x, y) \notin E} (1/S_u(x, y) - 1/(S_u(x, y)+1)) - \sum_{x \in N(u), x \notin L} 1/(S_u(v, x)+1) $. The case of updating $C_B(v)$ is similar.
\end{Lemma}

\comment{
\begin{proof}
For vertex $u$, after inserting an edge $(u, v)$ into $G$, $v$ is a new neighbor and is added into $G_{Ego}(u)$. For $x \in L$, $x$ and $v$ is linked, thus it does not make any contribution to $C_B(u)$. For $x \notin L$, the $(v, x)$ is a new vertex pair which makes $C_B(u)$ increase, thus we need to recompute $S_u(v, x)$ and update $C_B(u)$ as: $C_B(u) = C_B(u) + 1/(S_u(v, x)+1)$. In addition, for $x, y \in L$ and $(x, y) \notin E$, $C_B(u)$ has included the contribution of pair $(x, y)$, thus we should subtract this part and plus the newly contribution to $C_B(u)$. $v$ is a new vertex that connects $x$ and $y$, the number of the shortest paths between $x$ and $y$ only adds 1, thus we can calculate $S_u(x, y)$ and reveal the previous contribution for updating $C_B(u)$, i.e., $C_B(u) = C_B(u) + 1/(S_u(x, y)+1)-1/S_u(x, y)$.
For vertex $u$, $v$ is not a neighbor, thus $v$ is removed from $G_{Ego}(u)$ and there are some pairs deleted. For $w \notin L$, the pair $(v, w)$ delete, compute and delete $CB = CB - 1/{(b_{wx})}$. In addition, for $w \in L$ and $x \in L$ and $w, x$ are the neighbors of $v$ in $G_{Ego}(u)$, if there is no edge between $w$ and $x$, compute the before number of the shortest path between $w$ and $x$, $b_{wx}$, $CB = CB - 1/{(b_{wx})} + 1/(b_{wx}-1)$. can get bound by degree for $u$ and $v$.
\end{proof}
}

\begin{Lemma}
\label{obs:deletecnbr}
Consider a deleted edge $(u, v)$, the updated ego-betweenness of $w \in L$ is: $C_B(w) = C_B(w) + 1/(S_w(u, v)+1) + \sum_{x \in {N(w) \cap N(u)-\{v\}}, (x, v) \notin E} (1/S_w(x, v) - 1/(S_w(x, v)+1)) + \sum_{x \in {N(w) \cap N(v)-\{u\}}, (x, u) \notin E} (1/S_w(x, u) - 1/(S_w(x, u)+1))$.
\end{Lemma}

\comment{
\begin{proof}
For vertex $w \in L$, the insertion of $(u, v)$ causes the direct connection between $u$ and $v$ in $G_{Ego}(w)$ which makes $C_B(u)$ decrease. We need to compute $S_w(u, v)$ before the insert operation and update $C_B(u)$ as: $C_B(w) = C_B(w) - 1/(S_w(u, v)+1)$. In addition, for $x \in N(w) \cap N(v)$ and $(u, x) \notin E$, $v$ is a new vertex that connects $u$ and $x$, thus we calculate $S_w(u, x)$ and update $C_B(u)$ as: $C_B(u) = C_B(u) + 1/(S_w(u, x)+1)-1/S_w(u, x)$. Analogously, for $x \in N(w) \cap N(u)$ and $(v, x) \notin E$, $u$ is a new vertex that connects $v$ and $x$, we update $C_B(w)$ as the above operation.
Consider $w \in L$. After deleting an edge $(u, v)$, $u$ and $v$ are not connected in $G_{Ego}(w)$, thus compute after the shortest path between $u$ and $v$, $C_B(w) =  C_B(w) + 1/b_{uv}$. $x$ is $u$'s neighbor but not $v$'s neighbor, compute before $b_{xv}$, now $CB = CB - 1/{(b_{yu})} + 1/(b_{yu}-1)$. $y$ is $v$'s neighbor but not $u$'s neighbor, compute before $b_{yu}$, $CB = CB - 1/{(b_{yu})} + 1/(b_{yu}-1)$. increase.
\end{proof}
}

Note that $S_u(x, y)$ in \lamref{obs:deleteuv} and \lamref{obs:deletecnbr} represents the value before deleting the edge. In particular, $S_w(u, v)$ in \lamref{obs:deletecnbr} is the value after the deleting update. Based on these lemmas, we present a local update algorithm, called \deleteedge, to maintain the ego-betweennesses when an edge $(u, v)$ is deleted. The framework of \deleteedge is similar to that of \insertedge. We only need to make the following minor changes. For $u$ and $v$, \deleteedge modifies line 6 and line 8 of Algorithm~\ref{alg:insertuptalg} to $C_{B}(u) \leftarrow C_{B}(u) - 1/(S_{u}(x, y)+1)$ and $C_{B}(u) \leftarrow C_{B}(u)+ 1/S_{u}(x, y) - 1/(S_{u}(x, y)+1)$ based on \lamref{obs:deleteuv}. According to \lamref{obs:deletecnbr}, \deleteedge calculates $C_B(x)$ for a common neighbor $x$ as $C_{B}(x) \leftarrow C_{B}(x) + 1/(S_{x}(y, z)+1)$ and $C_{B}(x) \leftarrow C_{B}(x)+ 1/S_{x}(y, z) - 1/(S_{x}(y, z)+1) $ corresponding to line 13 and line 15 of Algorithm~\ref{alg:insertuptalg}. Note that \deleteedge first performs \uptlocalsmap (Algorithm~\ref{alg:uptlocalsmapalg}) before deleting $(u, v)$ and then updates the ego-betweennesses of the affected vertices. Finally, it removes $(u, v)$ from $G$ and terminates. We omit the pseudo-code of \deleteedge due to the space limit.

\begin{figure}[t]\vspace*{-0.2cm}
\centering
  \subfigure[$G_{E(c)}$]{
  \label{fig:deleteegoc}
  \begin{minipage}{2.5cm}
  \centering
  \includegraphics[width=\textwidth]{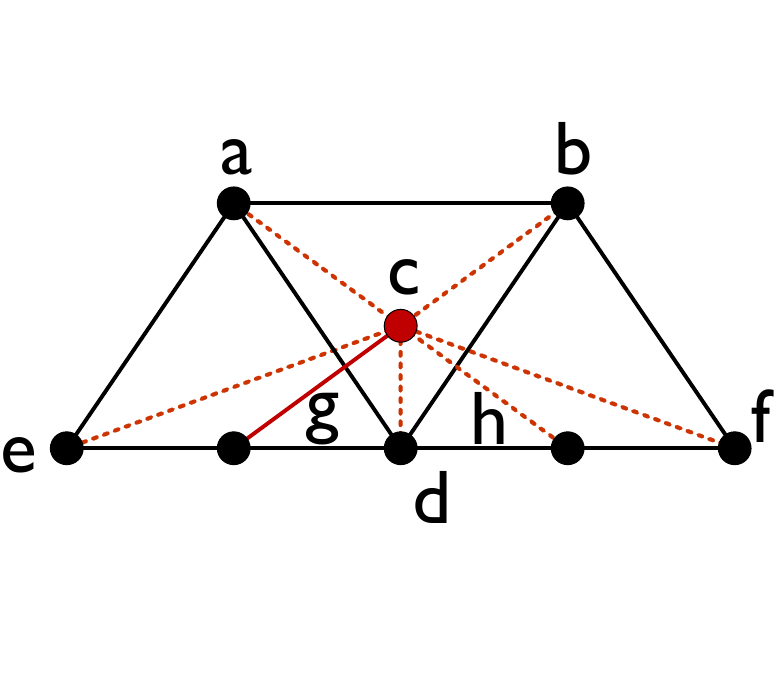}
  \end{minipage}
  }
  \subfigure[$G_{E(g)}$]{
  \label{fig:deleteegog}
  \begin{minipage}{2.5cm}
  \centering
  \includegraphics[width=\textwidth]{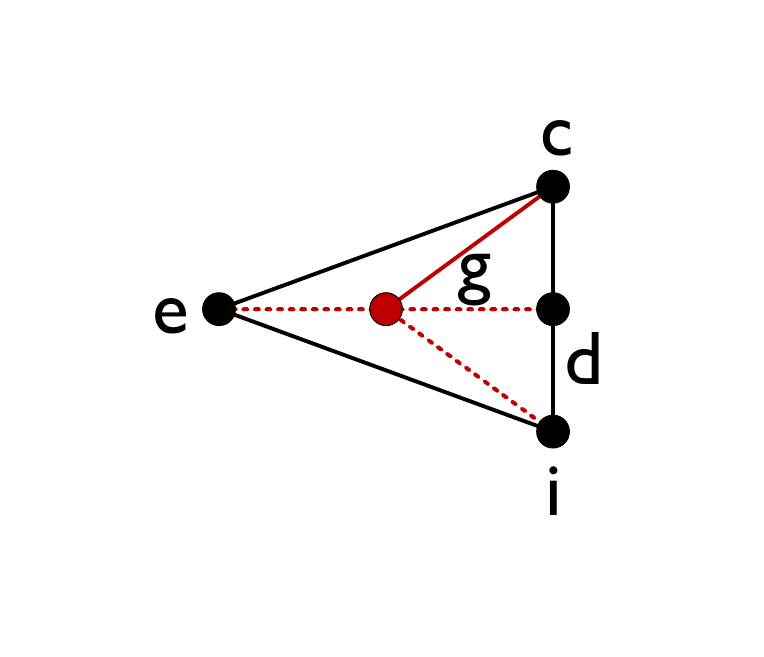}
  \end{minipage}
  }
  \subfigure[$G_{E(e)}$]{
  \label{fig:deleteegoe}
  \begin{minipage}{2.5cm}
  \centering
  \includegraphics[width=\textwidth]{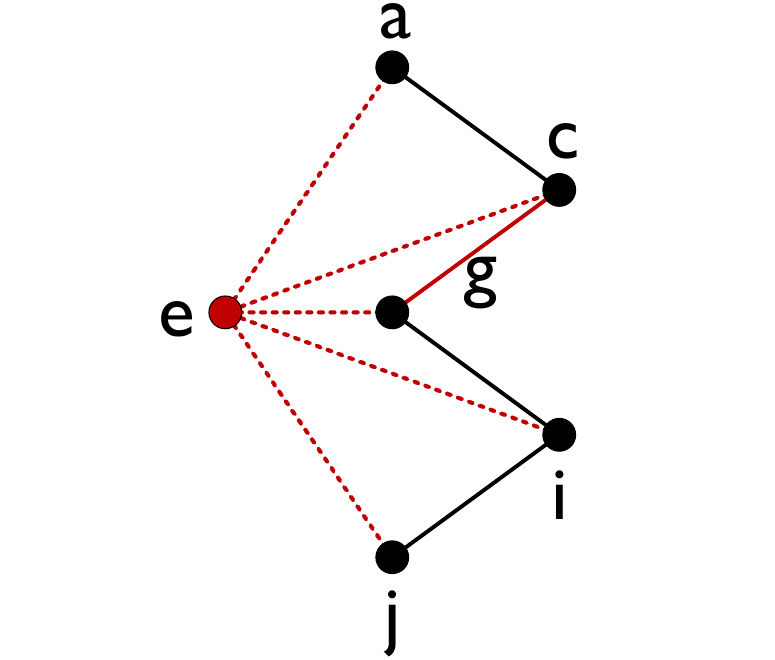}
  \end{minipage}
  }
  \vspace*{-0.2cm}
\caption{Running example}
\label{fig:deleteedge}\vspace*{-0.2cm}
\end{figure}

\begin{example}
Reconsider the graph $G$ in \figref{fig:expgraph}. Suppose that we delete an edge $(c, g)$ from $G$. The ego networks of $c, g$ and their common neighbor change; further their ego-betweennesses need to be updated. For vertices $c$ and $g$, their ego networks are depicted in \figref{fig:deleteegoc} and \figref{fig:deleteegog}. Since $c$ and $g$ are disconnected, the pair $(c, i)$ in \figref{fig:deleteegog} no longer exists and the number of shortest paths for the vertex pair $(e, d)$ changes. According to \lamref{obs:deleteuv}, $C_B(g)$ should be updated as $C_B(g) = C_B(g) - 1/(S_g(c, i)+1) + 1/S_g(e, d) - 1/(S_g(e, d)+1) = 2/3 - 1/3 + 1/2 - 1/3 = 1/2$. Analogously, the $C_B(c)$ changes from $41/6$ to $55/6$ which can be easily checked from \figref{fig:deleteegoc}. For the common neighbor $e$, its ego network $G_{E(e)}$ is shown in \figref{fig:deleteegoe}. After deleting $(c, g)$, $C_B(e)$ is still equal to $4.5$ according to \lamref{obs:deletecnbr}.
\end{example}

\subsection{Updating the top-$k$ results} \label{sec:updatetopk}
Here we present lazy update techniques to maintain the top-$k$ results when the graph is updated. The lazy-update techniques for edge insertion and edge deletion are designed by maintaining a sorted list of the vertices. Specifically, the sorted list, denoted by $H$, contains all vertices in $G$. For each vertex $u$ in $H$, $u$ associates with two variables, namely, $H(u).C_B$ and $H(u).F_G$, which represent the ego-betweenness $C_B(u)$ and the update state of $u$. If $H(u).F_G$ equals true, that means $H(u).C_B$ is not the exact value and should be re-calculated. Otherwise, $H(u).C_B$ is accurate. We calculate $H(u).C_B$ for each vertex $u$ and initialize $H(u).F_G$ as false, and then sort all vertices in non-increasing order of their ego-betweennesses to obtain $H$. Equipped with $H$, the lazy update techniques for edge insertion and edge deletion are as follows.

\stitle{Lazy update for edge insertion.} Consider an insertion edge $(u, v)$ and a common neighbor $w \in N(u) \cap N(v)$. The calculations of $C_B(u)$, $C_B(v)$ and $C_B(w)$ are described in \lamref{obs:insertuv} and \lamref{obs:insertcnbr}, respectively. Obviously, $S_{*}(*, *)+1 > S_{*}(*, *)$ holds, thus we have $1/(S_{u*}(*, *)+1) < 1/S_{*}(*, *)$. For vertex $w$, the parts $\sum_{x \in N(w) \cap N(u)-\{v\}, (x, v) \notin E} (1/(S_w(x, v)+1)-1/S_w(x, v))$ and $\sum_{x \in N(w) \cap N(v)-\{u\}, (x, u) \notin E} (1/(S_w(x, u)+1)-1/S_w(x, u))$ are both less than $0$, and $1/(S_w(u, v)+1)$ is subtracted from $C_B(w)$, thus $C_B(w)$ tends to decrease. However, for vertex $u$ (as well as $v$), the part $\sum_{x \in N(u), x \notin L} 1/(S_u(v, x)+1)$ increases, but the part $\sum_{x, y \in L, (x, y) \notin E} (1/(S_u(x, y)+1)-1/S_u(x, y))$ decreases, thus the changes of $C_B(u)$ and $C_B(v)$ are unclear. Nevertheless, an interesting finding is that with the insertion operation, the degrees of $u$ and $v$ increase and the upper bounds of $C_B(u)$ and $C_B(v)$ also increase. Based on these findings, we can implement a lazy update rule to maintain the top-$k$ results for edge insertion.

\begin{algorithm}[t]
  \scriptsize
  \LinesNumbered
  \caption{\insertedgetop}
  \label{alg:insertupttopalg}
  \KwIn{$G = (V, E)$, $H$, an inserted edge $(u, v)$, top-$k$ result set $R$.}
  \KwOut{the updated $R$.}
  \If{$u \in R$}
  {
    Compute $H(u).C_B$; $H(u).F_G \leftarrow false$\;
    \If{$H(u).C_B < \min_{p \in R \backslash \{u\}}C_{B}(p)$}
    {
        \While{$true$}
        {
            $y \leftarrow \arg\max_{p \in H-R}H(p).C_B$\;
            \If{$H(y).F_G = false$ and $H(u).C_B < H(y).C_B$}
            {
                $R \leftarrow (R - \{ u \}) \cup \{ y \}$; {\bf break}\;
            }
            {\bf {else}} Compute $H(y).C_B$; $H(y).F_G \leftarrow false$\;
        }
    }
  }
  \Else{
    $\bound(u) \leftarrow \frac{d(u)*(d(u)-1)}{2}$\;
    \If{$\bound(u) > \min_{p \in R}C_{B}(p)$}
    {
        Compute $H(u).C_B$; $H(u).F_G \leftarrow false$\;
        \If{$H(u).C_B > \min_{p \in R}C_{B}(p)$}
        {
            $y \leftarrow \arg\min_{p \in R}C_{B}(p)$\;
            $R \leftarrow (R - \{ y \}) \cup \{ u \}$\;
        }
    }
    {\bf {else}} $H(u).F_G \leftarrow true$\;
  }
  Update $H$ and $R$ according to $v$ as lines 1-16\;
  $L \leftarrow N(u) \cap N(v)$\;
  \For{$x \in L$}
  {
      {\bf {if}} $x \in R$ {\bf {then}} Update $H$ and $R$ according to $x$ as lines 2-8\;
      \comment{
      \If{$x \in R$}
      {
        Update $H$ and $R$ according to $x$ as lines 2-8\;
        \comment{
        Calculate $C_{B}(x)$\;
        $H(x).cb = C_{B}(x)$; $H(x).flag = false$\;
        \If{$C_{B}(x) < \min_{p \in R/\{x\}}C_{B}(p)$}
        {
            \For{$y \in H - R$}
            {
                \If{$H(y).flag = false$}
                {
                    $R \leftarrow R - \{x\}$; $R \leftarrow R \cup \{y\}$\;
                    {\bf break}\;
                }
                \Else{
                    Calculate $C_{B}(y)$\;
                    $H(y).cb = C_{B}(y)$\;
                    $H(y).flag = false$\;
                }
            }
        }
      }
      }
      }
      {\bf {else}} $H(x).F_G \leftarrow true$\;
  }
  {\bf return} $R$\;
\end{algorithm}

The lazy update algorithm to handle edge insertion, called \insertedgetop, is shown in Algorithm~\ref{alg:insertupttopalg}. For the endpoint $u$ of the inserted edge, \insertedgetop first identifies whether $u$ is included in the top-$k$ result set $R$. If $u \in R$, it calculates the ego-betweenness $H(u).C_B$ and sets $H(u).F_G$ to false to indicate the correctness of $H(u).C_B$. As $H(u).C_B$ is updated, we need to determine whether $u$ still belongs to $R$. If $H(u).C_B >= \min_{p \in R \backslash \{u\}}C_{B}(p)$ holds, $u$ is still included in the top-$k$ result set $R$. On the other hand, \insertedgetop compares $H(u).C_B$ with the ego-betweenness of the $(k+1)$-th element in the sorted list $H$ (lines 4-8). Let $y$ denote the $(k+1)$-th vertex in $H$. If $H(y).F_G$ is false, that means $y$ is the vertex with the highest ego-betweenness that is not contained in $R$. \insertedgetop compares $H(u).C_B$ with $H(y).C_B$ and maintains $R$ (lines 6-7). If $H(y).F_G = true$ holds, \insertedgetop computes $H(y).C_B$ and updates $H(y).F_G$, and then performs the next loop (line 8). While $u \notin R$, we derive the new upper bound $\bound(u)$ of $H(u).C_B$ to determine whether $H(u).C_B$ needs to be computed exactly (lines 10-16). If $\bound(u) <= \min_{p \in R}C_{B}(p)$ holds, it means that $H(u).C_B$ is not greater than $\min_{p \in R}C_{B}(p)$, thus $u$ is still not an answer of the top-$k$ results and \insertedgetop can avoid calculating the correct $H(u).C_B$ and only updates $H(u).F_G$ to true (line 16). Otherwise, \insertedgetop calculates $H(u).C_B$ and identifies whether $u$ should be inserted into $R$ (lines 12-15). Likewise, we perform the same operation for the other endpoint $v$ (line 17). Then, \insertedgetop handles the common neighbors of $u$ and $v$ (lines 18-21). For vertex $x \in L$, the algorithm judges whether $x$ is included in $R$. If yes, it updates $H$ and $R$ as the operations of $u$ (line 20). On the other hand, because $C_B(x)$ is decreasing, $x$ is still not in the top-$k$ result set and thus \insertedgetop avoids computing the exact $H(x).C_B$ and only sets $H(x).F_G$ to true (line 21). Note that the ego-betweennesses of the vertices in $R$ are always correct. Finally, \insertedgetop returns the top-$k$ vertices with the highest ego-betweennesses correctly.

\begin{example}
Reconsider the graph $G$ in \figref{fig:expgraph}. Before inserting the edge $(i, k)$ into $G$, we have $C_B(i) = 8$ and $C_B(k) = 1$. After the insertion, $C_B(i)$ is equal to $10.5$ and $C_B(k)$ is $0.5$. For the common neighbor $f$, $C_B(f)$ decreases from $11$ to $9.5$. Clearly, the change of the ego-betweennesses for the ends of the insertion edge is uncertain while it is decreasing for the common neighbors. Suppose that $k = 1$ and the current result set is $R = \{f\}$. For vertex $k$, it is not included in $R$ and its new bound is $(3*2)/2 = 3 < C_B(f) = 11$, thus the calculation of $C_B(k)$ can be skipped and we only set $H(k).FG$ to true. For vertex $i$, the new bound is $(7*6)/2 = 21 > C_B(f) = 11$, thus we need to calculate the new $C_B(i) = 10.5$ and update $R = \{i\}$. In this case, consider the common neighbor $f$, it is not included in $R$. Since $C_B(f)$ is not incremental, it definitely not in the top-1 result after inserting $(i, k)$, thus we can avoid calculating $C_B(f)$ and updating the results $R$.
\end{example}

\comment{
\begin{algorithm}[t]
  \caption{\deleteedgetop}
  \label{alg:deleteupttopalg}
  \KwIn{$G = (V, E)$, $H$, a deleted edge $(u, v)$, top-$k$ result set $R$.}
  \KwOut{the updated $R$.}
  \If{$u \in R$}
  {
    Update $H$ and $R$ as lines 2-8 of Algorithm~\ref{alg:insertupttopalg}\;
  }
  \Else{
    $\overline {bound}(u) \leftarrow \lfloor \frac{d(u)*(d(u)-1)}{2} \rfloor$\;
    \If{$\overline {bound}(u) > \min_{p \in R}C_{B}(p)$}
    {
        Compute $H(u).C_B$; $H(u).F_G \leftarrow false$\;
        \While{$H(\arg\min_{p \in R}C_{B}(p)).F_G = true$}
        {
            $y \leftarrow \arg\min_{p \in R}C_{B}(p)$\;
            Compute $H(y).C_B$; $H(y).F_G \leftarrow false$\;
        }
        \If{$H(u).C_B > \arg\min_{p \in R}C_{B}(p)$}
        {
            $y \leftarrow \arg\min_{p \in R}C_{B}(p)$\;
            $R \leftarrow (R - \{ y \}) \cup \{ u \}$\;
        }
    }
    \Else{
        $H(u).F_G \leftarrow true$\;
    }
  }
  Update $H$ and $R$ according to $v$ as lines 1-14\;
  $L \leftarrow N(u) \cap N(v)$\;
  \For{$x \in L$}
  {
      {\bf {if}} $x \notin R$ {\bf {then}} Update $H$ and $R$ according to $x$ as lines 12-20\;
      \If{$x \notin R$}
      {
        Update $H$ and $R$ according to $x$ as lines 12-20\;
      }
      {\bf {else}} $H(x).F_G \leftarrow true$\;
      \Else{
        $H(x).F_G \leftarrow true$\;
      }
  }
  {\bf return} $R$\;
\end{algorithm}
}

\stitle{Lazy update for edge deletion.} Consider the deletion edge $(u, v)$ and a common neighbor $w \in N(u) \cap N(v)$. Like the edge insertion, the changes of $C_B(u)$, $C_B(v)$ and $C_B(w)$ are as follows. $C_B(w)$ is definitely non-decreasing while $C_B(u)$ and $C_B(v)$ are uncertain. Fortunately, after deleting $(u, v)$, the degrees of $u$ and $v$ decrease and also the upper bounds of $C_B(u)$ and $C_B(v)$ decrease. Based on this, we can implement a lazy update algorithm which is very similar to edge insertion.

Our lazy update algorithm for handling edge deletion, called \deleteedgetop, can be easily devised by slightly modifying Algorithm~\ref{alg:insertupttopalg}. Like lines 14-15 of Algorithm~\ref{alg:insertupttopalg}, \deleteedgetop needs to find the vertex $y$ with the lowest ego-betweenness in the top-$k$ results. Armed with our lazy update technique, the ego-betweennesses of the vertices in $R$ are not all correct, thus \deleteedgetop must find $y \in R$ with the lowest ego-betweenness and $H(y).F_G=false$. The other steps of \deleteedgetop are similar to those of \insertedgetop. Due to the space limit, the pseudo-code of \deleteedgetop is omitted.

\begin{example}
Let us still consider the graph $G$ in \figref{fig:expgraph}. Before deleting the edge $(c, g)$ from $G$, we have $C_B(c) = 41/6$, $C_B(g) = 2/3$ and $C_B(e) = 9/2$. After the deletion, the new ego-betweennesses for $c$, $g$, and $e$ are $55/6, 1/2, 9/2$, respectively. Obviously, the change of the ego-betweennesses for the ends of the deletion edge is uncertain while it is non-decreasing for the common neighbors. Suppose that $k = 1$ and we can check that the current $R = \{f\}$. For vertex $g$, its new bound is equal to $(3*2)/2 = 3 < C_B(f) = 11$ and $g \notin R$, thus we do not need to calculate the new ego-betweenness for $g$ and only set $H(g).FG$ to true. For vertex $c$, its new bound is $(6*5)/2 = 15 > C_B(f) = 11$, thus we calculate the new $C_B(c) = 55/6$ and the top-1 answer is still $f$. When $k = 12$, the top-$k$ results before deleting the edge $(c, g)$ is the set $V-\{u, v, y, z\}$. In this case, the common neighbor $e$ is included in $R$. Since $C_B(e)$ is non-decreasing after deleting $(c, g)$ , it is definitely still contained in the top-$12$ results, thus we can avoid updating the answer set $R$.
\end{example}

\comment{
\begin{algorithm}[t]
  \scriptsize
  \caption{\vertexparaalg$(G)$}
  \label{alg:enumcoloralg}
  \KwIn{$G = (V, E)$.}
  \KwOut{$C_B(u)$ for all $u \in V$.}
  Construct the oriented graph $G^+ =( V, E^+)$ of $G$\;
  \For{parallel $u \in V$}
  {
    $\vparainneralg(u)$\;
  }
  \For{parallel $u \in V$}
  {
    $C_{B}(u) \leftarrow \lfloor \frac{d(u)*(d(u)-1)}{2} \rfloor$\;
    Compute $C_{B}(u)$ as lines 15-17 in Algorithm \ref{alg:bboundalg}\;
  }
  {\bf return} $C_B(u)$ for all $u \in V$\;

  \vspace*{0.1cm}
  {\bf Procedure} $\vparainneralg(u)$\\
  \For{$i \in N(u)$}
  {
    $rd(i) \leftarrow \emptyset$\;
    {\bf {if}} $i \in N^+(u)$ {\bf {do}} $EN \leftarrow EN \cup \{i\}$\;
    {\bf {if}} $i \in N^-(u)$ {\bf {do}} $DN \leftarrow DN \cup \{i\}$\;
  }
  \For{$i \in DN$}
  {
    {\bf {for}} $j \in N(i) \cap N(u)$ {\bf {do}} Initialize $rd(i)$ and $rd(j)$\;
  }
  Lock and update related $S_{*}$ as lines 6-24 in Algorithm \ref{alg:egobwdalg}\;
  \comment{
  \For{$i \in DN$}
  {
    Find $p^*$ as lines 6-10 in Algorithm \ref{alg:egobwdalg}\;
    \If{$V_{is}(p^*) = true$ and $p^* \in EN$}
	{
        Lock and update $S_{u}$, $S_{p^*}$ as line 12 in Algorithm \ref{alg:egobwdalg}\;
    }
  }
  \For{$i \in EN$}
  {
    \For{$j \in EN-\{i\}$ and $(i, j) \in E$}
    {
        Lock and update $S_{u}$, $S_{i}$, $S_{j}$ as lines 12-23 in Algorithm \ref{alg:egobwdalg}\;
        Update $rd(i)$, $rd(j)$ as line 24 in Algorithm \ref{alg:egobwdalg}\;
    }
  }
  }
\end{algorithm}
}

\section{The parallel algorithms} \label{sec:parallelalg}
In this section, we propose parallel ego-betweenness algorithms to improve the scalability of ego-betweenness computation. We first introduce a vertex-based parallel algorithm and then propose an edge-based parallel algorithm to further improve efficiency.

\comment{
\begin{algorithm}[t]
  \caption{\enumegoalg$(G)$}
  \label{alg:enumegoalg}
  \KwIn{$G = (V, E)$.}
  \KwOut{$C_B(u)$ for all $u \in V$.}
  \For{parallel $u \in V$}
  {
    $C_B(u) \leftarrow 0$\;
    \For{$v \in N(u)$}
    {
       \For{$w \in N(u)$}
       {
            \If{$(v, w) \notin E$}
            {
                $C_B(u) \leftarrow C_B(u) + 1/(N(v) \cap N(w) \cap N(u))$\;
            }
       }
    }
  }
  {\bf return} $C_{B}$\;
\end{algorithm}
}

\subsection{A vertex-based parallel algorithm}
The ego-betweenness of a vertex is defined on its ego network which can be calculated independently, thus a straightforward parallel solution is to process each vertex in parallel. However, such a simple solution may be inefficient, especially for large graphs. When processing each vertex independently, we need to construct its ego network and explore the \emph{diamond} structures (a diamond denotes two triangles that have a common edge), which makes the same \emph{diamond} enumerated multiple times, resulting in repetitive calculations. To solve this problem, we propose a vertex-based parallel algorithm as follows.

As can be seen from Algorithm \ref{alg:bboundalg} and Algorithm \ref{alg:optboundalg}, we explore the \emph{diamond} structures by searching triangles, thus we can employ a parallel triangle enumeration to calculate the ego-betweennesses for all vertices. The main idea is that every triangle in $G$ has a unique orientation based on the total ordering, and only be enumerated when processing the highest-ranked vertex in this triangle. When a triangle $\triangle_{(u, v, w)}$ is found, we utilize it to explore \emph{diamond}s and maintain $S_u$, $S_v$, and $S_w$ which record the number of shortest paths between their neighbors. Note that we should lock the map $S$ when it is updated to ensure the correctness of the parallel algorithm. To avoid frequent locking operations, we employ the idea of Algorithm \ref{alg:optboundalg} to divide the neighbors of a vertex into the in-neighbors and out-neighbors for delaying the updates of $S$, which can also search a triangle once. As all triangles are enumerated, that is, the information of the number of shortest paths is correctly maintained in the maps, we calculate the ego-betweenness for each vertex in parallel according to \lamref{lem:baselemma2}. We refer to this parallel implementation as \vertexparaalg and omit the pseudo-code due to the space limit.

\subsection{An edge-based parallel algorithm}
In practice, \vertexparaalg might still be inefficient, because the out-degrees of the vertices typically exhibit a skew distribution, resulting in the workloads of different threads are unbalanced. A better solution is to enumerate triangles for each directed edge in parallel. This is because the distribution of the number of common outgoing neighbors of the directed edges is typically not very skew, thus improving the parallelism of the algorithm. We refer to such an edge-parallel algorithm as \edgeparaalg. In the experiments, we will compare the efficiency of \vertexparaalg and \edgeparaalg.

\comment{
\begin{algorithm}[t]
  \caption{\edgeparaalg$(G)$}
  \label{alg:enumsmapalg}
  \KwIn{$G = (V, E)$.}
  \KwOut{$C_B(v)$ for all $v \in V$.}
  \For{parallel $u \in V$}
  {
    Initialize an array $B$ with $B(i) = false, 0 \le i < n$\;
    {\bf {for}} $v \in N(u)$ {\bf {do}} $B(v) \leftarrow true$\;
    \For{$v \in N(u)$}
    {
        {\bf {if}} $v > u$ {\bf {then break}}\;
        $L \leftarrow \emptyset$\;
        \For{$w \in N(v)$}
        {
            \If{$B(w) = true$}
		    {
                $L \leftarrow L \cup \{w\}$\;
                \If{$w < u$ and $w < v$}
                {
                    Lock and update $S_{u}$, $S_{v}$, $S_{w}$ as lines 12-14 in Algorithm \ref{alg:bboundalg}\;
                }
	        }
        }
        \For{$p \in L$}
        {
            \For{$q \in L$}
		    {

                \If{$(p, q) \notin E$ and $p \neq q$}
                {
                    Lock and update $S_{u}$, $S_{v}$ as lines 12-14 in Algorithm \ref{alg:bboundalg}\;
                }
	        }
        }
    }
  }
  {\bf return} $C_B(u)$ for all $u \in V$\;
\end{algorithm}
}

\section{Experiments} \label{sec:experiments}

\begin{table}[t!]\vspace*{-0.2cm}
    \scriptsize
	\centering
	\caption{Datasets}
    \label{tab:datasets}
	\vspace*{-0.2cm}
	\setlength{\tabcolsep}{1mm}{
		\begin{tabular}{c|c|c|c|c}
			\hline
			Dataset     & $n$       & $m$           & $d_{\max}$ & Desecription\\ \hline \hline
            \youtube     &1,134,890	 &2,987,624	     &28,754		 & Social network\\
            \wikitalk   &2,394,385	 &4,659,565	     &100,029		 & Communication network\\
            \dblp        &1,843,617	 &8,350,260      &2,213		   & Collaboration network\\
            \pokec       &1,632,803  &22,301,964       &14,854		& Social network\\
            \livejournal &3,997,962	 &34,681,189	 &14,815		 & Social network\\
			\hline
		\end{tabular}
	}
\vspace*{-0.2cm}
\end{table}

\begin{figure*}[t!]\vspace*{-0.2cm}
\centering
  \subfigure[\youtube (vary $k$)]{
  \label{fig:exp-cmpbaseopt-time-youtube}
  \begin{minipage}{3.2cm}
  \centering
  \includegraphics[width=\textwidth]{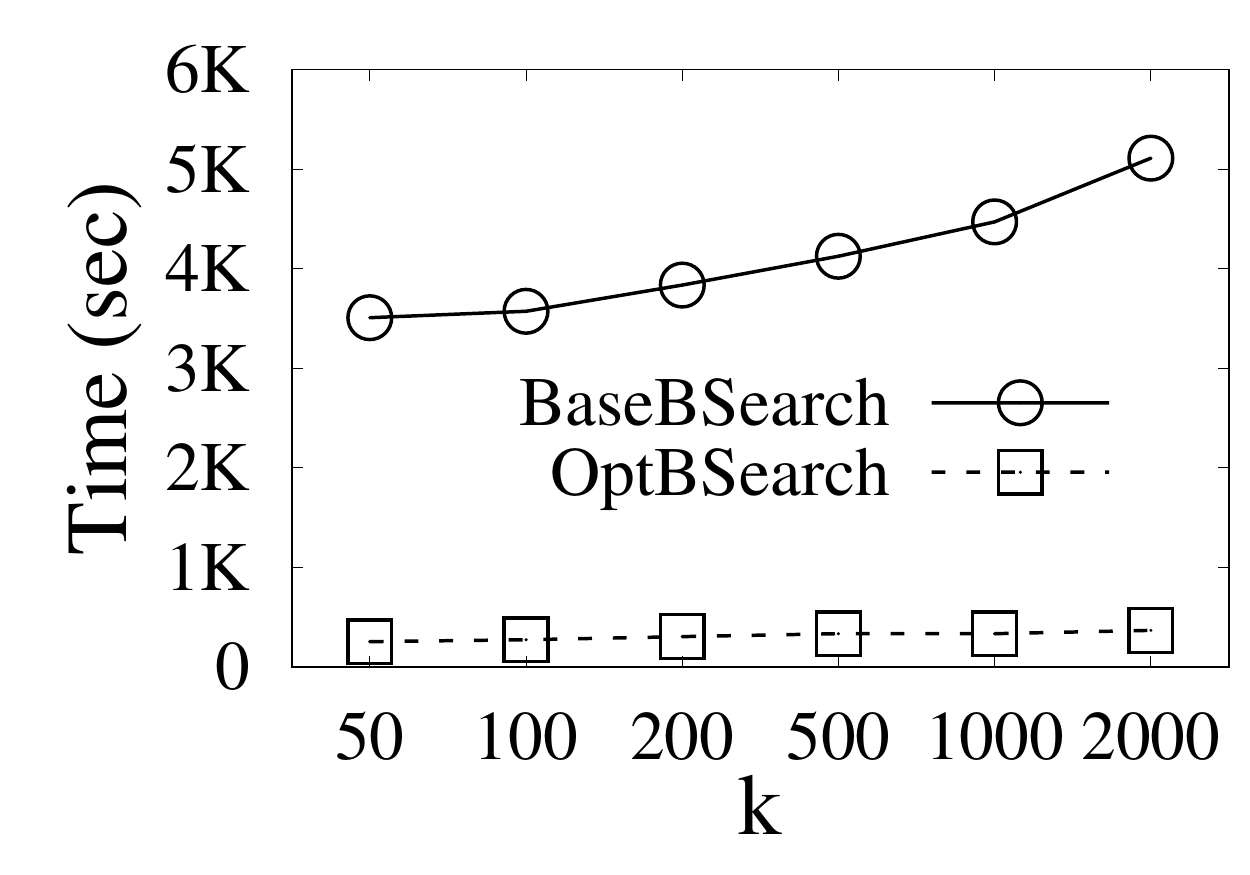}
  \end{minipage}
  }
  \subfigure[\wikitalk (vary $k$)]{
  \label{fig:exp-cmpbaseopt-time-wikitalk}
  \begin{minipage}{3.2cm}
  \centering
  \includegraphics[width=\textwidth]{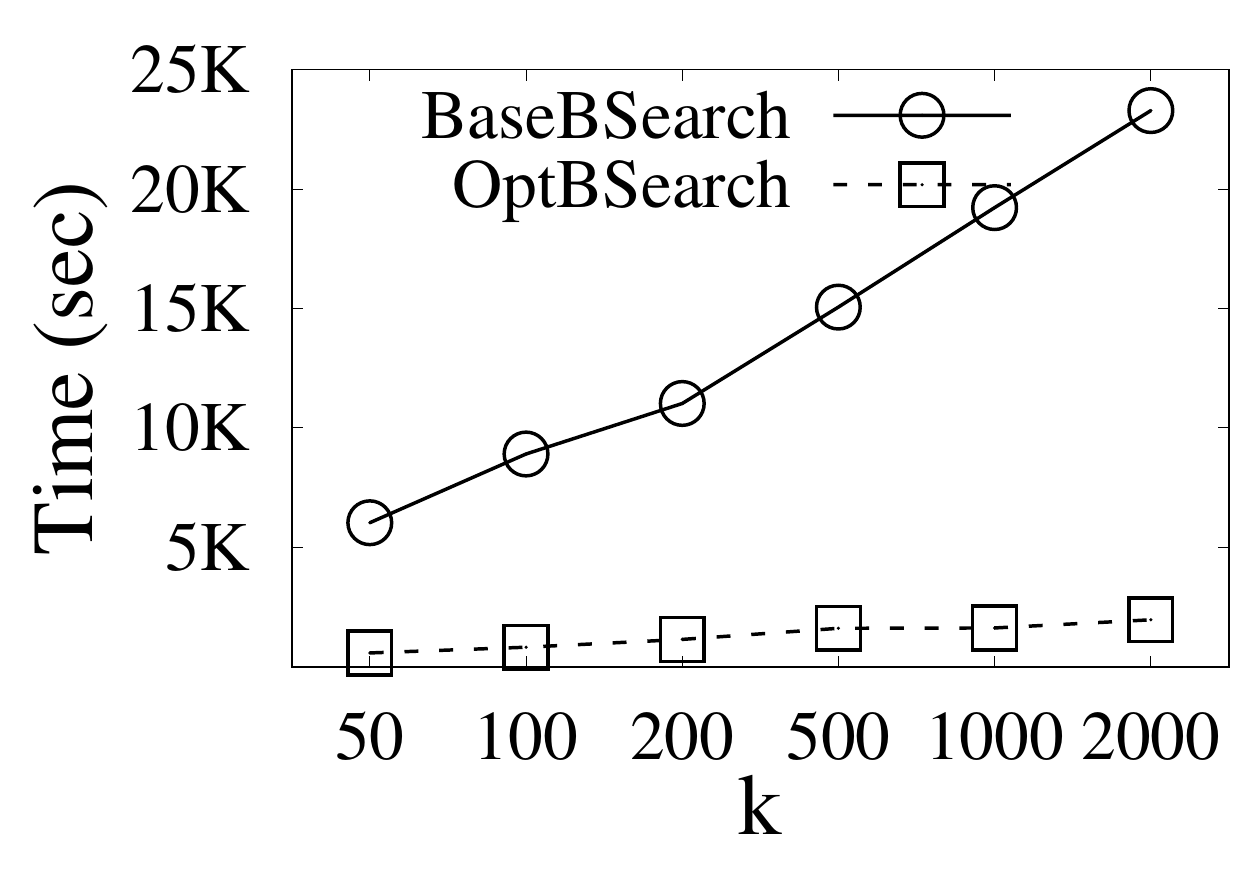}
  \end{minipage}
  }
  \subfigure[\dblp (vary $k$)]{
  \label{fig:exp-cmpbaseopt-time-dblp}
  \begin{minipage}{3.2cm}
  \centering
  \includegraphics[width=\textwidth]{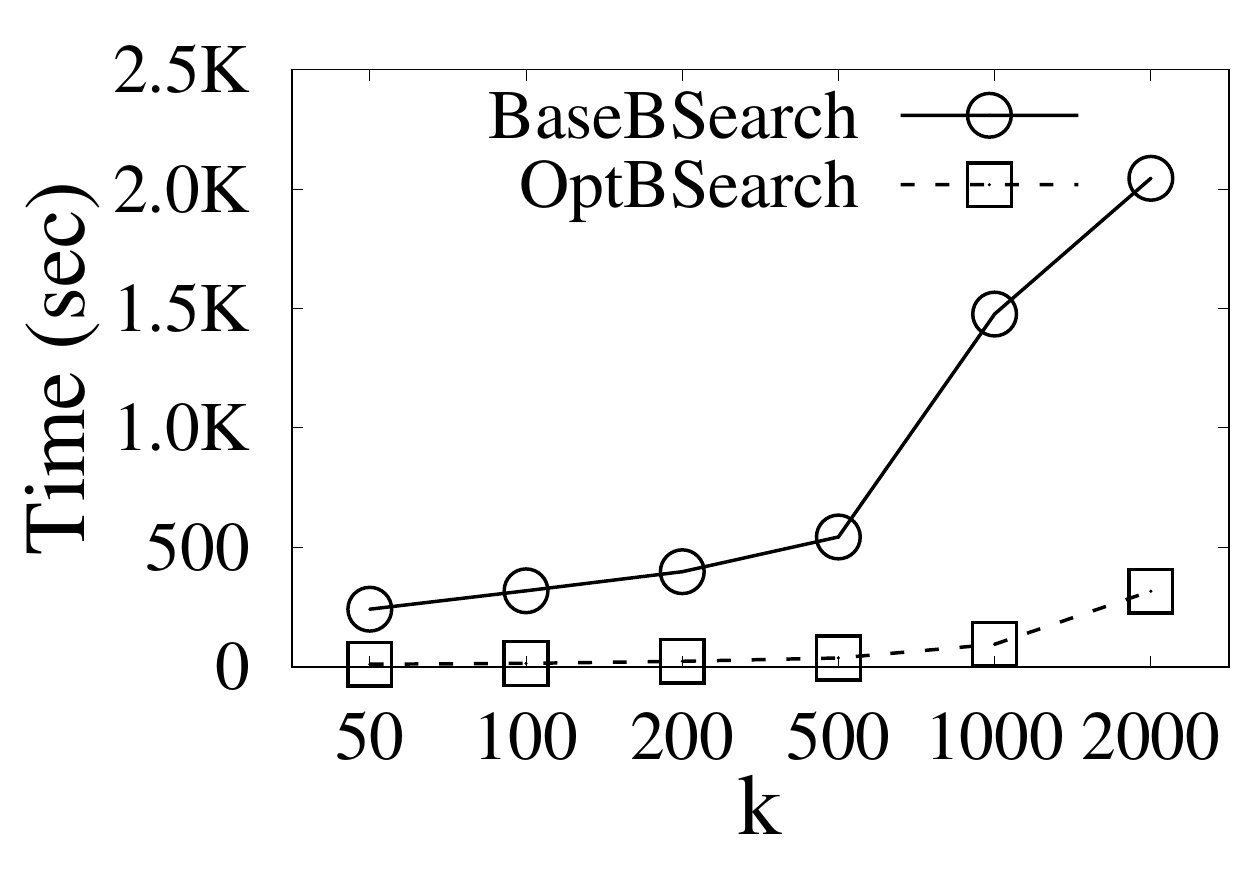}
  \end{minipage}
  }
  \subfigure[\pokec (vary $k$)]{
  \label{fig:exp-cmpbaseopt-time-pokec}
  \begin{minipage}{3.2cm}
  \centering
  \includegraphics[width=\textwidth]{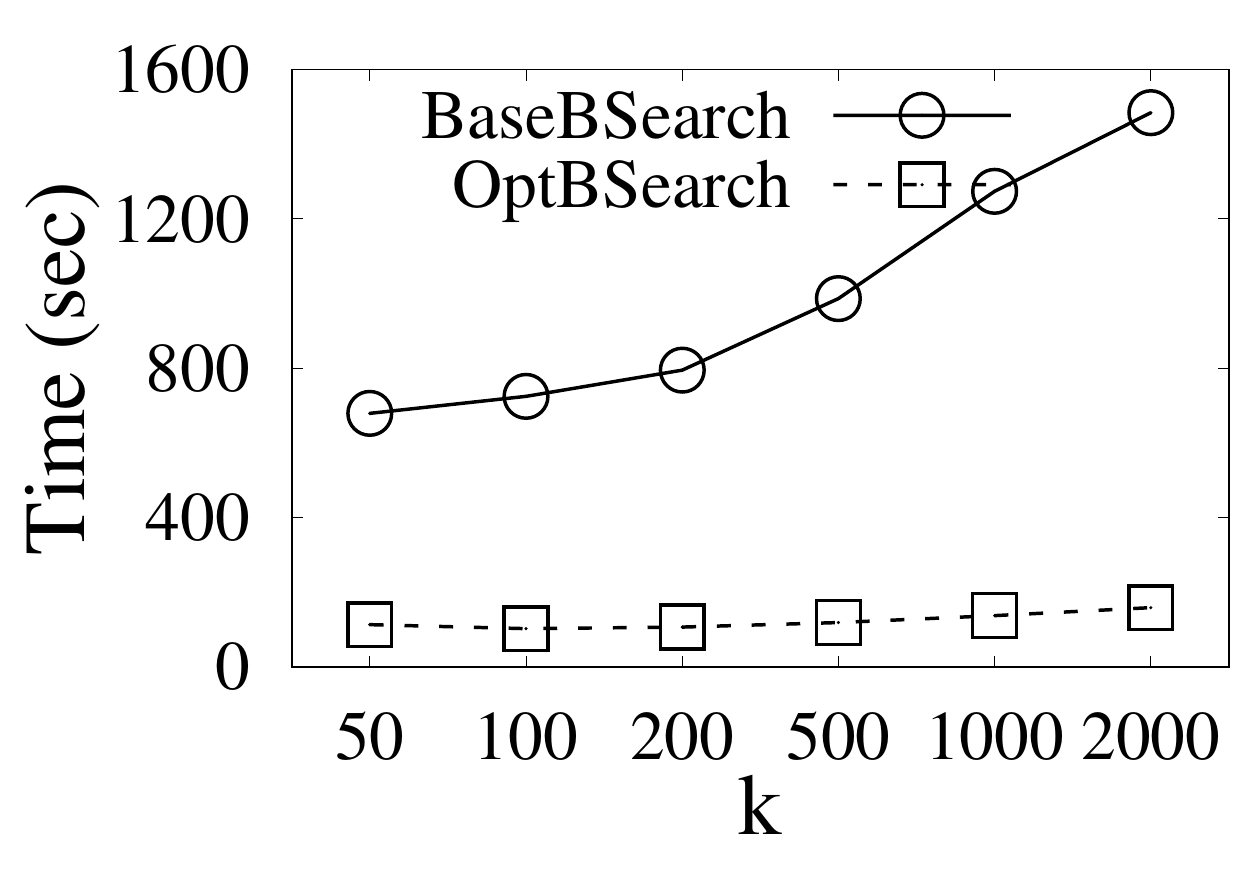}
  \end{minipage}
  }
  \subfigure[\livejournal (vary $k$)]{
  \label{fig:exp-cmpbaseopt-time-livejournal}
  \begin{minipage}{3.2cm}
  \centering
  \includegraphics[width=\textwidth]{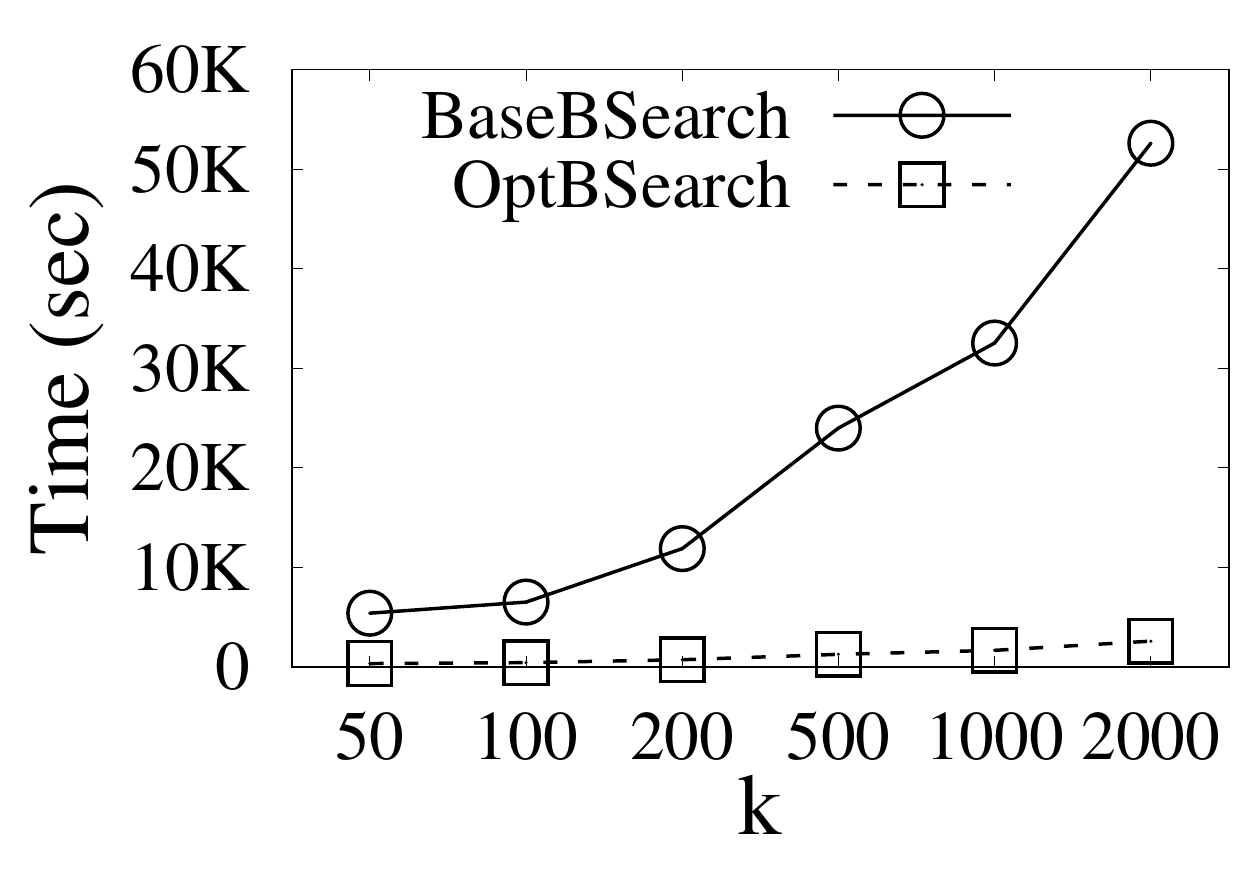}
  \end{minipage}
  }
  \vspace*{-0.2cm}
\caption{Comparisons of \baseboundalg and \optboundalg on various datasets}
\label{fig:exp-cmpbaseopt-time-varyk}\vspace*{-0.2cm}
\end{figure*}

In this section, we conduct extensive experiments to evaluate the efficiency and effectiveness of the proposed algorithms. We implement two top-$k$ ego-betweenness search algorithms, namely, \baseboundalg and \optboundalg (Algorithm \ref{alg:bboundalg} and Algorithm \ref{alg:optboundalg}). To maintain the ego-betweennesses for all vertices, we implement \insertedge (Algorithm~\ref{alg:insertuptalg}) and \deleteedge for handling edge insertion and edge deletion, respectively. We also implement \insertedgetop (Algorithm \ref{alg:insertupttopalg}) and \deleteedgetop to maintain the top-$k$ results for edge insertion and edge deletion, respectively. In addition, we implement two parallel algorithms, \vertexparaalg and \edgeparaalg, to calculate ego-betweennesses of all vertices using OpenMP. All algorithms are implemented in C++. All experiments are conducted on a PC with 2.10GHz CPU and 256GB memory running Red Hat 4.8.5. 

\stitle{Datasets.} We use 5 different types of real-life networks in the experiments, including social networks, communication networks and collaboration networks. The detailed statistics of the datasets are summarized in \tabref{tab:datasets}. In \tabref{tab:datasets}, $d_{\max}$ denotes the maximum degree of the graph.  All these datasets are downloaded from \url{snap.stanford.edu}. 



\stitle{Parameters.} The parameter $k$ in our algorithms is chosen from the set $\{50, 100, 200, 500, 1000, 2000\}$ with a default value of $k = 500$. The parameter $\theta$ in \optboundalg is selected from the set $\{1.05, 1.10, 1.15, 1.20, 1.25, 1.30\}$ with a default value 1.05. We will study the performance of our algorithms with varying $k$ and $\theta$. Unless otherwise specified, the value of a parameter is set to its default value when varying another parameter.

\subsection{Efficiency testing}\label{subsec:exp-efficiency}

\begin{table}[!t]\vspace*{-0.3cm}
\scriptsize
\centering
\caption{The number of vertices for exact computation}
\label{tab:cptcountcmp}
\vspace*{-0.2cm}
\resizebox{\linewidth}{!}
{
\begin{tabular}{c|c|c|c|c|c|c}
\hline
{ \multirow{2}*{Dataset} }& \multicolumn{2}{c|}{$k = 500$} &\multicolumn{2}{c|}{$k = 1000$} &\multicolumn{2}{c}{$k = 2000$}\\

\cline{2-7}
{}\rule{0pt}{8pt}&\bbsalg&\obsalg&\bbsalg&\obsalg&\bbsalg&\obsalg\\
\hline \hline
{\youtube}&564&\textbf{522}&1143&\textbf{1032}&2324&\textbf{2065}\\
{\wikitalk}&527&\textbf{508}&1052&\textbf{1013}&2098&\textbf{2013}\\
{\dblp}&557&\textbf{550}&1499&\textbf{1160}&3060&\textbf{2491}\\
{\pokec}&567&\textbf{552}&1230&\textbf{1168}&2498&\textbf{2367}\\
{\livejournal}&791&\textbf{615}&1723&\textbf{1282}&3406&\textbf{2413}\\
\hline
\end{tabular}
}
\end{table}

\begin{figure}[t!]
\centering
  \subfigure[\wikitalk (vary $\theta$)]{
  \label{fig:exp-opt-time-wikitalk}
  \begin{minipage}{3.2cm}
  \centering
  \includegraphics[width=\textwidth]{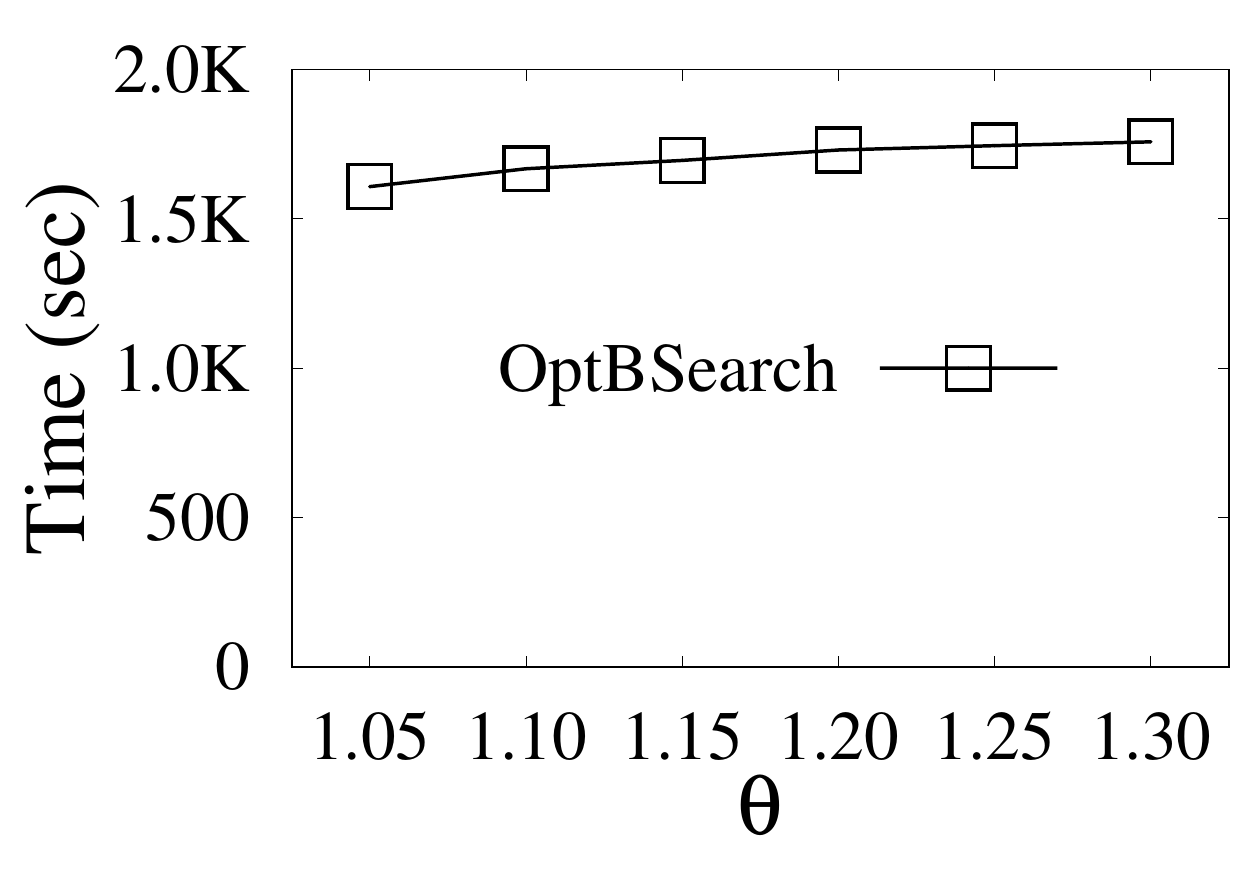}
  \end{minipage}
  }
  \subfigure[\livejournal (vary $\theta$)]{
  \label{fig:exp-opt-time-livejournal}
  \begin{minipage}{3.2cm}
  \centering
  \includegraphics[width=\textwidth]{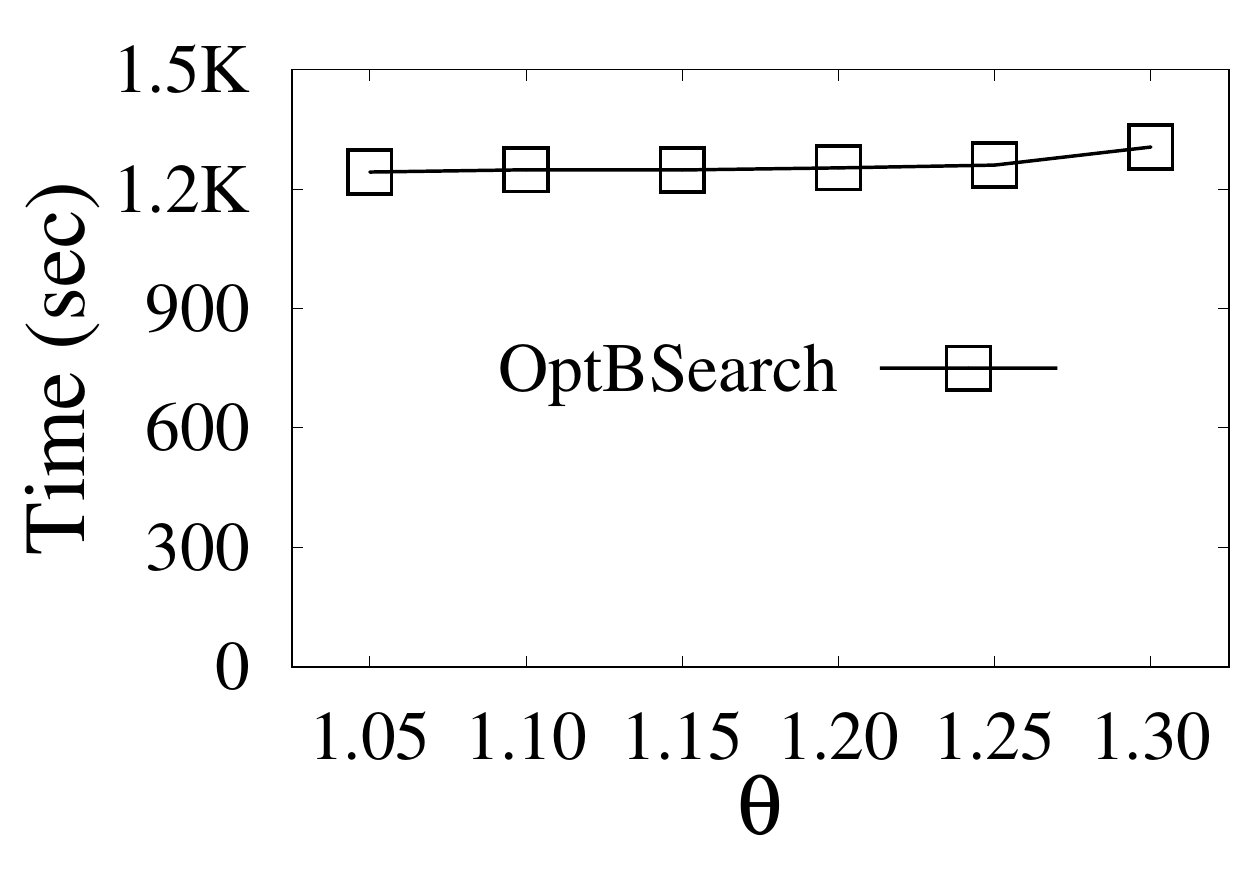}
  \end{minipage}
  }
  \vspace*{-0.2cm}
\caption{Evaluation of \optboundalg with varying $\theta$}
\label{fig:exp-opt-vary-theta}\vspace*{-0.2cm}
\end{figure}

\stitle{Exp-1: Comparison between \baseboundalg and \optboundalg.} \figref{fig:exp-cmpbaseopt-time-varyk} shows the runtime of \baseboundalg and \optboundalg with varying $k$ on all datasets. As expected, the runtime of both \baseboundalg and \optboundalg increases as $k$ increases. As can be seen, \optboundalg is around 6-23 times faster than \baseboundalg with all parameter settings. For example, on \dblp, \optboundalg takes 10.198 seconds, while \baseboundalg consumes 240.482 seconds to retrieve the top-50 results. In the case of $k = 2000$ on \livejournal, \optboundalg takes 2,558.002 seconds, while \baseboundalg consumes 52,599.764 seconds which is roughly 20 times slower than \optboundalg. This is because the dynamic upper bound is tighter than the static upper bound, thus it is more effective to prune the unpromising vertices that are not included in the top-$k$ results. We also record the number of vertices whose ego-betweennesses are computed exactly in \baseboundalg and \optboundalg. For brevity, we refer to \baseboundalg and \optboundalg as \bbsalg and \obsalg. \tabref{tab:cptcountcmp} illustrates the results of $k = 500, 1000, 2000$ on all datasets. Similar results can be observed for other $k$ values. As can be seen, the number of vertices computed by \optboundalg is significantly less than that computed by \baseboundalg on all datasets. For example,  to obtain the top-2000 results on \livejournal, \optboundalg only needs to compute the ego-betweennesses for 2,413 vertices, while \baseboundalg has to compute 3,406 vertices. These results further confirm our theoretical analysis in \secref{sec:onlinealg}.

\begin{figure}[t!]
\centering
\subfigure[edge insertion]{
  \label{fig:exp-insert-large}
  \begin{minipage}{3cm}
  \centering
  \includegraphics[width=\textwidth]{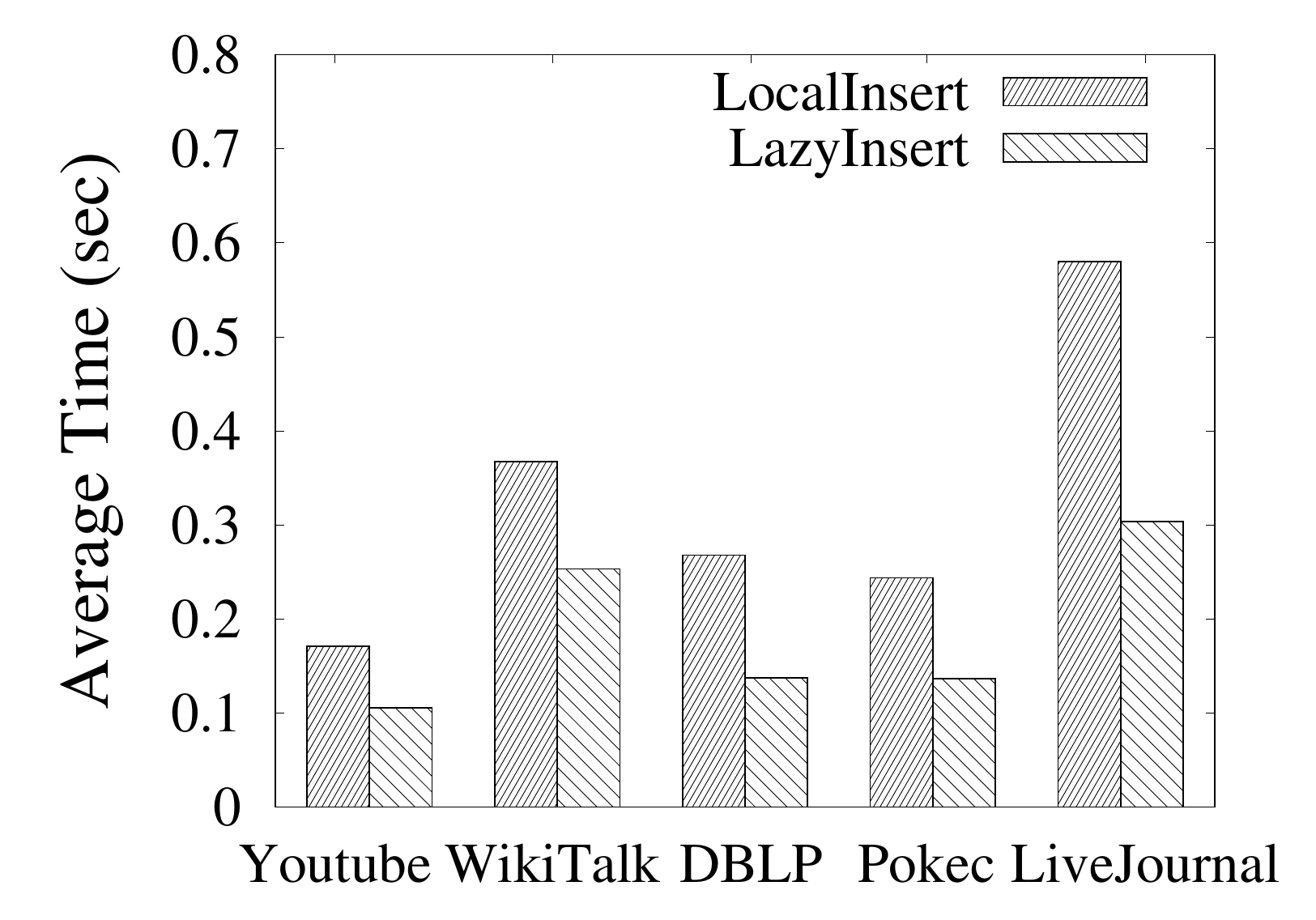}
  \end{minipage}
    }
  \subfigure[edge deletion]{
  \label{fig:exp-delete-large}
  \begin{minipage}{3cm}
  \centering
  \includegraphics[width=\textwidth]{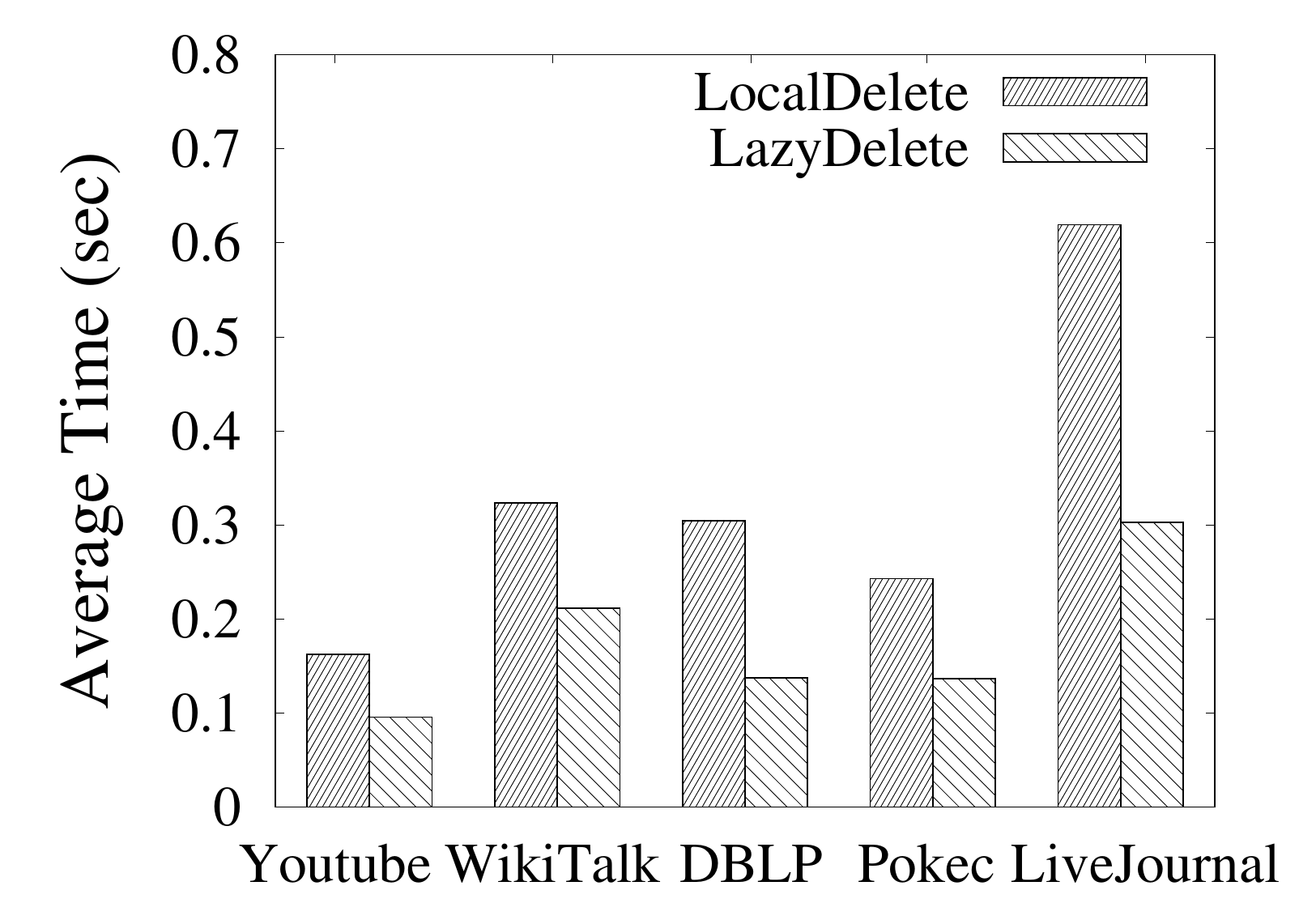}
  \end{minipage}
  }
  \vspace*{-0.2cm}
\caption{Average runtime of the updating algorithms}
\label{fig:exp-update-alldata}\vspace*{-0.08cm}
\end{figure}

\begin{figure}[t!]\vspace*{-0.2cm}
\centering
  \subfigure[\livejournal (vary $m$)]{
  \label{fig:exp-scala-varym-livejournal}
  \begin{minipage}{3.2cm}
  \centering
  \includegraphics[width=\textwidth]{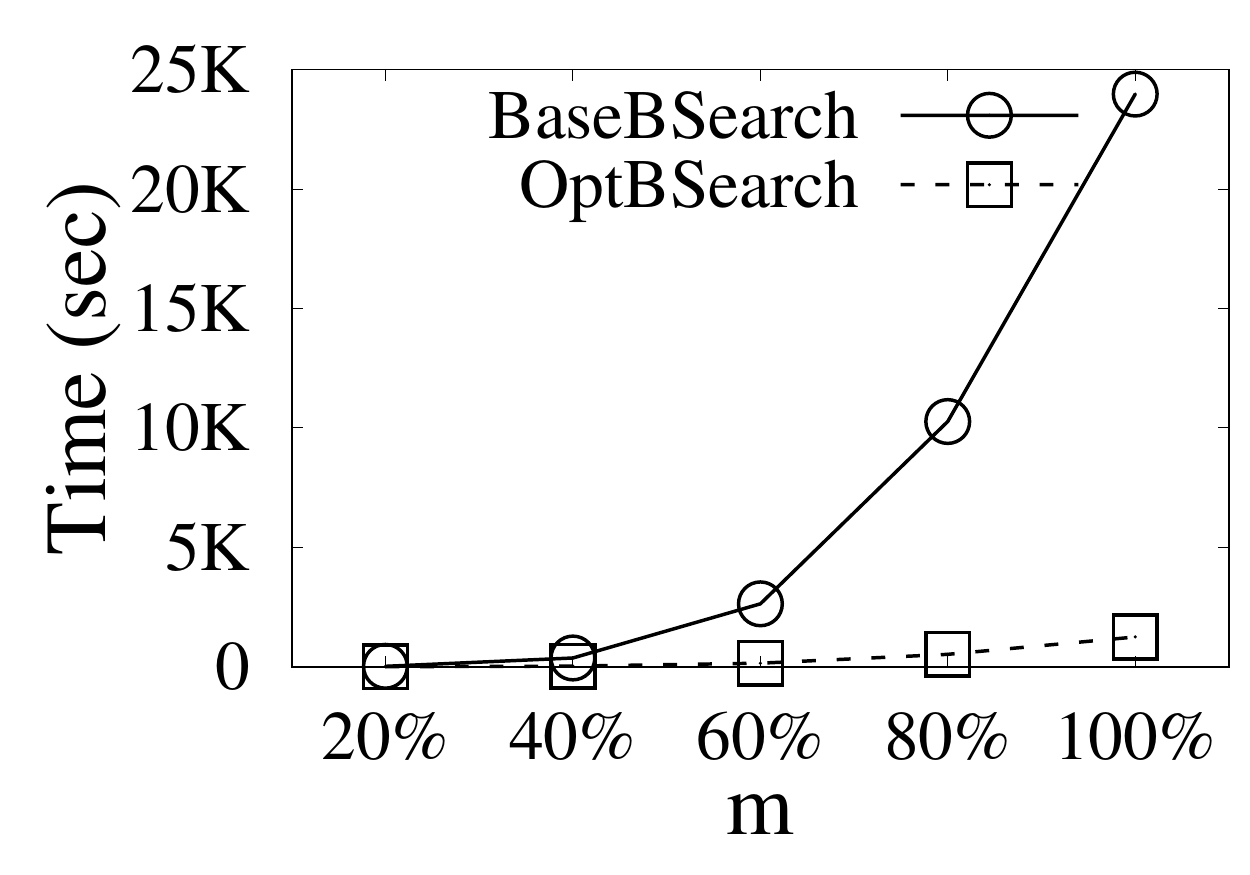}
  \end{minipage}
  }
  \subfigure[\livejournal (vary $n$)]{
  \label{fig:exp-scala-varyn-livejournal}
  \begin{minipage}{3.2cm}
  \centering
  \includegraphics[width=\textwidth]{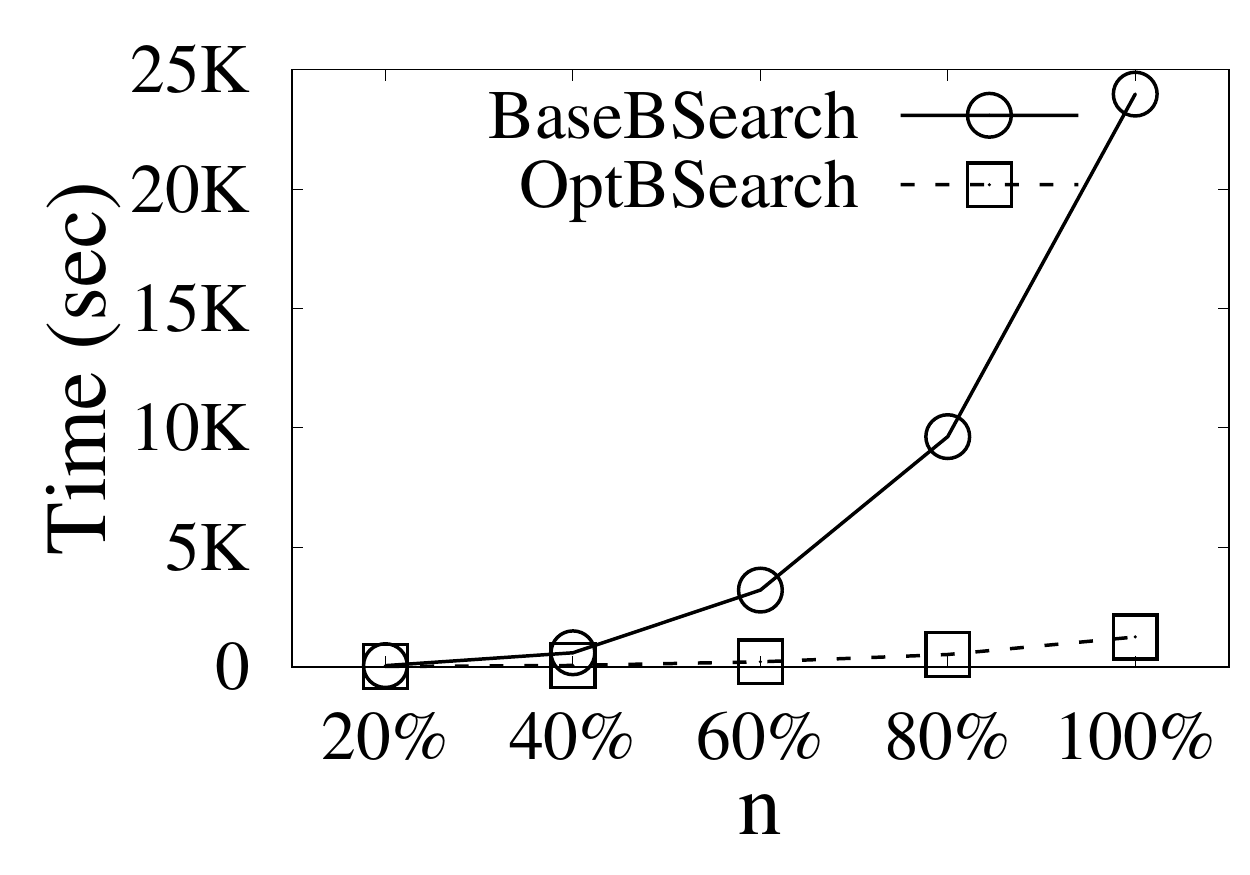}
  \end{minipage}
  }
  \vspace*{-0.2cm}
\caption{Scalability of \baseboundalg and \optboundalg}
\label{fig:exp-scala-varymn}\vspace*{-0.2cm}
\end{figure}

\stitle{Exp-2: The effect of $\theta$.} \figref{fig:exp-opt-vary-theta} reports the effect of parameter $\theta$ in \optboundalg on \wikitalk and \livejournal. The results on the other datasets are consistent. As can be seen, the runtime of \optboundalg varies slightly with different $\theta$ values. In general, \optboundalg performs slightly better with a relatively small $\theta$. For example, with $\theta= 1.05$, \optboundalg consumes the lowest runtime on both \wikitalk and \livejournal. Note that a large $\theta$ may increase the cost of computing the exact ego-betweennesses, while a small $\theta$ may increase the cost of updating the upper bounds in $H$. These results indicate that when $\theta= 1.05$, \optboundalg can achieve a good tradeoff between these two costs.

\stitle{Exp-3: Evaluation of the updating algorithms.} To evaluate the performance of our updating algorithms, we randomly select 1,000 edges for insertion and deletion on each dataset. \figref{fig:exp-update-alldata} shows the average runtime of \insertedge, \deleteedge, \insertedgetop and \deleteedgetop on all datasets. As expected, the update time of \insertedgetop is lower than that of \insertedge. For example, on \livejournal, the \insertedge consumes 0.578 seconds to maintain ego-betweennesses for all vertices, while \insertedgetop takes 0.304 seconds for updating the top-$k$ results. Similar results can also be observed for \deleteedge and \deleteedgetop. In addition, the average runtime of \insertedge (\insertedgetop) and \deleteedge (\deleteedgetop) is almost the same. Note that the runtime of all our updating algorithms is smaller than 0.7 seconds over all datasets. These results indicate that the proposed updating algorithms are very efficient on large real-life graphs.

\stitle{Exp-4: Scalability testing.} Here we evaluate the scalability of \baseboundalg and \optboundalg. To this end, we generate four subgraphs for each dataset by randomly picking 20\%-80\% of the edges (vertices), and evaluate the runtime of \baseboundalg and \optboundalg on these subgraphs. \figref{fig:exp-scala-varymn} illustrates the results on \livejournal. The results on the other datasets are similar. As can be seen, the runtime of \optboundalg increases very smoothly with increasing $m$ or $n$, while the runtime of \baseboundalg increases more sharply. Again,  we can see that \optboundalg is significantly faster than \baseboundalg with all parameter settings, which is consistent with our previous findings.

\begin{figure}[t!]\vspace*{-0.2cm}
\centering
  \subfigure[\livejournal]{
  \label{fig:exp-parallel-time-livejournal}
  \begin{minipage}{3.2cm}
  \centering
  \includegraphics[width=\textwidth]{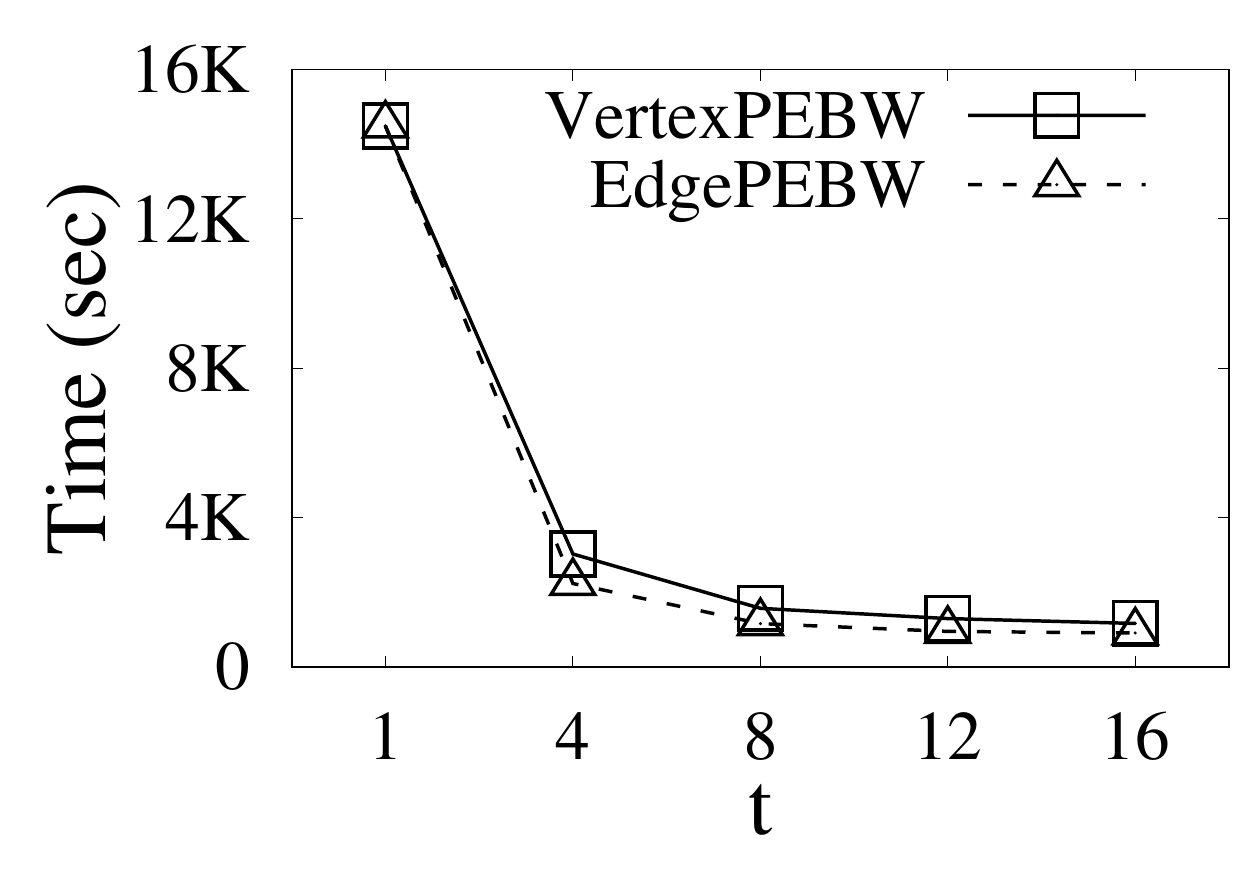}
  \end{minipage}
  }
  \subfigure[\livejournal]{
  \label{fig:exp-parallel-speedup-livejournal}
  \begin{minipage}{3.2cm}
  \centering
  \includegraphics[width=\textwidth]{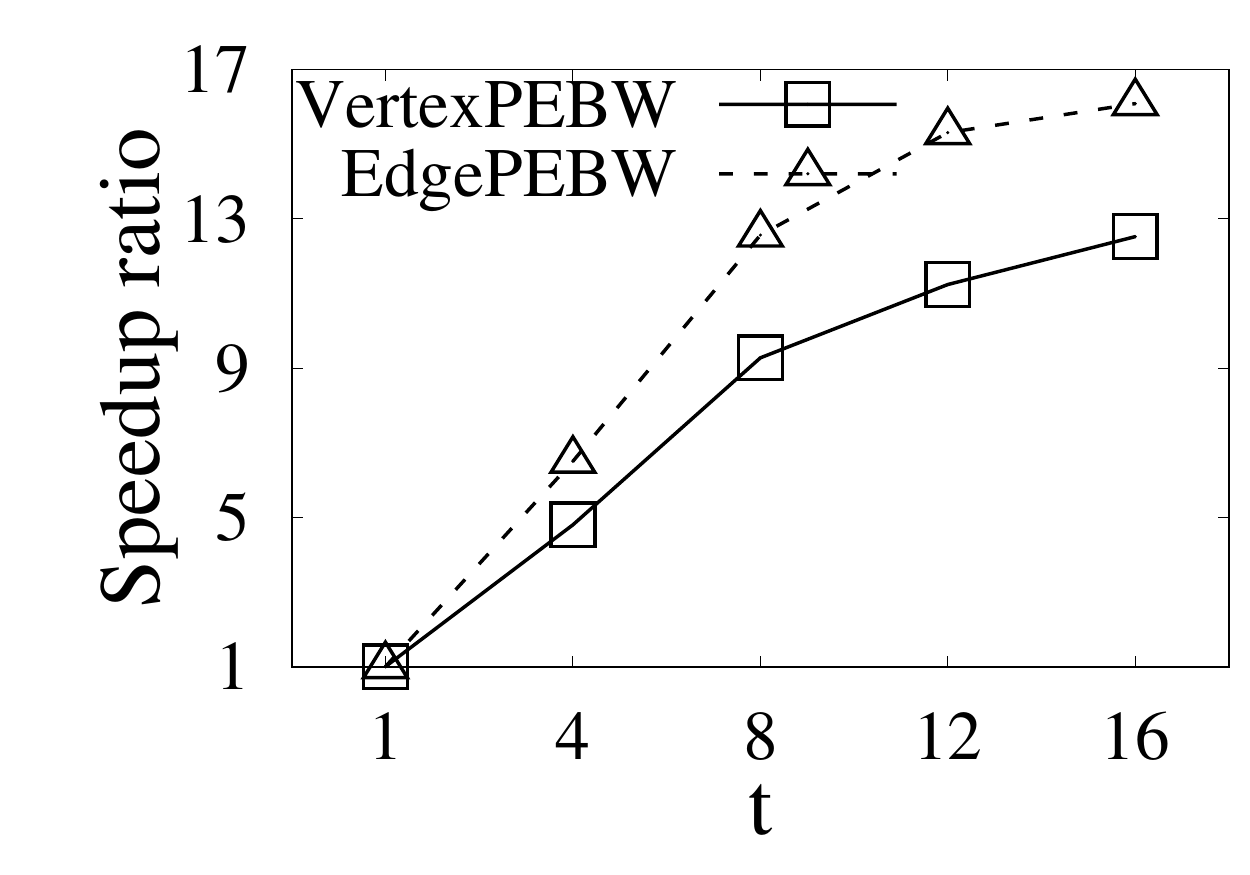}
  \end{minipage}
  }
  \vspace*{-0.2cm}
\caption{Evaluation of the parallel algorithms}
\label{fig:exp-parallel-time-vary-thread}\vspace*{-0.1cm}
\end{figure}

\begin{figure}[t!]\vspace*{-0.2cm}
\centering
\subfigure[{\wikitalk}]{
  \label{fig:exp-btebt-time-wikitalk}
  \begin{minipage}{3.2cm}
  \centering
  \includegraphics[width=\textwidth]{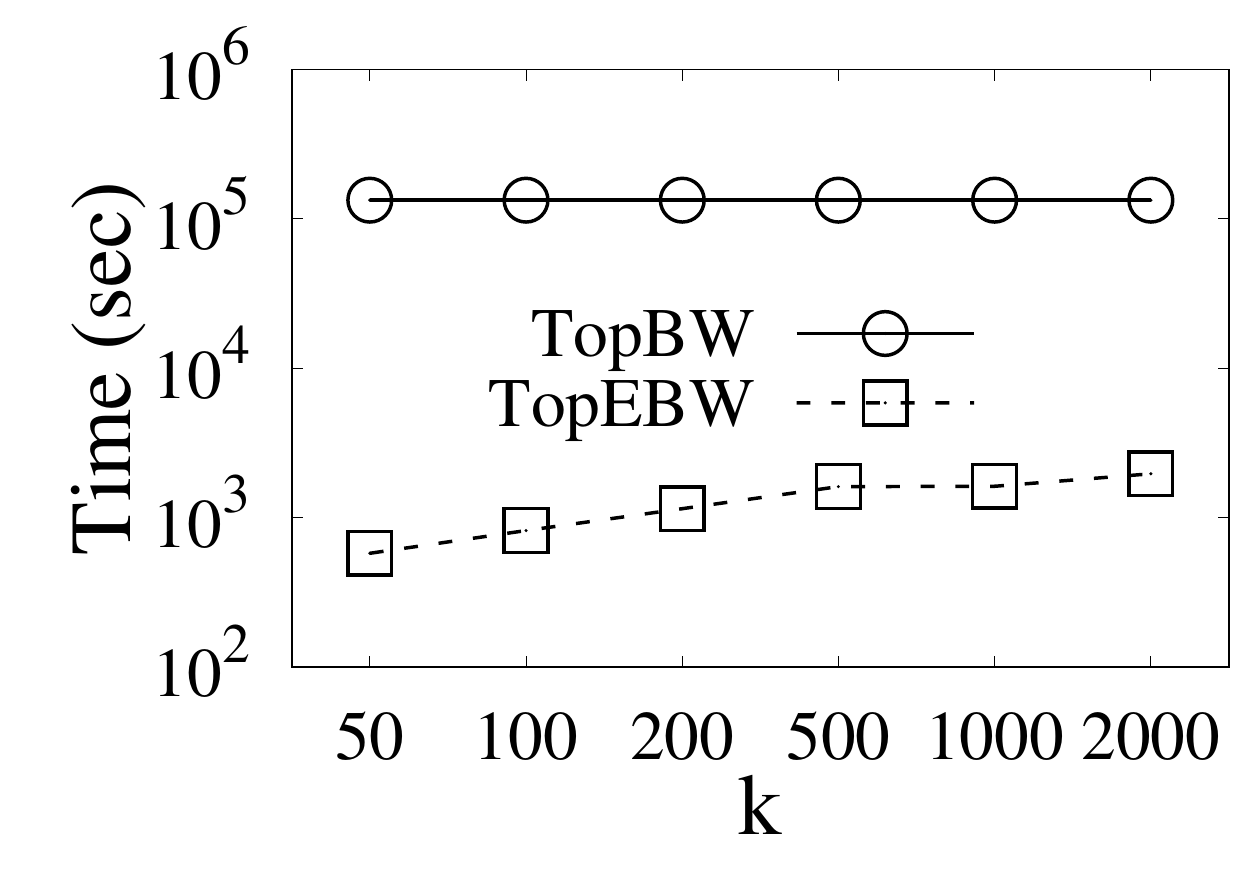}
  \end{minipage}
  }
  \subfigure[{\pokec}]{
  \label{fig:exp-btebt-time-pokec}
  \begin{minipage}{3.2cm}
  \centering
  \includegraphics[width=\textwidth]{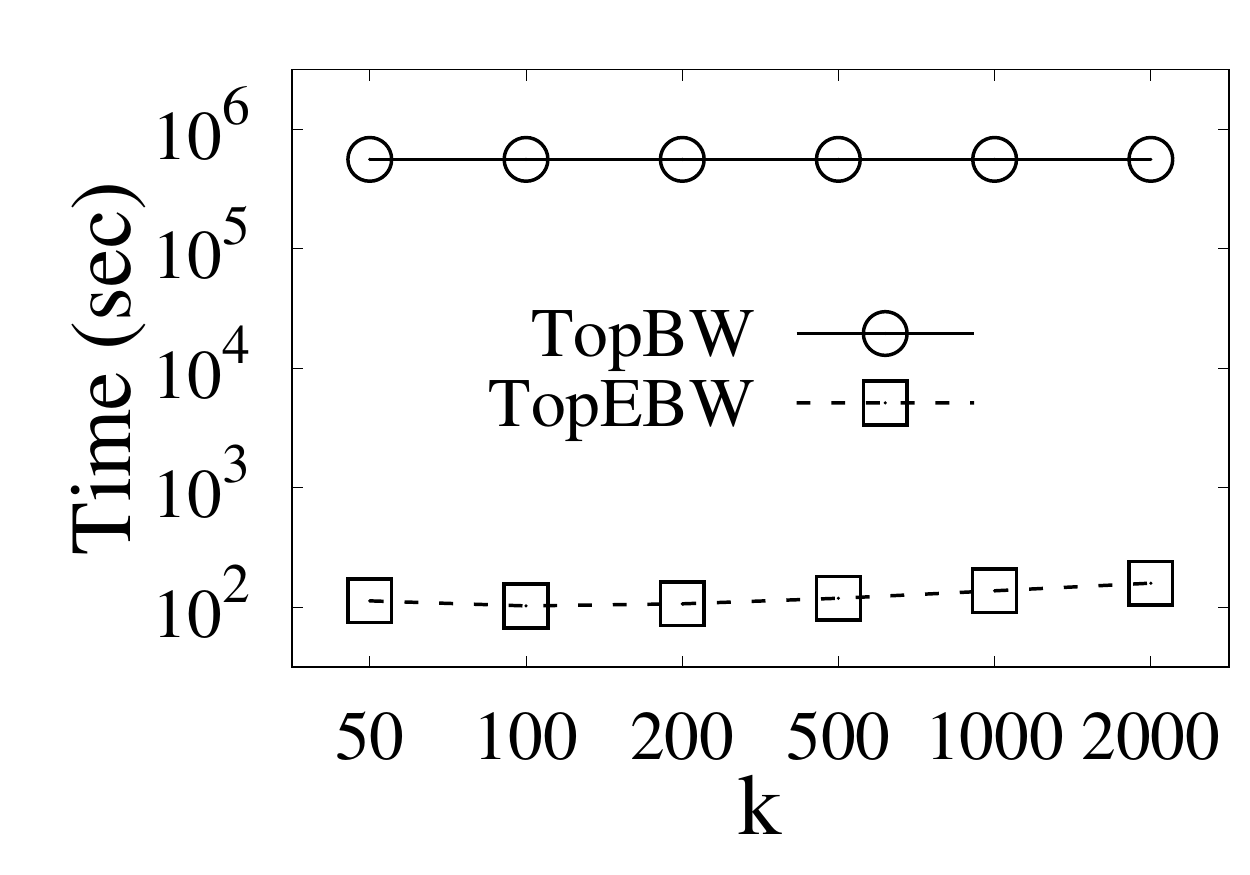}
  \end{minipage}
  }
  \vspace*{-0.3cm}

   \subfigure[{\wikitalk}]{
  \label{fig:exp-btebt-rpt-wikitalk}
  \begin{minipage}{3.2cm}
  \centering
  \includegraphics[width=\textwidth]{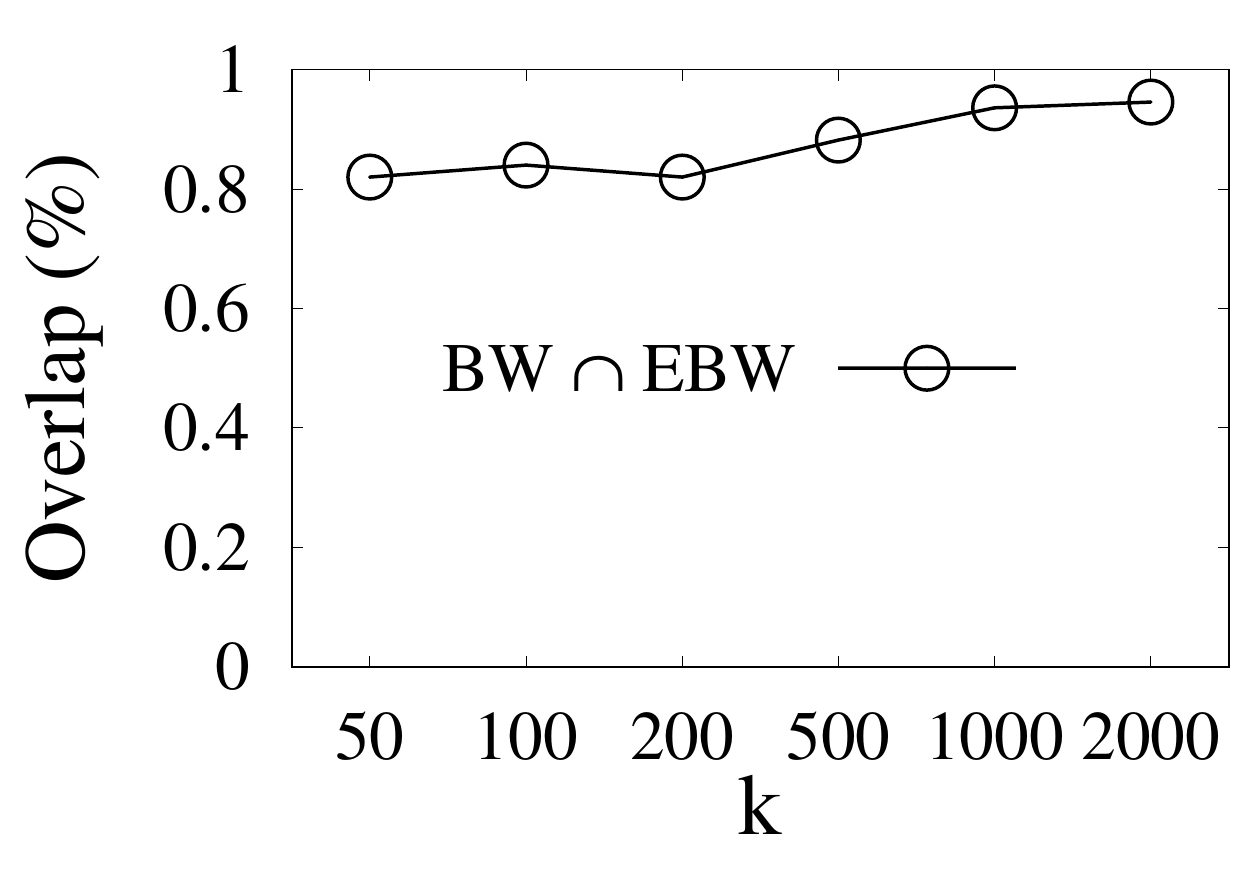}
  \end{minipage}
  }
  \subfigure[{\pokec}]{
  \label{fig:exp-btebt-rpt-pokec}
  \begin{minipage}{3.2cm}
  \centering
  \includegraphics[width=\textwidth]{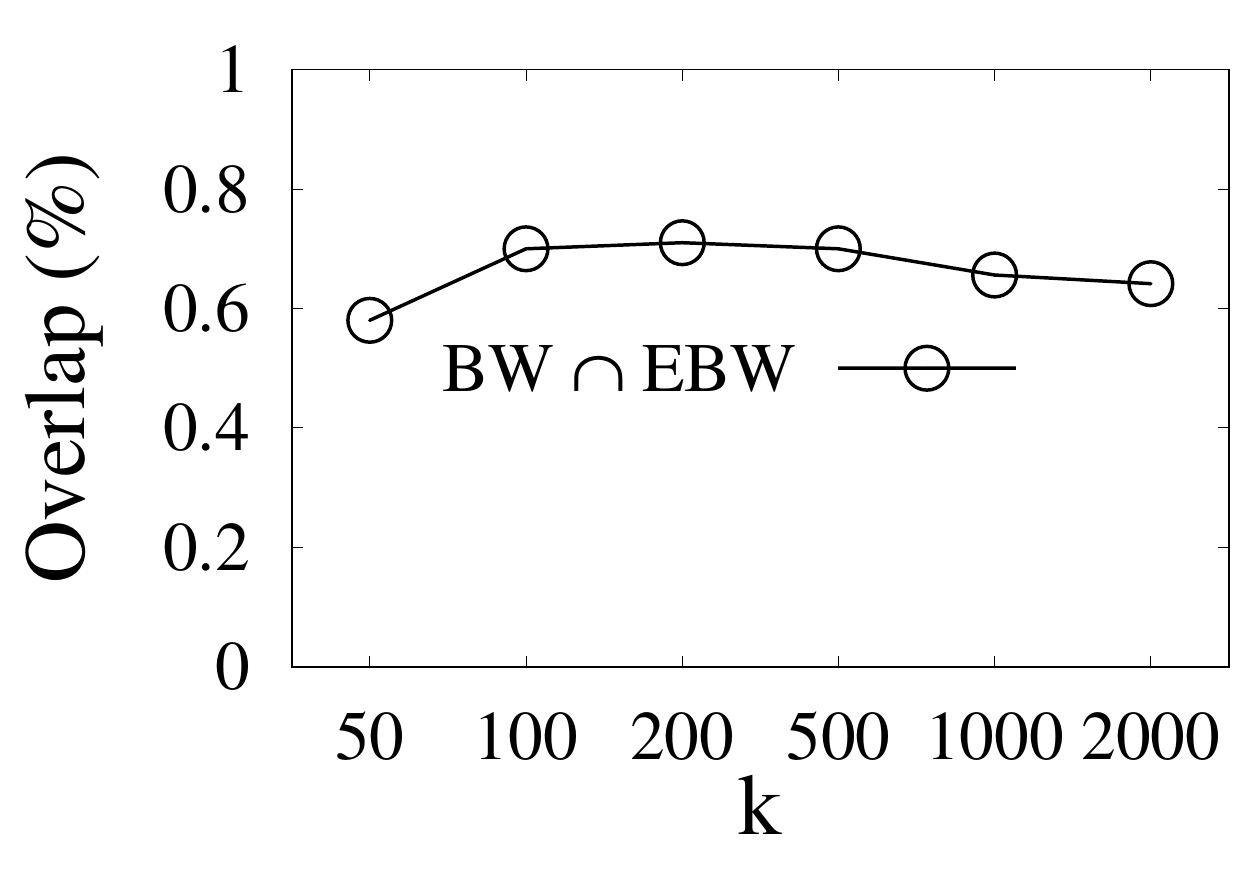}
  \end{minipage}
  }
  \vspace*{-0.2cm}
\caption{Comparison between \btalg and \ebtalg}
\label{fig:exp-btebt-alldata}\vspace*{-0.2cm}
\end{figure}

\stitle{Exp-5: Evaluation of parallel algorithms.} We vary the number of threads $t$ from 1 to 16, and evaluate two parallel algorithms, i.e., \vertexparaalg and \edgeparaalg, with an increasing $t$. We run \optboundalg with the parameter $k=n$ to compute ego-betweennesses as baseline for $t=1$. \figref{fig:exp-parallel-time-vary-thread} shows the results of runtime and speedup ratio on \livejournal. From \figref{fig:exp-parallel-time-vary-thread}, we can see that both \vertexparaalg and \edgeparaalg achieve very good speedup ratios. The runtime of \edgeparaalg is lower than \vertexparaalg with all parameter settings. For example, the running time of \optboundalg to calculate ego-betweennesses for all vertices is 14,487.840 seconds. When $t = 16$, \vertexparaalg takes 1,156.916 seconds and \edgeparaalg consumes 900.439 seconds to compute the results. The speedup ratios of \vertexparaalg and \edgeparaalg are roughly equal to 12 and 16, respectively. These results indicate that our parallel algorithms are very efficient on real-life graphs.

\subsection{Effectiveness testing} \label{subsec:effectiveness}

In this experiment, we evaluate the effectiveness of the proposed algorithms. For comparison, we make use of the state-of-the-art Brandes' algorithm \cite{brandes2001faster} to compute betweenness for each vertex and then identify the top-$k$ vertices with the highest betweennesses. We refer to this baseline algorithm as \btalg and our \optboundalg as \ebtalg for brevity. The top-$k$ results obtained by \btalg and \ebtalg are denoted as \kw{BW} and \kw{EBW} respectively.

\begin{figure}[t!]\vspace*{-0.2cm}
\centering
  \subfigure[\db]{
  \label{fig:exp-btebt-time-db}
  \begin{minipage}{3.2cm}
  \centering
  \includegraphics[width=\textwidth]{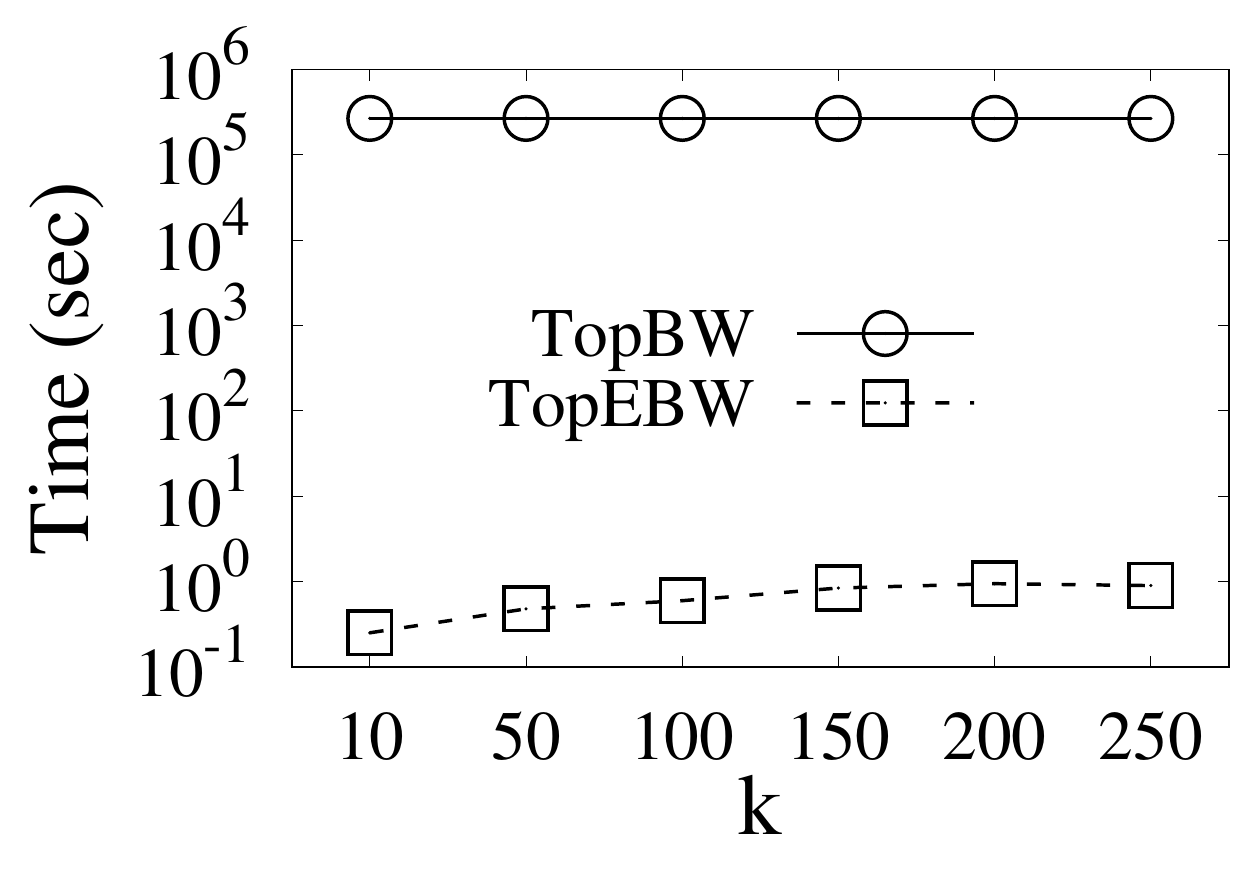}
  \end{minipage}
  }
  \subfigure[\ir]{
  \label{fig:exp-btebt-time-ir}
  \begin{minipage}{3.2cm}
  \centering
  \includegraphics[width=\textwidth]{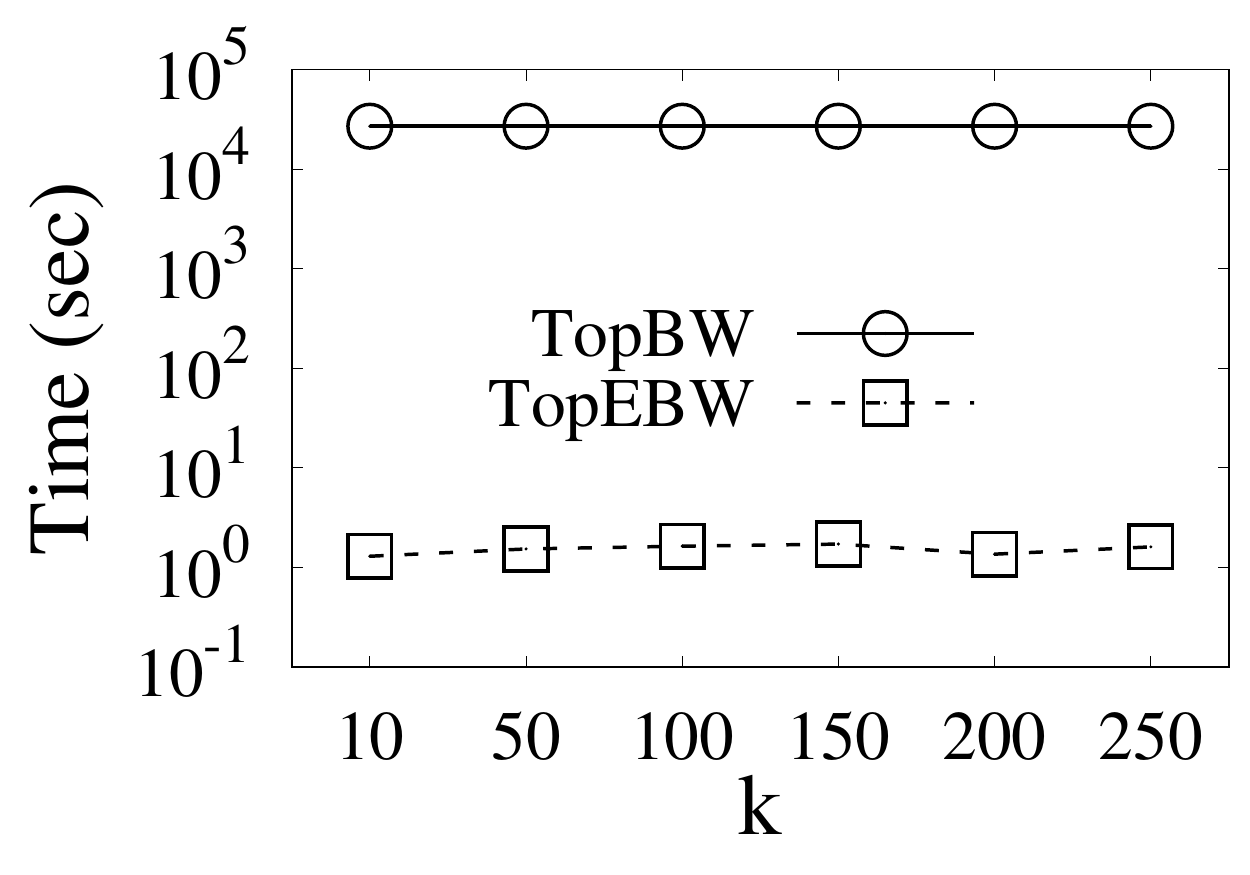}
  \end{minipage}
  }
  \vspace*{-0.3cm}

  \subfigure[\db]{
  \label{fig:exp-btebt-rpt-db}
  \begin{minipage}{3.2cm}
  \centering
  \includegraphics[width=\textwidth]{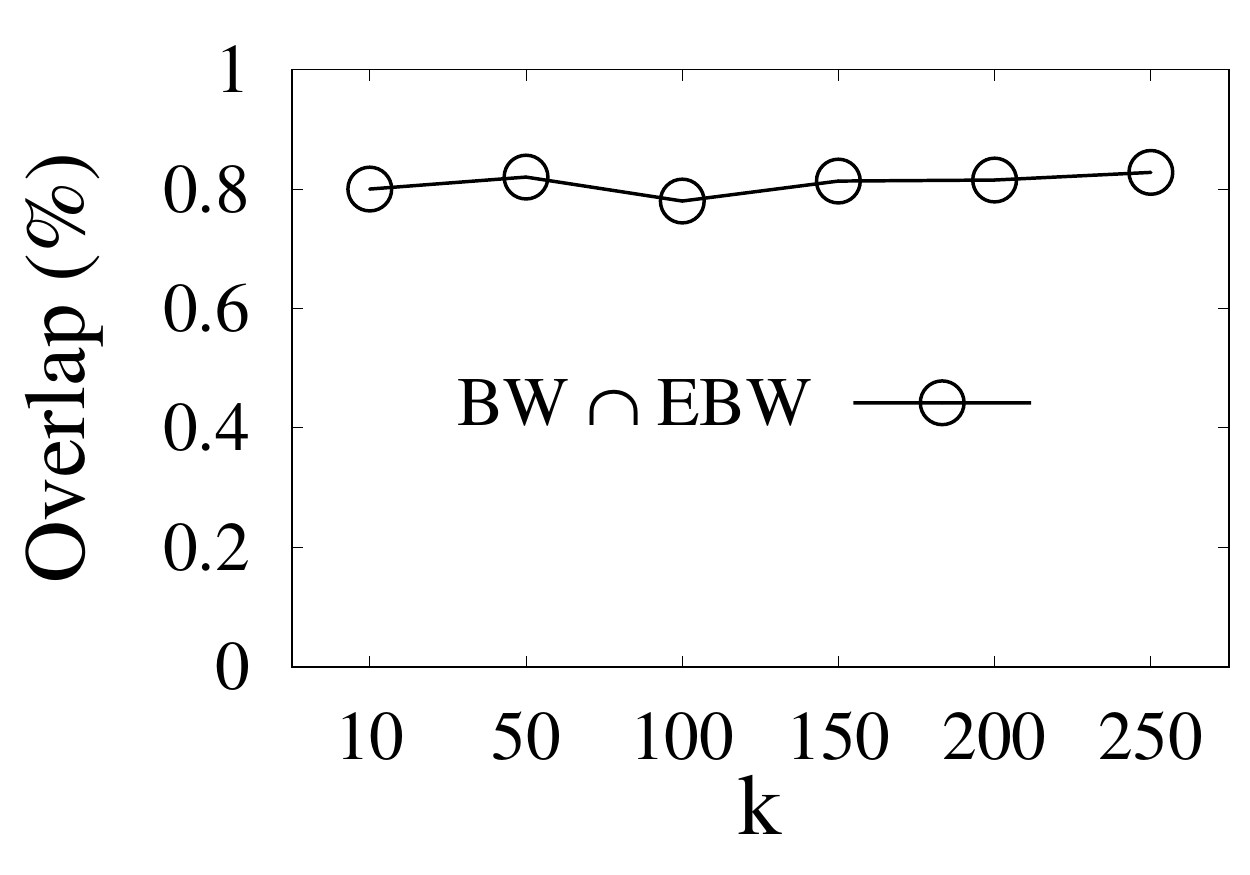}
  \end{minipage}
  }
  \subfigure[\ir]{
  \label{fig:exp-btebt-rpt-ir}
  \begin{minipage}{3.2cm}
  \centering
  \includegraphics[width=\textwidth]{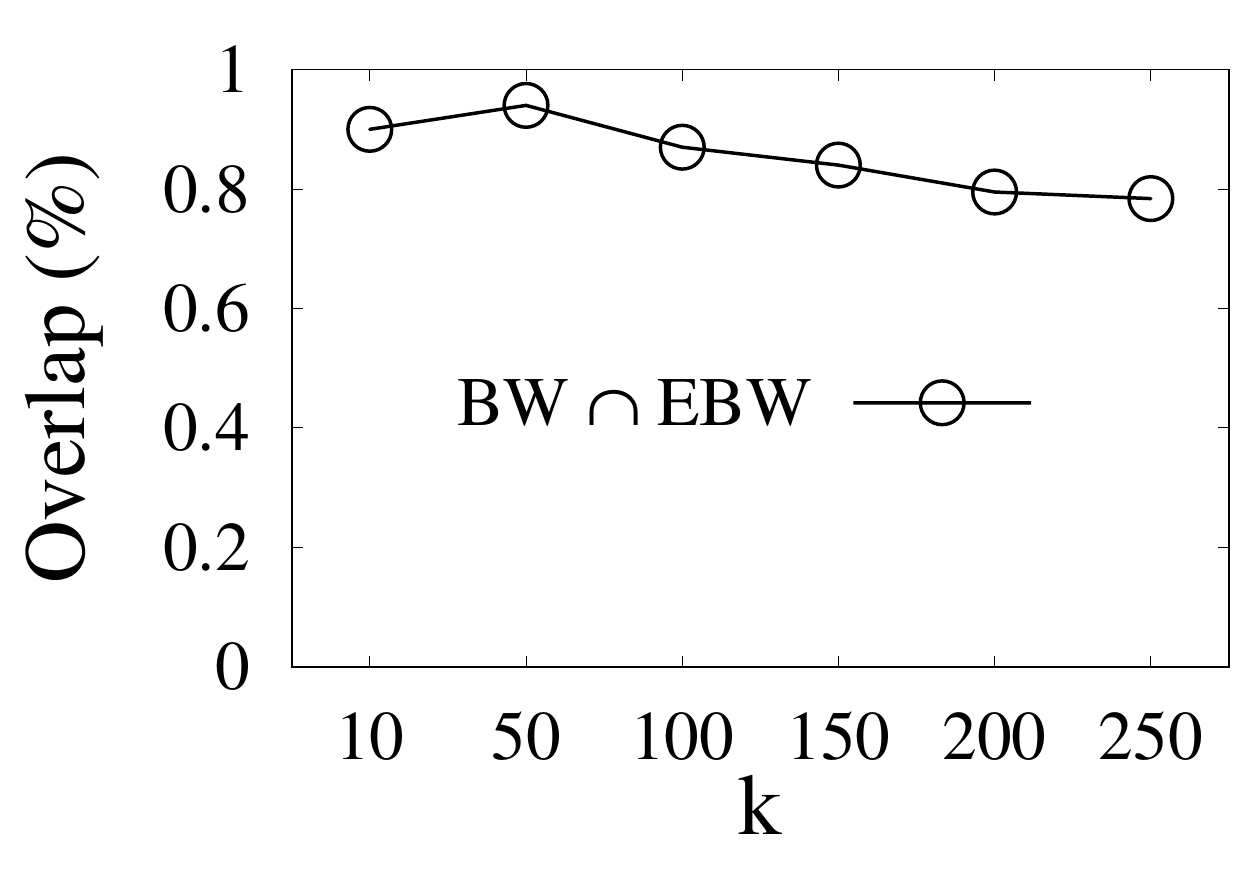}
  \end{minipage}
  }
  \vspace*{-0.2cm}
\caption{Comparison between \btalg and \ebtalg on \dblp}
\label{fig:exp-btebt-dblp}\vspace*{-0.2cm}
\end{figure}

\stitle{Exp-6: Comparison between \btalg and \ebtalg.} We compare \btalg and \ebtalg on \wikitalk and \pokec with $k \in \{50, 100, 200, 500, 1000, 2000\}$. The results on the other datasets are consistent. Note that to speed up the betweenness computation, we also implement a parallel version of \btalg for comparison. The running time of \btalg with 64 threads and \ebtalg is shown in \figref{fig:exp-btebt-alldata}(a-b). Clearly, \ebtalg is at least two orders of magnitude faster than the parallel \btalg within all parameter settings. For example, on \pokec, \ebtalg takes 112.369 seconds, while \btalg consumes 559,322.062 seconds to output the top-$50$ results.

\figref{fig:exp-btebt-alldata}(c-d) report the overlap of the top-$k$ results obtained by \btalg and \ebtalg. As can be seen, the overlap is generally higher than 60\% on all datasets. Particularly, on \wikitalk, the overlap is even more than 80\%. These results indicate that the ego-betweenness centrality is a very good approximation of the betweenness centrality. Moreover, compared to betweenness centrality, the ego-betweenness centrality is much cheaper to compute using the proposed algorithms.

\begin{table}[t!]\vspace*{-0.2cm}
    \scriptsize
	\centering
	\caption{Top-10 scholars in \db}
    \label{tab:topdbres}
	\vspace*{-0.2cm}
    \resizebox{\linewidth}{!}{
		\begin{tabular}{c|c|c|c|c|c}
			\hline
			{}\rule{0pt}{6pt}Top-10 EBW      & $d$       & $C_B$          & Top-10 BW     &   $d$ & $B_T$\\ \hline \hline
			{\bf{*Jiawei Han}}     &412	 &73,928.5	    & {\bf{*Philip S. Yu}}     &360	 &50,320,100\\
            {\bf{*Philip S. Yu}}   &360	 &58,834.1	    & {\bf{*Jiawei Han}}  &412	 &50,059,900\\
            {\bf{*Christos Faloutsos}}        &337	&52,192.9 & {\bf{*Christos Faloutsos}}        &337	&46,340,200\\
            {\bf{*Jian Pei}}       &215  &20,531.1  & {\bf{*Gerhard Weikum}}       &213  &26,232,700\\
            {\bf{*Gerhard Weikum}} &213	 &19,238.3	& {\bf{*Beng Chin Ooi}} &205	 &22,376,200\\
			{\bf{*Michael J. Franklin}}		&220 &17,867.5 & {\bf{*Jian Pei}}		&215 &21,470,900\\
            Michael Stonebraker		&210   &16,081.4  & {\bf{*Michael J. Franklin}}		&220  &20,809,000\\
            {\bf{*Raghu Ramakrishnan}}      &210		&15,930.1 & {\bf{*Raghu Ramakrishnan}}      &210		&18,481,900\\
            {\bf{*Beng Chin Ooi}}	 &205   	&14,848.2    & Haixun Wang	 &183   	&17,062,500\\
            Hector Garcia-Molina &197 	&14,664.8 & H. V. Jagadish &178 	&16,144,700\\
            \hline
		\end{tabular}
	}
\vspace*{-0.2cm}
\end{table}

\begin{table}[t!]
    \scriptsize
	\centering
	\caption{Top-10 scholars in \ir}
    \label{tab:topirres}
	\vspace*{-0.2cm}
    \resizebox{\linewidth}{!}{
		\begin{tabular}{c|c|c|c|c|c}
			\hline
			{}\rule{0pt}{6pt}Top-10 EBW     & $d$       & $C_B$           & Top-10 BW    &   $d$ & $B_T$\\ \hline \hline
			{\bf{*Jeffrey P. Bigham}}     &2441	 &1.4846e+06    &{\bf{*Taesup Moon}}     &2318	 &1.33948e+07\\
            {\bf{*Alex D. Wade}}   &2510	 &1.46767e+06	    & {\bf{*Jeffrey P. Bigham}}  &2441	 &1.1711e+07\\
            {\bf{*Adam Sadilek}}        &1993	&1.30844e+06 & {\bf{*Alex D. Wade}}        &2510	&1.10161e+07\\
            {\bf{*Taesup Moon}}       &2318  &1.25722e+06  & {\bf{*Adam Sadilek}}       &1993  &9.49158e+06\\
            {\bf{*Antonio Gulli}} &1951	 &1.16136e+06	& {\bf{*Antonio Gulli}} &1951	 &9.44098e+06\\
            {\bf{*Henry A. Kautz}}	&1731 &882,981  & {\bf{*Bob Boynton}}		&1618 &7.00364e+06\\
            {\bf{*Bob Boynton}}		&1618   &844,761   & {\bf{*Henry A. Kautz}}		&1731  &6.82747e+06\\
            {\bf{*Padmini Srinivasan}}      &1541		&822,131 & Linchuan Xu      &1834		&6.79258e+06\\
            {\bf{*Yelena Mejova}}	 &1210   	&580,116    & {\bf{*Padmini Srinivasan}}	 &1541   	&6.41121e+06\\
            Raymie Stata &796 	&224,422 &   {\bf{*Yelena Mejova}} &1210 	&5.82391e+06\\
			\hline
		\end{tabular}
	}
\end{table}

\stitle{Exp-7: Case study on \dblp.} We extract two subgraphs, namely, \db and \ir, from \dblp for case study. \db contains the authors in \dblp who had published at least one paper in the database and data mining related conferences. The \db subgraph contains 37,177 vertices and 131,715 edges. The \ir subgraph contains the authors who had published at least one paper in the information retrieval related conferences with 13,445 vertices and 37,428 edges. We invoke \btalg (\ebtalg) to find the top-$k$ highest (ego-)betweennesses scholars on \db and \ir with the parameter $k \in \{10, 50, 100, 150, 200, 250\}$. The results are shown in \figref{fig:exp-btebt-dblp}. Consistent with the previous findings, the running time of \ebtalg is significantly faster than \btalg. Moreover, the overlap of the top-$k$ results is significantly high. For example, on \db, \ebtalg takes 22.777 seconds, while \btalg consumes 27641.190 seconds to output the top-$100$ results. The overlap of the top-$100$ results on \db is 78\%. Similar results can also be observed on \ir.

We also illustrate the top-$10$ scholars on \db and \ir in \tabref{tab:topdbres} and \tabref{tab:topirres}. In both \tabref{tab:topdbres} and \tabref{tab:topirres}, $d$ denotes the number of co-authors of a scholar; $C_B$ and $B_T$ denote the ego-betweenness and betweenness of a scholar respectively. Clearly, the overlaps of the top-$10$ results are 80\% and 90\% on \db and \ir respectively. Moreover, we can see that the top-10 scholars with the highest ego-betweennesses are the most influential in the database, data mining, and information retrieval communities. Such scholars may play a bridge role in connecting different research groups. For example, in \tabref{tab:topdbres}, Professor Jiawei Han has 412 co-authors and maintains connections with many different research groups. Similarly, in \tabref{tab:topirres}, Taesup Moon is interested in diverse areas such as information retrieval, statistical machine learning, information theory, signal processing and so on, thus he plays an important role in promoting the interactions between different research communities. These results indicate that our algorithms can be used to find high influential vertices in a network that act as network bridges.

\section{Related work} \label{sec:relatedwork}
\stitle{Betweenness centrality.} Our work is closely related to betweenness centrality \cite{brandes2001faster, freeman1978centrality}. Betweenness centrality is an important measure of centrality in a graph based on the shortest path, which has been applied to a wide range of applications in social networks \cite{DBLP:conf/semco/Ostrowski15}, biological networks \cite{jeong2001lethalitybio}, computer networks \cite{baldesi2017usetrans}, road networks \cite{eigencent2008newman} and so on. The best-known algorithm for betweenness computation, proposed by Brandes \cite{brandes2001faster}, runs ${\mathcal O}(nm)$ time complexity for unweighted networks. Measuring the betweenness centrality scores of all vertices is notoriously expensive, thus many parallel and approximate algorithms have been developed to reduce the computation cost \cite{DBLP:journals/peerj-cs/Fan0Z17, DBLP:journals/snam/BeheraNRR20, DBLP:conf/cikm/Chehreghani13, DBLP:conf/complexnetworks/FurnoFSZ17, DBLP:conf/infocom/CrescenziFP20}. Fan \etal proposed an efficient parallel GPU-based algorithm for computing betweenness centrality in large weighted networks and integrated the work-efficient strategy to address the load-imbalance problem \cite{DBLP:journals/peerj-cs/Fan0Z17}. Furno \etal studied the performance of a parametric two-level clustering algorithm for computing approximate value of betweenness with an ideal speedup with respect to Brandes' algorithm \cite{DBLP:conf/complexnetworks/FurnoFSZ17}. In this paper, we focus on the ego-betweenness centrality which is first proposed by Everett \etal \cite{everett2005ego} as an approximation of betweenness centrality. Ego-betweenness centrality has gained recognition in its own right as a natural measure of a node’s importance as a network bridge \cite{marsden2002egocentric}. To the best of our knowledge, our work is the first to study the problem of finding top-$k$ ego-betweenness vertices in graphs.

\stitle{Top-$k$ retrieval.} Our work is also related to the top-$k$ retrieval problem, which aims to find $k$ results with the largest scores/relevances based on a pre-defined ranking function \cite{ilyas2008survey}. The general framework for answering top-$k$ queries is to process the candidates according to a heuristic order and prune the search space based on some carefully-designed upper bounds. An excellent survey can be found in \cite{ilyas2008survey}. There are many studies on top-$k$ query processing for heterogeneous applications, such as processing distributed preference queries \cite{chang2002minimal}, keyword queries \cite{luo2007spark}, set similarity join queries \cite{xiao2009top}. An influential algorithm was proposed by Fagin \etal \cite{fagin1999combining, fagin2003optimal}, which considers both random access and/or sequential access of the ranked lists. Recently, some studies take diversity into consideration in the top-$k$ retrieval in order to return diversified ranking results \cite{qin2012diversifying, li2013scalable, angel2011efficient, zhu2011unified, agrawal2009diversifying}. For instance, Li \etal proposed a scalable algorithm to achieve near-optimal top-$k$ diversified ranking with linear time and space complexity with respect to the graph size. Some studies have also been done which focus mainly on exploring influential communities, individuals, and relationships in different networks \cite{wang2010community, bi2018optimal, 15vldbjstrucdiv, chang2017scalable, zhang2020efficient}. For example, the study \cite{bi2018optimal} investigated an instance-optimal algorithm, which runs in linear time complexity without indexes, for computing the top-$k$ influential communities. In this paper, we develop two search frameworks to identify the top-$k$ vertices with the highest ego-betweennesses and propose efficient techniques to maintain top-$k$ results when the graph is updated.

\section{Conclusion} \label{sec:conclusion}
In this paper, we study a problem of finding the top-$k$ vertices in a graph with the highest ego-betweennesses. To solve this problem, we first develop two top-$k$ search frameworks with a static upper bound and a novel dynamic upper bound, respectively. Then, we propose efficient local maintenance algorithms to maintain the ego-betweenness for each vertex when the graph is updated. We also present lazy-update techniques to maintain the top-$k$ results in dynamic graphs. We conduct extensive experiments using five real-life datasets to evaluate the proposed algorithms. The results demonstrate the efficiency and scalability of our algorithms. Also, the results show that the top-$k$ ego-betweenness results are highly similar to the top-$k$ betweenness results, but they are much cheaper to compute by our algorithms.

\bibliographystyle{unsrt}
\bibliography{egobw-arxiv}

\end{document}